\documentclass[11pt,draftclsnofoot,onecolumn]{IEEEtran}
\usepackage[noadjust]{cite}

\usepackage[pdftex]{graphicx}
\usepackage{epstopdf}

\ifCLASSINFOpdf
\else
\fi
\usepackage[cmex10]{amsmath}
\usepackage{amsfonts}
\usepackage{amssymb,amsthm}
\usepackage{extarrows}
\usepackage{array}

\usepackage{enumerate}
\usepackage{mathtools,empheq}

\usepackage{eqparbox}

\usepackage[font=footnotesize]{caption}

\usepackage[caption=false,font=footnotesize,labelfont=bf]{subfig}
\usepackage{setspace, color, bbm}

\DeclareMathOperator*{\argmin}{arg\,min}

\def\mathbi#1{\textbf{\em #1}}
\newtheorem{prop}{\bf{Proposition}}
\newtheorem{lem}{\bf {Lemma}}

\theoremstyle{definition}
\newtheorem{defn}{\bf Definition}

\theoremstyle{remark}
\newtheorem{remk}{\bf Remark}

\begin{document}
%
\title{Joint Filter and Waveform Design for Radar STAP in Signal Dependent Interference}
%
%

%

\author{Pawan Setlur, \emph{Member, IEEE}, Muralidhar Rangaswamy, \emph{Fellow, IEEE}
\thanks{P. Setlur is affiliated with the Wright State Research Inst., and as a research contractor with the US AFRL, WPAFB, OH, email:pawan.setlur.ctr@wpafb.af.mil.}
\thanks{M. Rangaswamy is with Sensors Directorate, U.S. AFRL, WPAFB, OH,
email:muralidhar.rangaswamy@wpafb.af.mil.} 
\thanks{{Approved for Public Release No.: 88ABW-2014-3392}.} }%

%
%

\markboth{\tiny{P. Setlur and M. Rangaswamy: AFRL Sensors Directorate Tech. Report, 2014.}}%
{AFRL, Sensors Directorate Tech. Report, 2014: Approved for Public release}
%



\newcommand{\nrd}[1]{\textcolor{red}{#1}}
\newcommand{\pas}[1]{\textcolor{magenta}{#1}}
\newcommand{\nbl}[1]{\textcolor{blue}{#1}}

\maketitle

\begin{abstract}
Waveform design is a pivotal component of {\it the fully adaptive radar} construct. In this paper we consider  waveform design for radar space time adaptive processing (STAP), accounting for the waveform dependence of the clutter correlation matrix. Due to this dependence,  in general, the joint problem of receiver filter optimization and radar waveform design becomes an intractable, non-convex optimization problem, Nevertheless, it is however shown to be individually convex either in the filter or in the waveform variables. We derive constrained versions of: a) the alternating minimization algorithm, b) proximal alternating minimization, and c) the constant modulus  alternating minimization, which, at each step, iteratively optimizes either the STAP filter or the waveform independently. A fast and slow time model permits waveform design in radar STAP but the primary bottleneck is the computational complexity of the algorithms.

\end{abstract}


%
\IEEEpeerreviewmaketitle

\section{Introduction}
%
%
%


\IEEEPARstart{T}{he} objective of this report is to address waveform design in radar space time adaptive processing (STAP) \cite{klemm2002,ward1994,guerci2003,Brennan1973}. An air-borne radar is assumed with an array of sensor elements observing a moving target on the ground. We will assume that the waveform design and scheduling are performed over one CPI rather than on an individual pulse repetition interval (PRI).To facilitate waveform design, we develop a STAP model considering the fast time samples along with the slow time processing. This is different from traditional STAP which generally considers the data after matched filtering \cite{klemm2002,ward1994}. Nonetheless STAP research efforts have been proposed which consider inclusion of fast time samples in space time processing, see for example \cite{klemm2002,madurasinghe2006,seliktar2006} and references therein.

In line with traditional STAP, we formulate the waveform design, as an minimum variance distortion-less response (MVDR) type optimization \cite{capon1969}. As we will see in the sequel, inclusion of the waveform increases the dimensionality of the  correlation matrix. Classical Radar STAP is computationally expensive but the waveform adaptive STAP increases the complexity by several orders of magnitude. Therefore, benefits of waveform design in STAP come at the expense of increased computational complexity. The noise, clutter, and interference are modeled stochastically and are assumed to be mutually uncorrelated. Endemic to airborne STAP, clutter is persistent in most range gates resulting from ground reflections. The clutter correlation matrix is a function of the waveform causing the joint reliever filter and waveform optimization to be non-convex with no closed form solution. However, it is analytically shown here that the STAP MVDR objective is convex with respect to (w.r.t.) the receiver filter for a fixed but arbitrary waveform, and vice versa. Therefore, alternating minimization approaches arise as natural candidate solutions. As such, alternating minimization itself has a rich history in the optimization literature, possibly motivated directly from the works in \cite{Powell1964,Powell1973,Zangwill1967,Ortega1970}, with some not so recent seminal contributions \cite{Luo1992, Auslender1992, Bertsekas1999,Grippo2000} and recent contributions (not exhaustive) \cite{Attouch2010,Beck2013}. Other celebrated algorithms such as the Arimoto-Blahut algorithm to calculate channel capacity, and the expectation-maximization (EM) algorithm are all examples of the alternating minimization.

Here we address the joint optimization problem via a constrained alternating minimization approach, which has the favorable property of monotonicity in successive objective evaluations. Convergence, performance guarantees and other properties pertinent to this algorithm are further addressed. Full rank correlation matrices are required in implementing the constrained alternating minimization approach. In practice, radar STAP contends with rank deficient correlation matrices due to lack of homogeneous training data. In this case, the constrained alternating minimization approach is not implementable. To addresses this issue, we consider regularization of the STAP objective via strongly convex functions resulting in the constrained proximal alternating minimization \cite{Setlurasilomar2014}. Proximal algorithms, originally proposed by \cite{Martinet1970,Rockafeller1973} are well suited candidate techniques for constrained, large scale optimization
\cite{Attouch2010,Bertsekas1994,Rockafeller1976,Combettes2011,
Parikh2013}, applicable readily to our waveform adaptive STAP problem. In fact, as we will see subsequently the constrained proximal alternating minimization results in diagonal loading solutions, and for optimization-specific interpretations, the  load factors may be related to the Lipschitz constants (w.r.t. the gradient).

{\bf Signal dependent interference: Chicken or the Egg?} The fundamental problem in practical radar waveform design  is analogous to the chicken or the egg problem. Signal dependent interference, i.e., clutter, can only be perfectly characterized by transmitting a signal. Herein lies the central problem. The estimated clutter properties could therefore be dependent on what was transmitted in the first place. This is especially true for frequency selective and dispersive clutter responses frequently encountered in radar operations, for example, urban terrain. Therefore, any claim of optimality is myopic. Sadly the same problem would also persist when the target impulse responses are used to shape the waveform. Unfortunately, and  as famously stated by Woodward \cite{woodward1952information,
woodward1953theory}, ``\ldots what to transmit remains substantially unanswered''  \cite{woodward1953probability,Benedetto2009}. 

We will assume like other works in the signal dependent interference waveform design
\cite{Delong1967,Delong1969,Delong1970,Kay2007,vaidyanathan2009,
Pillai2000,Demaio2012,Demaio2013,Hongbin2014}, that the clutter response is known {\it a priori}. To a certain extent, this may be obtained via a combination of, either previous radar transmission \cite{cochran2009waveform}, or  assuming that the topography is known from ground elevation maps, synthetic aperture radar imagery \cite{SetlurJSTSP2014}, or access to knowledge aided databases as in the DARPA's KASSPER program \cite{Guerci2006}. 

{\bf Literature:}
The signal dependent interference waveform design problem has had a rich history \cite{middleton1996,van2004detection}. Iterative approaches but not limited alternating minimization type techniques have been the subject of work in \cite{Delong1967,Delong1969,Delong1970,Kay2007,vaidyanathan2009,
Pillai2000,Stoica2012,Demaio2012,Demaio2013,Hongbin2014} for SISO, MIMO radars but never in radar STAP. Waveform design for STAP without considering the signal dependent interference clutter was addressed in \cite{pattonstap2012}, where the authors premise is that the degrees of freedom from the waveform could be used in suppressing the interference and noise, while the degrees of freedom from the filter could be used  exclusively for suppressing the clutter. A joint STAP waveform and STAP filter design was never considered. Further, their premise is erroneous for the following several reasons. For any radar application, but especially in STAP, obtaining range cells which are interference free or clutter free is impossible. Nonetheless assuming this was possible, then, the weight vector for exclusive clutter suppression uses {\it the  inverse of the clutter interference correlation matrix} only, and not, as stated in \cite{pattonstap2012}, the inverse of the (clutter+noise+interference) correlation matrix. Furthermore such a detector may have disastrous consequences,  because control in the false alarm rates becomes impossible due to the self induced coloring on other range cells which are contaminated by the clutter plus interference plus noise.

Other contributions in waveform design and waveform scheduling for  extended targets in radar using information theoretic measures, tracking etc can be seen in \cite{setlurspie2012,setlurssp2012,bell1993information,
bell1988thesis,Lesham2007,romero2008information}, \cite{paper:AussiesSensorScheduleRadar,sira2006waveform,
sira2009waveform,Kershaw:2004,kershaw1997waveform,Li2008}, and the references therein.

We outline some of the contributions for the signal dependent interference problem which have thus far appeared in the literature.

{\it Approaches different from Alternating minimization}: In \cite{Delong1967,Kay2007,Pillai2000,Stoica2012} a single sensor radar was assumed. In \cite{Delong1967}, the authors used the symmetry property of the cross-ambiguity function to design an iterative algorithm for the signal dependent interference problem. Their algorithm cannot be modified easily for the multi-sensor framework and when noise is in general colored. The problem was addressed from  a detection perspective in \cite{Kay2007}, and lead to a waterfilling \cite{bell1993information} type solution. A similar waterfilling type metric albeit in the discrete time domain was obtained in \cite{Stoica2012}, where the authors also imposed constant modulus and peak to average power ratio (PAPR) waveform constraints. An iterative algorithm was derived in \cite{Pillai2000}, where monotonic increase in SINR was not guaranteed, and was shown that waveform could always be chosen as minimum phase. 

{\it Alternating minimization type approaches}: In \cite{vaidyanathan2009}, a MIMO sensor framework was employed, convergence was not addressed, convexity was not proven, and no practical waveform constraints were imposed on the design. See also in this report, Section III, paragraph following Rem.~\ref{limitpointremk} where some of the conclusions drawn in \cite{vaidyanathan2009} are further discussed. Alternating minimization was used in \cite{Demaio2012,Demaio2013,Hongbin2014} but for reasons unknown, was called as sequential optimization. In \cite{Demaio2012,Demaio2013}, a SISO model advocating joint filter and radar code design (after matched filtering) was employed.  Analysis of the convexity of the objective in the individual filter or radar code was never shown. Convergence in iterates was not proven formally, neither was it shown via simulations. The constant modulus constraint was not invoked directly but through a similarity constraint. In \cite{Hongbin2014}, the authors used a  MIMO radar framework, and relaxation techniques were employed in their iterative algorithm. Neither convergence nor convexity was demonstrated analytically. Constant modulus constraint and similarity constraints were enforced separately in the waveform design. 

{\bf Notation:} The variable $N$ is used interchangeably with the number of the fast time samples, as well as, the conventional dimension of arbitrary real or complex (sub)spaces. Its meaning is readily interpreted from context. The  symbol $|| \cdot||$ always denotes the $l_2$ norm. Vectors are always lowercase bold, matrices are bold uppercase, $\lambda$ is typically reserved for eigenvalues (with $\lambda_o$ being an exception it used for the spatial frequency, defined later) and $\gamma$ is strictly reserved for the Lagrange multipliers ($\gamma_{pq}$ is an exception used for the radar cross section of the $p$-th scatterer in the $q$-th clutter patch). Solutions to the optimization are denoted as $(\cdot)_o$, i.e. the subscript $o$. the complex conjugate is denoted with $(\cdot)^{\ast}$. The set of reals, complex numbers, and natural numbers are denoted as $\mathbbm{R},\mathbbm{C},\mathbbm{N}$, respectively. Other symbols are defined upon first use and are standard in the literature. 

{\bf Organization:} The STAP fast time-slow time model is delineated in Section II, and in Section III, the filter and waveform optimization is derived. Some preliminary simulations are presented in section IV and  the resulting conclusions are drawn in Section V.


\section{STAP Model}
The radar consists of a calibrated  air-borne linear array, comprising $M$ sensor elements, each having an identical antenna pattern. Without loss of generality, assume that the first sensor in the array is the phase center, and acts as both a transmitter and receiver, the rest of the elements are purely receivers. The first sensor is located at $\mathbf{x_r}\in\mathbbm{R}^3$ and the ground based point target at  $\mathbf{x_t}\in\mathbbm{R}^3$. The radar transmits the burst of pulses: 
\begin{equation} \label{eq1}
 u(t)=\sum\limits_{l=1}^{L}s(t-lT_p)\exp(j2\pi f_o(t-lT_p)),t\in[0,T)
 \end{equation}
where, $f_o$ is the carrier frequency, and $T_p=1/f_p$ is the inverse of the pulse repetition frequency, $f_p$. The pulse width and bandwidth are denoted as $T$, $B$, respectively. The coherent processing interval (CPI) consists of $L$ pulses, each of width equal to $T$. The geometry of the scene is shown in Fig.~1, where $\theta_t$ and $\phi_t$ denote the azimuth and elevation. The radar and target are both assumed to be moving. 
 
For the time being, we ignore the noise, clutter and interference and assume a non-fluctuating target. Then the desired target's received signal for the $l$-th pulse, and at the $m$-th sensor element is given by
\begin{equation}\label{eq2}
s_{ml}(t)=\rho_t s(t-lT_p-\tau_m)e^{(j2\pi (f_o+f_{dm}) (t-lT_p-\tau_m))}
\end{equation}
where the target's observed Doppler shift is denoted as $f_{dm}$, and its  complex back-scattering coefficient  as  $\rho_t$. Assume that the array is along the local $x$ axis as shown in Fig.~1. Then, the coordinates of the $m$-th element is given by $\mathbf{x_t}+m\mathbf{d},\mathbf{d}:=[d,0,0]^T,m=0,1,2\ldots,M-1$, where $d$ is the inter-element spacing. The delay $\tau_m$ could be re-written as
\small \begin{align}
&\tau_m=||\mathbf{x_r}-\mathbf{x_t}||/c+||\mathbf{x_r}+m\mathbf{d}-\mathbf{x_t}||/c \nonumber \\
       &=\dfrac{||\mathbf{x_r}-\mathbf{x_t}||}{c}+\dfrac{||\mathbf{x_r}-\mathbf{x_t}||}{c} \sqrt{1+\frac{||m\mathbf{d}||^2}{||\mathbf{x_r}-\mathbf{x_t}||^2}+\frac{2m\mathbf{d}^T(\mathbf{x_r}-\mathbf{x_t})}{||\mathbf{x_r}-\mathbf{x_t}||^2}} \nonumber \\
       &\overset{(a)}{\equiv}\dfrac{||\mathbf{x_r}-\mathbf{x_t}||}{c}+\dfrac{||\mathbf{x_r}-\mathbf{x_t}||}{c} \left(1+\frac{m\mathbf{d}^T(\mathbf{x_r}-\mathbf{x_t})}{||\mathbf{x_r}-\mathbf{x_t}||^2} \right) \label{eq3} \\
       &=2\dfrac{||\mathbf{x_r}-\mathbf{x_t}||}{c}+\frac{m\mathbf{d}^T(\mathbf{x_r}-\mathbf{x_t})}{c||\mathbf{x_r}-\mathbf{x_t}||}, \label{eq4}
\end{align}
where in approximation (a), the term $\propto||m\mathbf{d}||^2$ was ignored, i.e. it is assumed that $d/|| \mathbf{x_r}-\mathbf{x_t}||<<1$, and then a binomial approximation was employed. From geometric manipulations, we also have:
\begin{equation*}
\frac{\mathbf{x_r}-\mathbf{x_t}}{||\mathbf{x_r}-\mathbf{x_t}||}=[\sin(\phi_t)\sin(\theta_t),\sin(\phi_t)\cos(\theta_t),\cos(\phi_t)]^T.
\end{equation*}
Using the above equation in \eqref{eq4}, the delay $\tau_m,m=0,1,\ldots,M-1$ can be rewritten as
\begin{equation} \label{eq5}
\tau_m=2\dfrac{||\mathbf{x_r}-\mathbf{x_t}||}{c}+\dfrac{md\sin(\phi_t)\sin(\theta_t)}{c}.
\end{equation}
The Doppler shift, i.e. $f_{dm}$ is computed as
\begin{align} \label{eq6}
&f_{dm}=2f_o\dfrac{(\mathbf{\dot{x}_r}-\mathbf{\dot{x}_t})^T(\mathbf{x_r}-\mathbf{x_t})}{c||\mathbf{x_r}-\mathbf{x_t}||}  \\
&+ f_o\dfrac{m\mathbf{d}^T}{c}\left[\dfrac{\mathbf{\dot{x}_r}-\mathbf{\dot{x}_t}}{||\mathbf{x_r}-\mathbf{x_t}||^2} -\dfrac{(\mathbf{x_r}-\mathbf{x_t})(\mathbf{\dot{x}_r}-\mathbf{\dot{x}_t})^T(\mathbf{x_r}-\mathbf{x_t})}{\|| \mathbf{x_r}-\mathbf{x_t}||^3} \right] \nonumber
\end{align}
where $\mathbf{\dot{x}_{(\cdot)}}$ is the vector differential of $\mathbf{x_{(\cdot)}}$ w.r.t. time. In practice $d$ is a fraction of the wavelength, and assuming that $d/|| \mathbf{x_r}-\mathbf{x_t}||<<1$ we approximate the second term in \eqref{eq6} as $0$. The Doppler shift is no longer a function of the sensor index, $m$, and is rewritten as
\begin{equation} \label{eq7}
f_{dm}=f_d=2f_o\dfrac{(\mathbf{\dot{x}_r}-\mathbf{\dot{x}_t})^T(\mathbf{x_r}-\mathbf{x_t})}{c||\mathbf{x_r}-\mathbf{x_t}||}
\end{equation}

\subsection{Vector signal model}
Let $s(t)$ be sampled discretely resulting in $N$ discrete time samples. Consider for now the single range gate corresponding to the time delay $\tau_t$. After a suitable alignment to a common local time (or range) reference,  and invoking some standard assumptions, see also \cite[{\it A.1-A.3}] {Setlurradar2013}, the radar returns in $l$-th PRI  written as a vector defined as $\mathbf{y}_{\mathbi{l}}\in\mathbbm{C}^{NM}$, is given by
\begin{align}
&\mathbf{y}_{\mathbi{l}}=\rho_t\mathbf{s}\otimes \mathbf{a}(\theta_t,\phi_t)\exp(-j2\pi f_d(l-1)T_p) \label{eq9} \\
&\mathbf{a}(\theta_t,\phi_t):=[1,e^{-j2\pi\vartheta},\ldots,e^{-j2\pi (M-1)\vartheta}]^T \in\mathbbm{C}^M \nonumber
\end{align}
where $\mathbf{s}:=[s_1,s_2,\ldots,s_N]^T \in\mathbbm{C}^N  $ and $\vartheta:=d\sin(\theta_t)\sin(\phi_t)/\lambda_o$ is defined as the spatial frequency. Further it is noted that in \eqref{eq9}, the constant phase terms have been absorbed into $\rho_t$. Considering the $L$ pulses together, i.e. concatenating the desired target's response for the entire CPI in a tall vector $\mathbf{y}$, is defined as
\begin{align}
&\mathbf{y}\in\mathbbm{C}^{NML}=[\mathbf{y_0}^T,\mathbf{y_1}^T,\ldots,\mathbf{y}_{\mathbi{L-1}}^T]^T =\rho_t \mathbf{v}(f_d) \otimes\mathbf{s}\otimes \mathbf{a}(\theta_t,\phi_t)\nonumber \\
&\mathbf{v}(f_d):=[1,e^{-j2\pi f_dT_p},\ldots,e^{-j2\pi f_d(L-1)T_p}]^T \label{eq10}.
\end{align}
The vector $\mathbf{y}$ consists of both the spatial and the temporal steering vectors as in classical STAP, as well as the waveform dependency, via waveform vector $\mathbf{s}$. Due to inclusion of the fast time samples in the waveform $\mathbf{s}$, the STAP data cube is modified to reflect this change, and is depicted in Fig.~\ref{fig2}.

At the considered range gate, the measured snapshot vector consists of the target returns and the undesired returns, i.e. clutter returns, interference and noise. The contaminated snapshot at the considered range gate is then given by
\begin{align}
\mathbf{\tilde{y}}=\mathbf{y}+\mathbf{y_i}+\mathbf{y_c}+\mathbf{y_n} 
                =\mathbf{y}+\mathbf{y_{u}} \label{eq.11} 
\end{align}
where $\mathbf{y_i,y_c,y_n}$ are the contributions from the interference, clutter and noise, respectively, and are assumed to be statistically uncorrelated with one another. The contribution of the undesired returns are treated in detail, starting with the noise as it is the simplest.

{\bf Noise}: The noise is assumed to be zero mean, identically distributed across the sensors, across pulses, and  in the fast time samples. The correlation matrix of $\mathbf{y_n}$ is denoted as $ \mathbf{R_n}\in\mathbbm{C}^{NML\times NML}$. 

{\bf Interference}: The interference consists of jammers and other intentional / un-intentional sources which may be ground based, air-borne or both. Let us assume that there are $K$ interference sources. Further, since nothing is known about the jammers waveform characteristics, 
the waveform itself is assumed to be a stationary zero mean random process. Consider the $k$-th interference source in the $l$-th PRI, and at spatial co-ordinates $(\theta_k,\phi_k)$. Its corresponding snapshot contribution is modeled as,
\begin{equation*}
\mathbf{y}_{kl}=\boldsymbol{\alpha_{kl}}\otimes\mathbf{a}(\theta_k,\phi_k),k=1,2,\ldots,K,l=0,1,\ldots,L-1
\end{equation*}
where $\boldsymbol{\alpha_{kl}}=[\alpha_{kl}(0),\alpha_{kl}(1),\ldots,\alpha_{kl}(N-1)]^T\in\mathbbm{C}^{N}$ is the random discrete segment of the jammer waveform, as seen by the radar in the $l$-th PRI. Stacking $\mathbf{y}_{kl}$ for a fixed $k$ as a tall vector, we have 
\begin{align} \label{eq12}
\mathbf{y_k}&=\boldsymbol{\alpha_k}\otimes\mathbf{a}(\theta_k,\phi_k) 
            =[\mathbf{y}_{ko}^T,\mathbf{y}_{k1}^T,\ldots,\mathbf{y}_{kL-1}^T]^T \in\mathbbm{C}^{NML}  \nonumber\\
\boldsymbol{\alpha_k}:&=[\boldsymbol{\alpha_{k0}}^T,\boldsymbol{\alpha_{k1}}^T,\ldots,\boldsymbol{\alpha_{kL-1}}^T]^T \in\mathbbm{C}^{NL} 
\end{align}
Using the Kronecker mixed product property, (see for e.g. \cite{horn1994}), the correlation matrix of $\mathbf{y}_k$ is expressed as 
$\mathbbm{E}\{\mathbf{y}_k\mathbf{y}_k^H \}=\mathbf{R}_{\boldsymbol{\alpha}}^k\otimes \mathbf{a}(\theta_k,\phi_k)\mathbf{a}(\theta_k,\phi_k)^H$
where, $\mathbbm{E}\{ \boldsymbol{\alpha_k} \boldsymbol{\alpha_k}^H\}:=\mathbf{R}_{\boldsymbol{\alpha}}^k$. For $K$ mutually uncorrelated interferers, the correlation matrix is $\mathbf{R_i}=\sum\limits_{k=1}^K\mathbbm{E}\{\mathbf{y}_k\mathbf{y}_k^H \}=\sum\limits_{k=1}^K \mathbf{R}_{\boldsymbol{\alpha}}^k\otimes \mathbf{a}(\theta_k,\phi_k)\mathbf{a}(\theta_k,\phi_k)^H=\sum\limits_{k=1}^K (\mathbf{I}_{NL}\otimes\mathbf{a}(\theta_k,\phi_k))\mathbf{R}_{\boldsymbol{\alpha}}^k(\mathbf{I}_{NL}\otimes\mathbf{a}(\theta_k,\phi_k)^H) $, and is simplified as
\begin{align} \label{eq13}
\mathbf{R_i}&= \mathbf{A}(\theta,\phi)\mathbf{R}_{\boldsymbol{\alpha}}\mathbf{A}(\theta,\phi)^H
\end{align}
where $\mathbf{R}_{\boldsymbol{\alpha}}:=\mbox{Diag}\{\mathbf{R}_{\boldsymbol{\alpha}}^1,\mathbf{R}_{\boldsymbol{\alpha}}^2,\ldots,\mathbf{R}_{\boldsymbol{\alpha}}^K\} \in\mathbbm{C}^{NMLK\times NMLK}$ and $\mathbf{A}(\theta,\phi)\in\mathbbm{C}^{NML\times NMLK}\\
=[\mathbf{I}_{NL}\otimes \mathbf{a}(\theta_1,\phi_1),\mathbf{I}_{NL}\otimes \mathbf{a}(\theta_2,\phi_2),\ldots,\mathbf{I}_{NL}\otimes \mathbf{a}(\theta_K,\phi_K)]$, 
here $\mathbf{I}_{NL}$ the identity matrix of size $NL\times NL$, and $\mbox{Diag}\{ \cdot,\cdot,\ldots,\cdot\}$ the matrix diagonal operator which converts the matrix arguments into a bigger diagonal matrix. For example, $\mbox{Diag}\{\mathbf{A,B,C}\}=\left[ \begin{smallmatrix} \mathbf{A}&\mathbf{0}&\mathbf{0}\\
\mathbf{0}&\mathbf{B}&\mathbf{0} \\ \mathbf{0}&\mathbf{0}&\mathbf{C} \end{smallmatrix} \right].$

{\bf Clutter}: The ground is a major source of clutter in air-borne radar applications and is persistent in all range gates upto the gate corresponding to the platform horizon. Other sources of clutter surely exist, such as buildings, trees, as well as other un-interesting targets, which are ignored. We therefore consider only  ground clutter and treat it stochastically.

Let us assume that there are $Q$ clutter patches indexed by parameter $q$. Each of these clutter patches are comprised of say $P$ scatterers. The radar return from the $p$-th scatterer in the $q$-th clutter patch is given by
\begin{equation*}
\gamma_{pq}\mathbf{v}(fc_{pq})\otimes\mathbf{s} \otimes a(\theta_{pq},\phi_{pq})
\end{equation*}
where $\gamma_{pq}$ is its random complex reflectivity, $fc_{pq}$ is the Doppler shift observed from the $p$-th scatterer in the $q$-th clutter patch, and $\theta_{pq},\phi_{pq}$ are the azimuth and elevation angles of this scatterer.
The Doppler $fc_{pq}$ is given by,
\begin{equation} \label{eq14}
fc_{pq}:=\dfrac{2f_o\mathbf{\dot{x}_r}^T (\mathbf{x_r}-\mathbf{x_{pq} })}{c||\mathbf{x_r}-\mathbf{x_{pq} }||}.
\end{equation}
where $\mathbf{x_{pq}}$ is the location of the $p$-th scatter in the $q$-th clutter patch. Since the clutter patch is stationary, the Doppler is purely from the motion of the aircraft as seen in \eqref{eq14}. The contribution from the $q$-th clutter patch to the received signal is given by
\begin{equation} \label{eq15}
\mathbf{y}_q=\sum \limits_{p=1}^P \gamma_{pq} \mathbf{v}(fc_{pq} )\otimes\mathbf{s} \otimes a(\theta_{pq},\phi_{pq} ),
\end{equation}
with corresponding correlation matrix 
\begin{align} \label{cluteq}
&\mathbf{R}_{\boldsymbol{\gamma}}^q:=\mathbf{B_q}\mathbf{R}_{\boldsymbol{\gamma}}^{pq} \mathbf{B_q}^H
\end{align}
where, 
$\mathbf{B_q}
=[\mathbf{v}(fc_{1q})\otimes\mathbf{s} \otimes \mathbf{a}(\theta_{1q},\phi_{1q}), \mathbf{v}(fc_{2q})\otimes\mathbf{s} \otimes \mathbf{a}(\theta_{2q},\phi_{2q})\ldots,\mathbf{v}(fc_{Pq})\otimes\mathbf{s} \otimes \mathbf{a}(\theta_{Pq},\phi_{Pq})] \in\mathbbm{C}^{NML\times P}$ and $\mathbf{R}_{\boldsymbol{\gamma}}^{pq}$ is the correlation matrix of the random vector, $[\gamma_{1q},\gamma_{2q},\ldots,\gamma_{Pq}]^T$. It is readily shown that the matrix $\mathbf{B_q}$ could be simplified as, $\mathbf{B_q}:=\mathbf{\breve{B}_q}(\mathbf{I}_P\otimes\mathbf{s})$, where $\mathbf{\breve{B}_q}:=[\mathbf{v}(fc_{1q})\otimes\mathbf{A_{1q}},\mathbf{v}(fc_{2q})\otimes\mathbf{A_{2q}},\ldots, \mathbf{v}(fc_{Pq})\otimes\mathbf{A_{Pq}}]\in \mathbbm{C}^{NML\times PN}$, and the structure of the matrix $\mathbf{A_{pq}}\in\mathbbm{C}^{NM\times N}$ (straightforward but not shown here) is defined such that $\mathbf{s}\otimes\mathbf{a}(\theta_{pq},\phi_{pq} )=\mathbf{A_{pq}}\mathbf{s},p=1,\ldots,P$. Assuming that a particular scatterer from one clutter patch is uncorrelated to any other scatterer belonging to any other clutter patch, we have the net contribution of clutter $\mathbf{y_c}=\sum\limits_{q=1}^Q \mathbf{y}_q$, with corresponding correlation matrix given by
\begin{equation} \label{eq16}
\mathbf{R_c}=\sum\limits_{q=1}^Q \mathbf{R}_{\boldsymbol{\gamma}}^q.
\end{equation}
The clutter model could further be simplified by the following arguments. Assuming a large range resolution which is typically the case for radar STAP \cite{ward1994} the scatterers in a particular clutter patch are in the same range gate and hence are assumed to possess approximately identical Doppler shifts, i.e. $fc_{pq}\approx fc_q=\tfrac{2f_o\mathbf{\dot{x}_r}^T (\mathbf{x_r}-\mathbf{x_{q} })}{c||\mathbf{x_r}-\mathbf{x_{q} }||}$. Similarly for the far field operation, and considering scatterers in the same azimuth resolution cell, and from the large range resolution argument, we may assume $\theta_{pq}\approx \theta_q$ and $\phi_{pq}\approx \phi_{q}$, i.e. their nominal angular centers. These assumptions can now be incorporated in matrix $\mathbf{B_q}$ to simplify the clutter model, see also \cite{Setlurradar2013}.
\begin{figure*}[tbp!]
\begin{minipage}[b]{0.5\linewidth}
\centering
\includegraphics[scale=0.4]{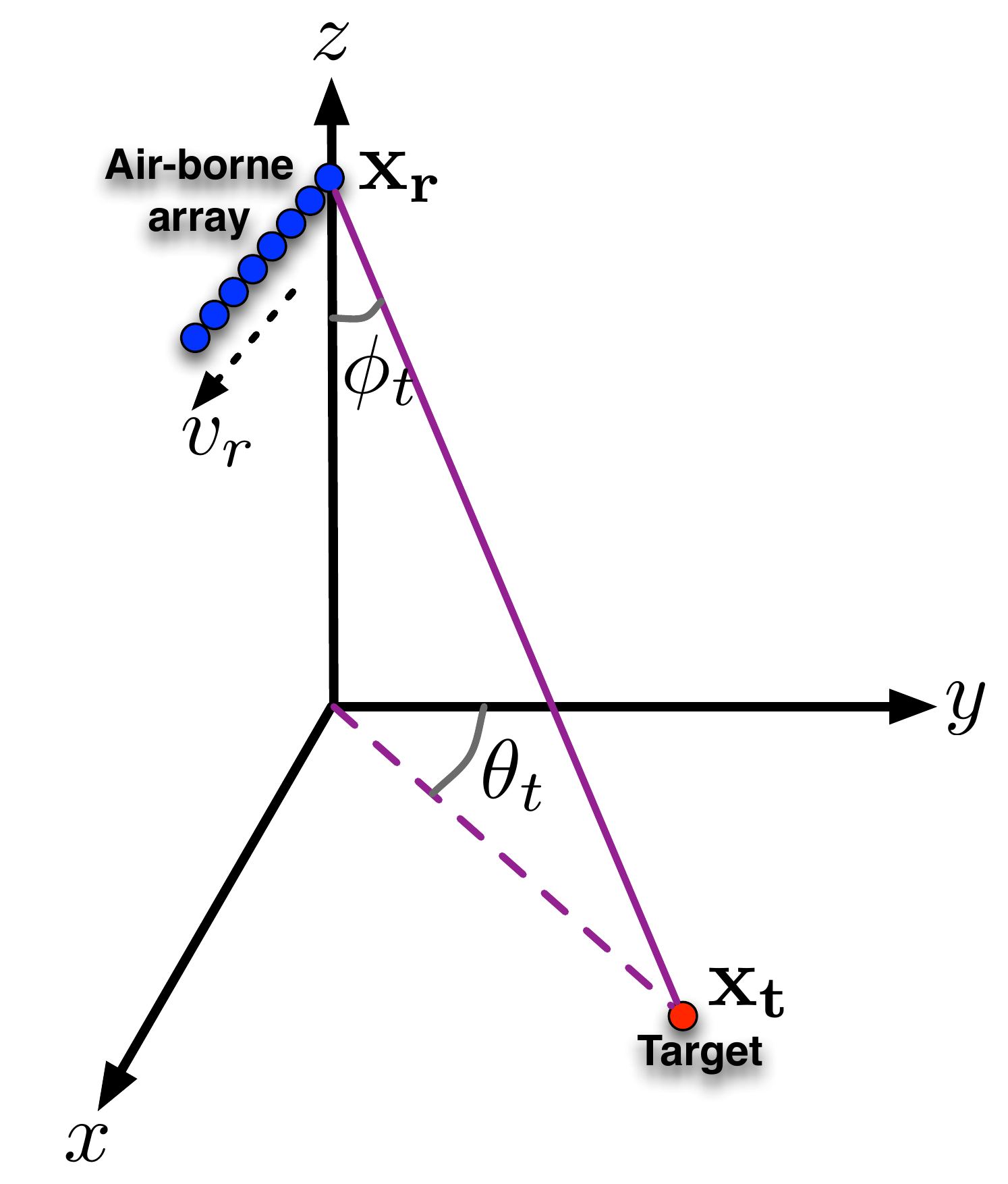}
\caption{Radar scene considering the ground based target at azimuth $(\theta_t)$, elevation $(\phi_t)$. The $(x,y,z)$ axis are local to the aircraft carrying the array.}
\label{fig1}
\end{minipage}
\hspace{0.4cm}
\begin{minipage}[b]{0.5\linewidth}
\centering
\includegraphics[scale=0.4]{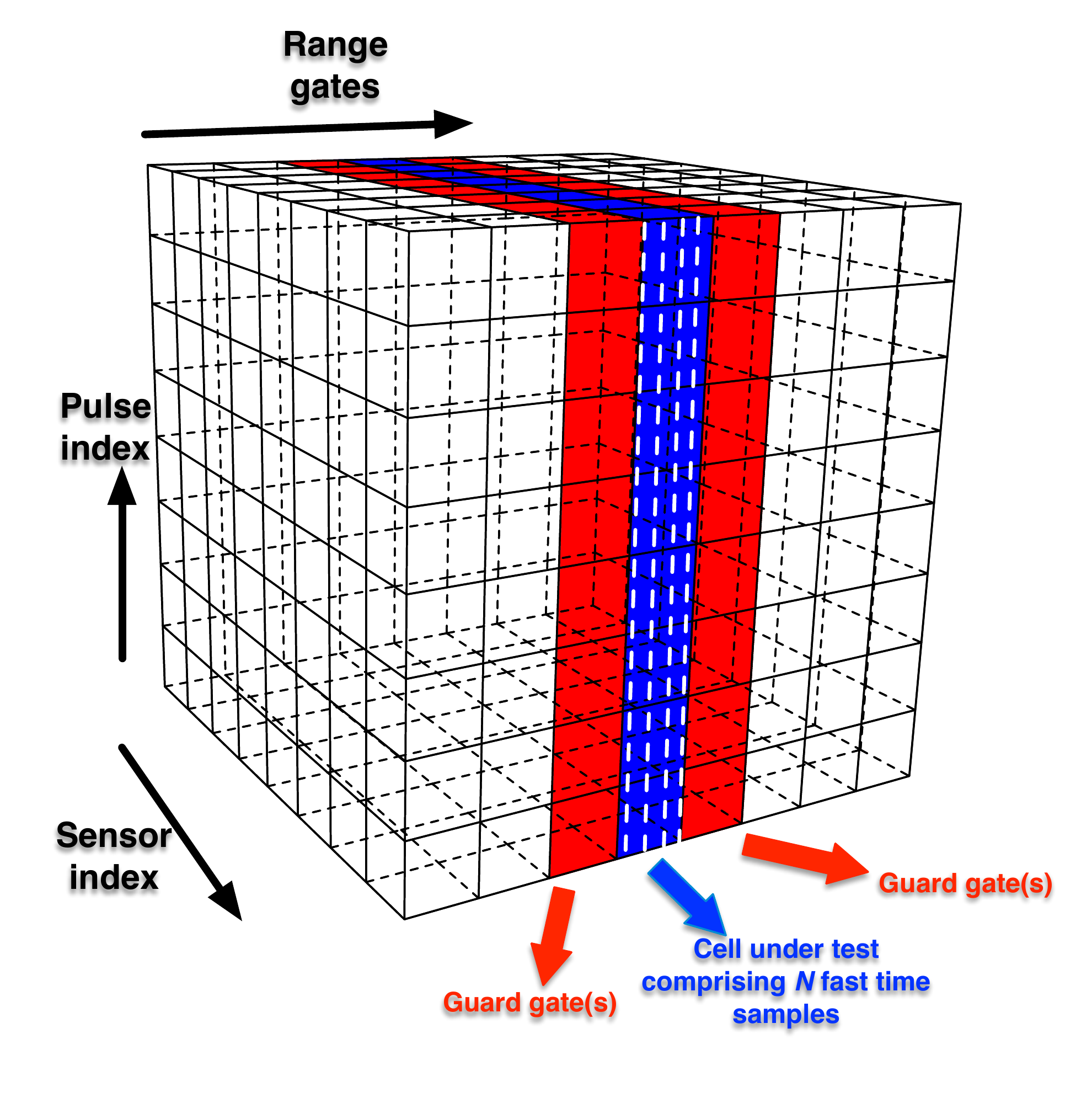}
\caption{STAP data cube before matched filtering or range compression, depicting the considered range gate/cell and fast time slices (dashed lines).}
\label{fig2}
\end{minipage}
\end{figure*}
\section{Waveform Design}
The radar return at the considered range gate is processed by a filter characterized by a weight vector, $\mathbf{w}$, whose output is given by $\mathbf{w}^H\mathbf{\tilde{y}}$. Since the vector $\mathbf{s} \in\mathbbm{C}^{N}$ prominently figures in the steering vectors, the objective is to jointly obtain the desired weight vector, $\mathbf{w}$ and waveform vector, $\mathbf{s}$. It is desired that the weight vector will  minimize the output power, $\mathbbm{E}\{|\mathbf{w}^H\mathbf{y_u}|^2\}=\mathbf{w}^H\mathbf{R_u}( \mathbf{s})\mathbf{w}$. Mathematically, we may formulate this problem as:
\begin{align}
\min_{\mathbf{w},\mathbf{s}}\;\;\;\;\; &\mathbf{w}^H\mathbf{R_u}( \mathbf{s})\mathbf{w}\label{eq17} \\
\mbox{ s. t } \;\;\;\;\;&\mathbf{w}^H(\mathbf{v}(f_d)\otimes\mathbf{s}\otimes\mathbf{a}(\theta_t,\phi_t))=\kappa \nonumber \\
\;\;\;\;\;\; & \mathbf{s}^H \mathbf{s}\leq P_o \nonumber \nonumber
\end{align}
In \eqref{eq17}, the first constraint is the renowned,well known Capon constraint with $\kappa\in\mathbbm{R}$, typically $\kappa=1$. An energy constraint enforced via the second constraint is to addresses hardware limitation.  Before we derive the solutions to the optimization problem, it is useful to recall  Lem.~\ref{crlemma}, which is well-known, used throughout this report but not stated explicitly. This fundamental result discusses the technique to compute stationary points of a real valued function w.r.t. its complex valued argument and its conjugate.

\begin{lem}\label{crlemma}
Let $f(\mathbf{x},\mathbf{x}^{\ast}):\mathbbm{C}^N\rightarrow\mathbbm{R}$. The stationary point of $f(\mathbf{x},\mathbf{x}^{\ast})=\bar{f}(\mathbf{x_r},\mathbf{x_i})$ is found from the three equivalent conditions, 1.  $\nabla_{\bf x_r}\bar{f}(\mathbf{x_r},\mathbf{x_i})=\mathbf{0}$ and $\nabla_{\bf x_i}\bar{f}(\mathbf{x_r},\mathbf{x_i})=\mathbf{0}$, or 2. $\nabla_{\bf x}f(\mathbf{x},\mathbf{x}^{\ast})=\mathbf{0}$, or 3. $\nabla_{\bf x^{\ast}} f(\mathbf{x},\mathbf{x}^{\ast})=\mathbf{0}$. Here $\bar{f}:\mathbbm{R}^{N}\times\mathbbm{R}^{N}\rightarrow\mathbbm{R}$  is the real equivalent of $f(\cdot,\cdot)$, $\mathbf{x_r}=\mathrm{Re}\{\mathbf{x}\},\mathbf{x_i}=\mathrm{Im}\{\mathbf{x}\}$, where we define the gradient $\nabla_{\bf x} f(\mathbf{x},\mathbf{x}^{\ast}):=[\tfrac{\partial f( \cdot,\cdot)}{\partial x_1},\tfrac{\partial f( \cdot,\cdot)}{\partial x_2},\cdots,\tfrac{\partial f( \cdot,\cdot)}{\partial x_N} ]^T$ with $x_i$ as the $i$-th element of ${\bf x},i=1,2,\ldots N$, and $\mathbf{0}$ is a column vector of all zeros of dimension $N$.
\end{lem}
\begin{proof}
This arises from the Wirtinger calculus see \cite{Gesbert2007} \footnote{ also see refs. Brandwood, and A. van den Bos in \cite{Gesbert2007} } for a recent formal proof.
\end{proof}
Optimizing \eqref{eq17} w.r.t. $\mathbf{w}$ first, the solution to \eqref{eq17} is well known, and expressed as
\begin{equation} \label{weightcomp}
\mathbf{w}_{o}=\frac{\kappa \mathbf{R}_{\bf u}^{-1}(\mathbf{s})(\mathbf{v}(f_d)\otimes\mathbf{s}\otimes\mathbf{a}(\theta_t,\phi_t))}{(\mathbf{v}(f_d)\otimes\mathbf{s}\otimes\mathbf{a}(\theta_t,\phi_t))^H\mathbf{R}_{\bf u}^{-1} (\mathbf{s})(\mathbf{v}(f_d)\otimes\mathbf{s}\otimes\mathbf{a}(\theta_t,\phi_t))}
\end{equation}
where $\mathbf{R_u}(\mathbf{s})=\mathbf{R_i}+\mathbf{R_c}(\mathbf{s})+\mathbf{R_n}$.  We further emphasize that the  weight vector is an explicit function of the waveform.  Now substituting $\mathbf{w}_{o}$ back into the cost function in \eqref{eq17}, the minimization is purely w.r.t. $\mathbf{s}$, and cast as,
\begin{align} \label{eq18}
\min_{\mathbf{s}} &\frac{\kappa^2}{(\mathbf{v}(f_d)\otimes\mathbf{s}\otimes\mathbf{a}(\theta_t,\phi_t))^H\mathbf{R}_{\bf u}^{-1}(\mathbf{s})(\mathbf{v}(f_d)\otimes\mathbf{s}\otimes\mathbf{a}(\theta_t,\phi_t))} \nonumber \\
\mbox{ s. t. } & \mathbf{s}^H \mathbf{s}\leq P_o 
\end{align}
A solution to \eqref{eq18} is not immediate, given the dependence of  $\mathbf{R_u}$ on the waveform vector $\mathbf{s}$. We consider first, the case when the clutter dependence on the waveform is ignored. Solutions to the design when clutter is considered are treated subsequently.

\subsection{Rayleigh-Ritz: Minimum eigenvector solution}
Ignoring the dependency of $\mathbf{R_u}$ on $\mathbf{s}$, we readily see that the \eqref{eq18} can be recast as a Rayleigh-Ritz optimization, whose solution is given by
\begin{align} \label{eq21}
\mathbf{v}(f_d)\otimes\mathbf{s}\otimes\mathbf{a}(\theta_t,\phi_t)=\boldsymbol{\mu}_{\min}(\mathbf{R_u})
\end{align}
where $\boldsymbol{\mu}_{\min}(\mathbf{R_u})$ is the eigenvector corresponding to the minimum eigenvalue of $\mathbf{R_u}$. This tensor equation implicitly defines the optimal $\mathbf{s}$. It is readily seen that, $\mathbf{v}(f_d)\otimes\mathbf{s}\otimes\mathbf{a}(\theta_t,\phi_t)=\mathbf{G}\mathbf{s}$, where $\mathbf{G}=\mathbf{v}(f_d)\otimes\mathbf{A_t}$, and 
\begin{align*}
\mathbf{A_t}=\begin{bmatrix}
    \mathbf{a}(\theta_t,\phi_t) &\mathbf{0} &\mathbf{0} &\cdots &\mathbf{0} \\
    \mathbf{0}& \mathbf{a}(\theta_t,\phi_t) & \mathbf{0} &\cdots &\mathbf{0} \\
    \mathbf{0}& \mathbf{0} &\mathbf{a}(\theta_t,\phi_t) &\vdots  &\vdots \\
    \vdots & \vdots &\vdots &\vdots &\vdots
    \end{bmatrix} \in\mathbbm{C}^{MN\times N}.
\end{align*}
In general, the system is over-determined, and we solve this equation approximately via least squares (LS), 
\begin{align} \label{lsubopt}
\mathbf{\hat{s}}=(\mathbf{G}^H\mathbf{G})^{-1}\mathbf{G}^H\mu_{\min}(\mathbf{R_u}).
\end{align}
Moreover from \eqref{eq21} and the structure of the temporal and spatial steering vectors,  as well as the orthonormality of the eigenvectors, it is readily seen that,
\begin{align} \label{eqrayleigh}
||\mathbf{v}(f_d)\otimes\mathbf{s}\otimes\mathbf{a}(\theta_t,\phi_t)||^2&=
|| \mathbf{v}(f_d)||^2 || \mathbf{s}||^2 ||\mathbf{a}(\theta_t,\phi_t) ||^2=|| \mathbf{s}||^2 \nonumber \\
||\boldsymbol{\mu}_{\min}(\mathbf{R_u})||^2&=1.
\end{align}
Hence  the LS solution in \eqref{lsubopt} must be scaled to satisfy the desired energy requirements of the radar system.
 
 {\bf Decoupling LS:} The LS solution in \eqref{lsubopt} can be further simplified due to the following linear relation between elements of $\mathbf{v}(f_d),\mathbf{a}(\theta_t,\phi_t),\mathbf{s}$ and elements of $\boldsymbol{\mu}_{\min}(\mathbf{R_u})$, expressed as
\begin{align} \label{linearrel}
v_la_ms_n=\mu_h,\; l&=1,2,\ldots,L,\;m=1,2,\ldots,M,\;n=1,2,\ldots,N \nonumber \\
h&=(l-1)MN+(n-1)M+m.
\end{align}
where $v_l$, $a_m$, $s_n$ are the $l$-th, $m$-th, $n$-th elements of $\mathbf{v}(f_d)$, $\mathbf{a}(\theta_t,\phi_t)$, $\mathbf{s}$, and $\mu_h$ is the $h$-th element of $\boldsymbol{\mu}_{min}(\mathbf{R_u})$, respectively. Therefore, the LS solution in \eqref{lsubopt} decouples as
\begin{align} \label{lsfinal}
s_n=\frac{(\mathbf{v}(f_d)\otimes\mathbf{a}(\theta_t,\phi_t) )^H\boldsymbol{\mu}_{\bf n} }{(\mathbf{v}(f_d)\otimes\mathbf{a}(\theta_t,\phi_t) )^H(\mathbf{v}(f_d)\otimes\mathbf{a}(\theta_t,\phi_t) )},\; n=1,2,\ldots,N
\end{align}
where the vector $\boldsymbol{\mu}_{\bf n}\in\mathbbm{C}^{ML}$ for a {\it particular $n$} consists of the $ML$ appropriate elements, $\mu_h,\; h=(l-1)MN+(n-1)M+m,\;  m=1,2,\ldots,M, \;l=1,2,\ldots,L$, as highlighted in \eqref{linearrel}.

The min. eigenvector solution is most relevant when noise and interference are considered and clutter is ignored in the waveform design \cite{guerci2003}. it has some nice spectral properties  similar (but not identical) to water-filling \cite{guerci2003,bell1993information}. Therefore this solution, although suboptimal, is a good initial waveform to  interrogate the radar scene, but is unfortunately well known to suffer from poor modulus and sidelobe properties. Nonetheless, in certain exceptional cases and in the presence of clutter, this suboptimal solution is shown to be optimal, and is discussed at a later stage. The ensuing definitions and lemma proves useful subsequently.
\begin{lem} \label{lemma1}
(a) If vectors $\boldsymbol{\alpha}$, $\boldsymbol{\beta}$ and $\boldsymbol{\gamma}$ consist of the eigenvalues of the square but not necessarily Hermitian matrices, $\mathbf{X}\in\mathbbm{C}^{N\times N}$, $\mathbf{Y}\in\mathbbm{C}^{M\times M}$ and $\mathbf{X}\otimes\mathbf{Y}$, respectively. Then $\boldsymbol{\gamma}=\boldsymbol{\alpha}\otimes\boldsymbol{\beta}$. (b) Also,  $\mathrm{rank}(\mathbf{X}\otimes\mathbf{Y})=\mathrm{rank}(\mathbf{X})\otimes\mathrm{rank}(\mathbf{Y})$.
\end{lem}

\begin{proof}
For (a), let  $\mathbf{x}_i,i=1,2,\ldots,N$ and $\mathbf{y}_j,j=1,2,\ldots,M$ are the eigenvectors corresponding to $\alpha_i,\beta_j$ i.e. the $i$-th and $j$-th eigenvalues, of $\mathbf{X,Y}$, respectively. Then, from the mixed property of the Kronecker product, $\mathbf{X}\mathbf{x}_i\otimes\mathbf{Y}\mathbf{y}_j=(\mathbf{X}\otimes\mathbf{Y})(\mathbf{x}_i\otimes\mathbf{y}_j)$ but the eigenvector relations imply that $\mathbf{X}\mathbf{x}_i=\alpha_i\mathbf{x}_i,\mathbf{Y}\mathbf{y}_j=\beta_j\mathbf{y}_j$. This implies that the $ij$-th eigenvalue of of $\mathbf{X}\otimes\mathbf{Y}$ is  $\gamma_{ij}=\alpha_i\beta_j$ with associated eigenvector $\mathbf{x}_i\otimes\mathbf{y}_j$. Since the rank is equal to the number of non-zero eigenvalues for square matrices, the second follows directly from (a). Hence proved.
\end{proof}

\begin{defn}\label{mydef1} ({\it Convexity})
 A function $f(\mathbf{x}):\mathbbm{R}^N\rightarrow\mathbbm{R}$ is convex if : \newline
(a) $f(t \mathbf{x}_1+(1-t)\mathbf{x}_2)\leq t f(\mathbf{x}_1)+(1-t)f(\mathbf{x}_2)$ for any $t\in[0,1]$ \newline
(b) If $f(\mathbf{x})$ is first order differentiable, then it is convex if  $f(\mathbf{x}_j)\geq f(\mathbf{x}_i)+ \nabla_{\mathbf{x}_i} f(\mathbf{x}_i)^T( f(\mathbf{x}_j)-f(\mathbf{x}_i) )$\newline
where  in (a)(b) $\mathbf{x}_i\in\mathbbm{R}^N,i=1,2, j=1,2,j\neq i$.
\end{defn}

From our extensive simulations, we noticed that the original cost function in \eqref{eq17} is {\it not jointly convex} in $\mathbf{w}$ and $\mathbf{s}$. Nevertheless, it is not straightforward to prove / disprove joint convexity w.r.t. both $\mathbf{w}$ and $\mathbf{s}$ analytically. Consider, then, the following propositions:
\begin{prop} \label{propos1}
The objective function in \eqref{eq17} is individually convex w.r.t. $\mathbf{s}$, for any fixed but arbitrary $\mathbf{w}$
\end{prop}

\begin{proof}
Definition \ref{mydef1} cannot be directly invoked as the objective $g(\mathbf{s})= \mathbf{w}^H\mathbf{R_u}( \mathbf{s})\mathbf{w}:\mathbbm{C}^N\rightarrow \mathbbm{R}$ depends on the waveform $\mathbf{s}$, which is complex. Consider the following transformation\footnote{Ideally one must decompose the function into real and imaginary components (as accomplished subsequently), but due to Hermitian symmetry, real valued-ness e.t.c., we take this shortcut, here, instead}, $\mathbf{s}=\mathbf{D}\bar{\mathbf{s}}$ where $\bar{\mathbf{s}}\in\mathbbm{R}^{2N}=[\mbox{Re}\{\mathbf{s}\}^T,\mbox{Im}\{\mathbf{s}\}^T] ^T$ and $\mathbf{D}=[\mathbf{I}_N, j\mathbf{I}_N]\in\mathbbm{C}^{N\times2N}$. Now, we may  define an equivalent $g(\bar{\mathbf{s}}) :\mathbbm{R}^{2N}\rightarrow \mathbbm{R}$ to invoke the definition of convexity. We have to prove that,
\begin{align} \label{eq19}
&\mathbf{w}^H \left( \begin{aligned} &\mathbf{R_n}+\mathbf{R_i} \\
&+\sum\limits_{q=1}^Q \mathbf{\breve{B}_q}\begin{aligned} &(\mathbf{I}_P\otimes \mathbf{D}(t\bar{\mathbf{s}}_1+(1-t)\bar{\mathbf{s}}_2 ) )\mathbf{R}_{\gamma}^{pq} \\
&( \mathbf{I}_P\otimes (t\bar{\mathbf{s}}_1+(1-t)\bar{\mathbf{s}}_2 )^T\mathbf{D}^H ) \mathbf{\breve{B}_q}^H \end{aligned}  \end{aligned} \right) \mathbf{w} \nonumber \\
&\leq t \mathbf{w}^H \left(  \begin{aligned} &\mathbf{R_n}+\mathbf{R_i} \\
&+\sum\limits_{q=1}^Q \mathbf{\breve{B}_q}
(\mathbf{I}_P\otimes \mathbf{D}\bar{\mathbf{s}}_1 )\mathbf{R}_{\gamma}^{pq} ( \mathbf{I}_P\otimes \bar{\mathbf{s}}_1^T\mathbf{D}^H ) \mathbf{\breve{B}_q}^H
 \end{aligned} \right) \mathbf{w} \nonumber \\
 &+(1-t)\mathbf{w}^H \left( \begin{aligned} &\mathbf{R_n}+\mathbf{R_i} \\
 &+\sum\limits_{q=1}^Q
\mathbf{\breve{B}_q}(\mathbf{I}_P\otimes \mathbf{D}\bar{\mathbf{s}}_2 )\mathbf{R}_{\gamma}^{pq} ( \mathbf{I}_P\otimes \bar{\mathbf{s}}_2^T\mathbf{D}^H ) \mathbf{\breve{B}_q}^H
  \end{aligned} \right) \mathbf{w} 
\end{align}
where $t\in[0,1]$ and $\mathbf{\bar{s}_i}\in \mbox{dom}\{ g(\bar{\mathbf{s}})\},i=1,2$. After elementary algebra, the convexity  requirement in \eqref{eq19} transforms to:

\begin{align} \label{eq20}
\sum\limits_{q=1}^{Q} \mathbf{x}^H_{\bf q} \left(\mathbf{R}_{\gamma}^{pq}\otimes \mathbf{D}(\bar{\mathbf{s}}_1-\bar{\mathbf{s}}_2)(\bar{\mathbf{s}}_1-\bar{\mathbf{s}}_2)^T\mathbf{D}^H\right)\mathbf{x_q} \geq 0
\end{align}
where $\mathbf{x_q}\in\mathbbm{C}^{NP}:=\mathbf{\breve{B}_q}^H\mathbf{w}$. In other words, it is sufficient to show that iff \eqref{eq20} is true then \eqref{eq19} is also true and therefore convex. We notice immediately that \eqref{eq20} is a sum of Hermitian quadratic forms. Consider the matrix $\mathbf{R}_{\gamma}^{pq}\otimes \mathbf{D}(\bar{\mathbf{s}}_1-\bar{\mathbf{s}}_2)(\bar{\mathbf{s}}_1-\bar{\mathbf{s}}_2)^T\mathbf{D}^H$, we know that $\mathbf{R}_{\gamma}^{pq}\succeq \mathbf{0}$\footnote{Here $\succeq$ is the L\"{o}wner partial order \cite{horn1994} }, since it is  a covariance matrix and by definition atleast positive semi-definite (PSD). The other matrix, i.e. $\mathbf{D}(\bar{\mathbf{s}}_1-\bar{\mathbf{s}}_2)(\bar{\mathbf{s}}_1-\bar{\mathbf{s}}_2)^T\mathbf{D}^H$ is of course  rank-1 Hermitian, and is clearly PSD. From Lem.~\ref{lemma1}, it is straightforward to show that $\mathbf{R}_{\gamma}^{pq}\otimes \mathbf{D}(\bar{\mathbf{s}}_1-\bar{\mathbf{s}}_2)(\bar{\mathbf{s}}_1-\bar{\mathbf{s}}_2)^T \mathbf{D}^H\succeq \mathbf{0},\forall q$. Then from the definition of positive semi-definiteness, each of the $Q$ Hermitian quadratic forms in \eqref{eq20} is greater than zero, hence their sum is also greater than zero.
\end{proof}
\begin{prop} \label{propos2}
The objective function in \eqref{eq17} is individually convex w.r.t. $\mathbf{w}$, for any fixed but arbitrary $\mathbf{s}$.
\end{prop}
\begin{proof}
Given the guaranteed positive semi-definiteness of $\mathbf{R_u}(\mathbf{s})$, the proof is straightforward to demonstrate by invoking the convexity definition on the vector consisting of the real and imaginary parts of $\mathbf{w}$.
\end{proof}

In fact, Prop.~1, Prop.~2 may be sharpened to include strong convexity, which, as we will show subsequently is desired for the solutions to exist, see the note immediately after $\eqref{eq32}$.  For now, however, individual convexity is sufficient to proceed with our analysis.

\begin{remk}\label{remarkstapobj} 
({\it Characteristic of STAP objective}) The STAP objective in \eqref{eq17} has {\it at most one} minima for a fixed but arbitrary $\mathbf{w}\in\mathbbm{C}^{NML}$ but $\forall \mathbf{s} \in\mathbbm{C}^{N}$. Likewise, it has {\it at most one} minima for a fixed but arbitrary $\mathbf{s}\in\mathbbm{C}^{N}$ but $\forall \mathbf{w} \in\mathbbm{C}^{NML}$
\end{remk}
This is concluded readily from Prop.~\ref{propos1}, Prop.~\ref{propos2}, i.e. the individual convexity. An illustrative example is provided in Fig.~\ref{figlocalmin}.
\begin{figure} [htbp!]
\centering
\includegraphics[scale=0.5]{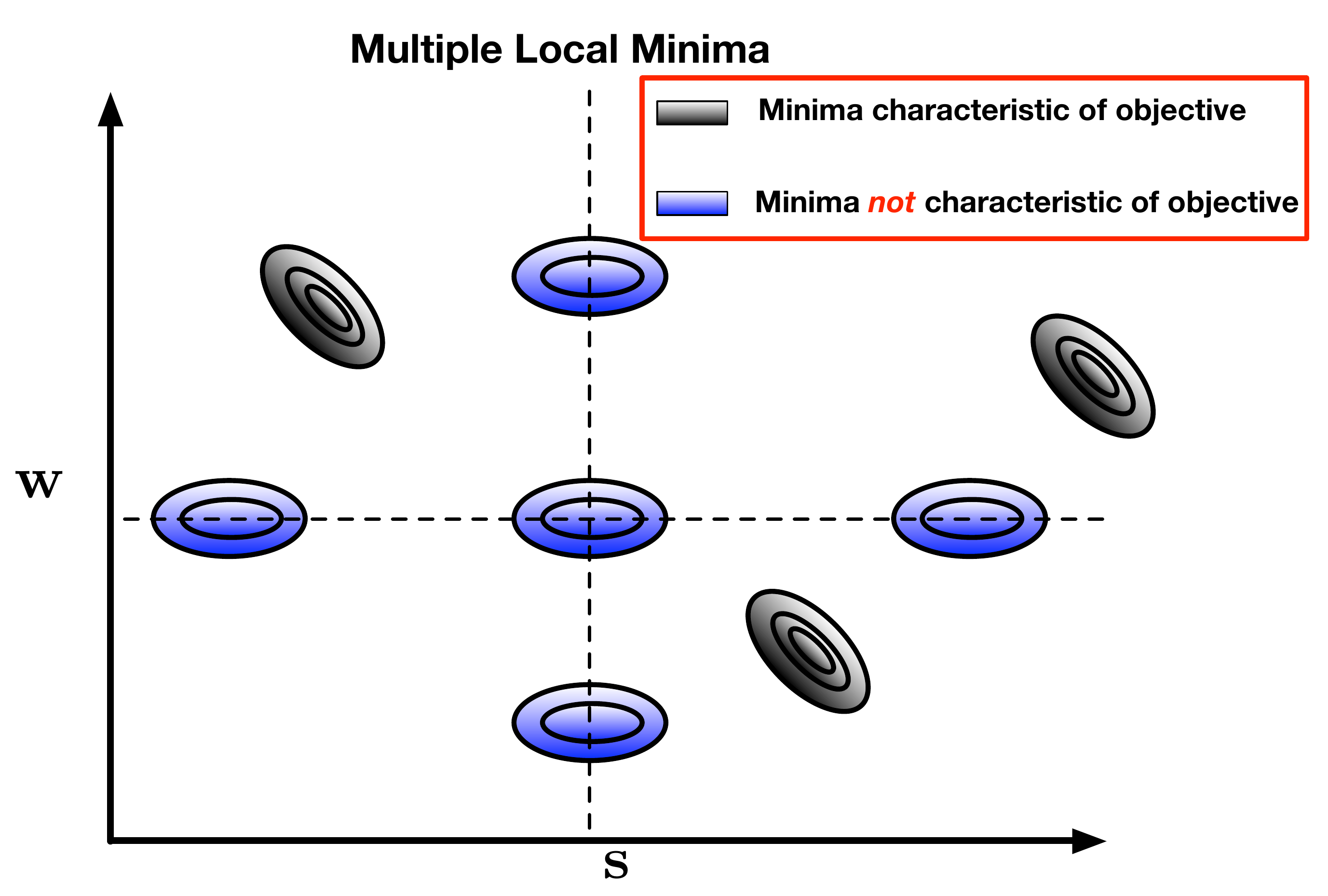}
\caption{An {\it illustrative} non-convex example with multiple local minima. Contours in black are characteristic of the objective. Contours in blue violate convexity in the $\mathbf{w}$, and $\mathbf{s}$ dimension individually, and are therefore {\it not characteristic} of the objective function.}
\label{figlocalmin}
\end{figure}

\subsection{Constrained alternating minimization}
Motivated from Prop.~\ref{propos1}, and Prop.~\ref{propos2}, we propose a constrained alternating minimization technique which is iterative. Before we present details on this technique, consider the following minimization problem, which optimizes $\mathbf{s}$, but for a fixed and arbitrary $\mathbf{w}$:
\begin{align}
\min\limits_{\mathbf{s}} \;\;\;\;\; &\mathbf{w}^H\mathbf{R_u}(\mathbf{s})\mathbf{w} \nonumber \\
\mbox{s. t. }\;\;\;\;\; & \mathbf{w}^H(\mathbf{v}(f_d)\otimes\mathbf{s}\otimes\mathbf{a}(\theta_t,\phi_t))=\kappa  \label{eq22} \\ 
\;\;\;\;\;\; & \mathbf{s}^H \mathbf{s}\leq P_o. \nonumber \nonumber
\end{align}
In \eqref{eq22}, the objective function could be rewritten as, 
\begin{align} \label{eq23}
\mathbf{w}^H\mathbf{R_u}(\mathbf{s})\mathbf{w}=&\mathbf{w}^H(\mathbf{R_n}+\mathbf{R_i})\mathbf{w} \\
+&\sum\limits_{q=1}^Q \mbox{Tr}\{ \mathbf{R}_{\gamma}^{pq} (\mathbf{I}_P\otimes\mathbf{s}^H)\mathbf{x_q}\mathbf{x}_{\bf q}^H (\mathbf{I}_P\otimes\mathbf{s})\} \nonumber.
\end{align}
In \eqref{eq23}, the trace operation  is further simplified as:
\begin{align} \label{eq24}
\mbox{Tr}\{ &\mathbf{R}_{\gamma}^{pq} (\mathbf{I}_P\otimes\mathbf{s}^H)\mathbf{x_q}\mathbf{x}_{\bf q}^H (\mathbf{I}_P\otimes\mathbf{s})\} \nonumber \\
&=\mbox{vec}\left( \left(\mathbf{R}_{\gamma}^{pq} (\mathbf{I}_P\otimes\mathbf{s}^H)\mathbf{x_q}\mathbf{x}_{\bf q}^H \right)^T \right)^T \mbox{vec}( \mathbf{I}_P\otimes\mathbf{s}) \nonumber 
\\
&=\mathbf{s}^H\mathbf{H}^T ( \mathbf{R}_{\gamma}^{pq}\otimes \mathbf{x_q}\mathbf{x}_{\bf q}^H)\mathbf{H}\mathbf{s}  \nonumber\\
&=\mathbf{s}^H\mathbf{Z_q}(\mathbf{w})\mathbf{s}
\end{align}
where $\mbox{vec}(\mathbf{I}_P\otimes\mathbf{s})=\mathbf{H}\mathbf{s}$, with $\mathbf{H}\in \mathbbm{R}^{P^2N \times N}=[\mathbf{H_1}^T,\mathbf{H_2}^T,\ldots,\mathbf{H_P}^T]^T$. The matrix $\mathbf{H_k}\in\mathbbm{R}^{PN \times N}, k=1,2,\ldots,P$ is further decomposed into $P$, $N\times N$ matrices, and is defined such that the $k$-th $N \times N$ matrix is $\mathbf{I}_N$ and the other $(N-1)$, $N\times N$ matrices are all equal to zero matrices. 

\begin{remk} \label{PSDremark}
(a) At the very least, $\sum\limits_{q=1}^Q\mathbf{Z_q}\succeq \mathbf{0}$.  (b) The matrix $\mathbf{Z_q}\succeq \mathbf{0}$ for $P<N$, always.  (c) However, it may be positive definite, i.e. $\mathbf{Z_q}\succ\mathbf{0}$  and hence $\sum\limits_{q=1}^Q\mathbf{Z_q}\succ\mathbf{0}$ for $P\geq N$ and for $\mathbf{R}_\gamma^{pq}\succ \mathbf{0}$.
\end{remk}

We note that (a) is readily implied from Prop.~\ref{propos1} since  a Hermitian quadratic form $\mathbf{x}^H\mathbf{B}\mathbf{x}$ is convex (strictly convex) iff $\mathbf{B}\succeq\mathbf{0}$ ( $\mathbf{B}\succ\mathbf{0}$). Since $\mathbf{R}_{\gamma}^{pq}\otimes \mathbf{x_q}\mathbf{x}_{\bf q}^H \succeq \mathbf{0}$ always ($\leq P$ non-zero eigenvalues and the rest are zeros) and that $P<N$, in other words, the transformation $\mathbf{H}^T(\mathbf{R}_{\gamma}^{pq}\otimes \mathbf{x_q}\mathbf{x}_{\bf q}^H)\mathbf{H}:\mathbbm{C}^{P^2 N\times N} \times \mathbbm{C}^{P^2 N\times N} \rightarrow \mathbbm{C}^{N\times N}$ and from the structure of $\mathbf{H}$, the result (b) is obvious. For (c), we know that $\mathrm{rank}(\mathbf{H})=N$, hence it could be shown after some tedious algebra that $\mathbf{Z_q}$ may be PD only when $P\geq N$ and that $\mathbf{R}_{\gamma}^{pq}$ is PD in the first place, also see for example \cite[pg. 399]{horn1994}.

Using \eqref{eq23} and \eqref{eq24}, the Lagrangian of \eqref{eq22} is readily cast as,
\begin{align} \label{eq25}
\mathcal{L}(\mathbf{s},\gamma_1,\gamma_2)&=\mathbf{w}^H( \mathbf{R_i+R_n})\mathbf{w}+\sum\limits_{q=1}^Q\mathbf{s}^H\mathbf{Z_q}(\mathbf{w})\mathbf{s} \\
&+\mbox{Re}\{ \gamma_1^{\ast} (\mathbf{w}^H\mathbf{G}\mathbf{s}-\kappa)\}+\gamma_2\mathbf{s}^H\mathbf{I}_N\mathbf{s}-\gamma_2P_o \nonumber
\end{align}
where $\gamma_1\in\mathbbm{C}$ and $\gamma_2\in\mathbbm{R}^{+}$ are the complex and real Lagrange parameters.

{\bf Lagrange Dual:} The Lagrange dual, denoted as $\mathcal{H}(\gamma_1,\gamma_2)=\inf\limits_{\mathbf{s}} \mathcal{L}( \mathbf{s},\gamma_1,\gamma_2)$. Since \eqref{eq25} consists of Hermitian quadratic forms and other linear terms of $\mathbf{s}$, we have $\mathcal{H}(\gamma_1,\gamma_2)=\mathcal{L}(\mathbf{s_o}(\gamma_1,\gamma_2),\gamma_1,\gamma_2)$, where $\mathbf{s_o}(\gamma_1,\gamma_2)$ is obtained by solving the first order optimality conditions, i.e.
\begin{align} \label{eq26}
\frac{\partial \mathcal{L}(\mathbf{s},\gamma_1,\gamma_2)}{\partial \mathbf{s}} ={\bf 0}
\end{align}
where, $\mathbf{0}$ is a  column vector of size $N$ and consists of all zeros. Further, in \eqref{eq26}, while taking the derivative the usual rules of complex vector differentiation apply, i.e. treat $\mathbf{s}^H$ independent of $\mathbf{s}$. The solution to \eqref{eq26} is readily obtained by differentiating \eqref{eq25}, and expressed as:
\begin{align} \label{eq27}
\mathbf{s_o}(\gamma_1,\gamma_2)=-\frac{\gamma_1}{2}\Bigl( \sum\limits_{q=1}^Q \mathbf{Z_q}( \mathbf{w})+\gamma_2\mathbf{I}_N\Bigr)^{-1}\mathbf{G}^H\mathbf{w}.
\end{align}
Using \eqref{eq27}, the dual $\mathcal{H}(\gamma_1,\gamma_2)$ is given by:
\begin{align} \label{eq28}
&\mathcal{H}( \gamma_1,\gamma_2)= \mathbf{w}^H( \mathbf{R_i+R_n})\mathbf{w} -\kappa\mbox{Re}\{ \gamma_1^\ast \}-\gamma_2P_o \nonumber \\
-&\frac{|\gamma_1|^2}{4} \mathbf{w}^H\mathbf{G} \Bigl( \sum\limits_{q=1}^Q \mathbf{Z_q}( \mathbf{w})+\gamma_2\mathbf{I}_N\Bigr)^{-1}\mathbf{G}^H\mathbf{w}.
\end{align} 
Equation \eqref{eq28} is further simplified by decomposing, $\gamma_1=\gamma_{1r}+j\gamma_{1i}$. In which case, we notice that \eqref{eq28} is quadratic in $\gamma_{1r},\gamma_{1i}$,and purely linear in $\lambda_2$. The Lagrange dual optimization is therefore, 
\begin{align} \label{eq29}
\max \limits_{\gamma_{1r},\gamma_{1i},\gamma_2} \;\;\;\; &\mathcal{H}(\gamma_{1r},\gamma_{1i},\gamma_2) \nonumber \\
\mbox{s. t } \;\;\;\; &\gamma_2\geq0.
\end{align}
Maximizing first w.r.t. $\gamma_{1r},\gamma_{1i}$, we have the solutions,
\begin{align*}
\bar{\gamma}_{1r}=\frac{-2\kappa}{\mathbf{w}^H\mathbf{G} \Bigl( \sum\limits_{q=1}^Q \mathbf{Z_q}( \mathbf{w})+\gamma_2\mathbf{I}_N\Bigr)^{-1}\mathbf{G}^H\mathbf{w}},\;\bar{\gamma}_{1i}=0.
\end{align*}
Substituting the above solutions into \eqref{eq28}, the Lagrange dual optimization problem and after ignoring an unnecessary additive constant, takes the form,
\begin{align} \label{eq30}
\max\limits_{\gamma_2} \;\; &\kappa^2 \Bigl( \mathbf{w}^H\mathbf{G} \Bigl( \sum\limits_{q=1}^Q \mathbf{Z_q}( \mathbf{w})+\gamma_2\mathbf{I}_N\Bigr)^{-1}\mathbf{G}^H\mathbf{w}\Bigr)^{-1}-\gamma_2P_o \nonumber \\
\mbox{s. t. } \;\;&\gamma_2\geq0
\end{align}
The associated Lagrangian for \eqref{eq30} is
\begin{align}\label{lagra1}
\mathcal{D}(\gamma_2,\gamma)=\frac{\kappa^2} {\mathbf{w}^H\mathbf{G}\mathbf{F}^{-1} \mathbf{G}^H\mathbf{w}} -\gamma_2P_o-\gamma_3\gamma_2
\end{align}
where $\mathbf{F}:=\sum\limits_{q=1}^Q \mathbf{Z_q}( \mathbf{w})+\gamma_2\mathbf{I}_N$. The first order optimality condition for the optimization \eqref{eq30} is given by:
\begin{align*} 
&\frac{\partial}{\partial \gamma_2} \bigl( \frac{\kappa^2} {\mathbf{w}^H\mathbf{G}\mathbf{F}^{-1} \mathbf{G}^H\mathbf{w}} \bigr)-P_o-\gamma_3 =0 \nonumber \\
\mbox{or } &\frac{-\kappa^2}{ ( \mathbf{w}^H\mathbf{G}\mathbf{F}^{-1} \mathbf{G}^H\mathbf{w})^{2}} \mathbf{w}^H\mathbf{G}\frac{\partial\mathbf{F}^{-1}}{\partial\gamma_2}\mathbf{G}^H\mathbf{w}-P_o -\gamma_3 =0 \nonumber \\
\mbox{or } &\frac{\kappa^2}{ ( \mathbf{w}^H\mathbf{G}\mathbf{F}^{-1} \mathbf{G}^H\mathbf{w})^{2}} \mathbf{w}^H\mathbf{G}\bigl( \mathbf{F}^{-1}\frac{\partial\mathbf{F}}{\partial\gamma_2} \mathbf{F}^{-1} \bigr)\mathbf{G}^H\mathbf{w}-P_o -\gamma_3 =0 \nonumber 
\end{align*}
where $\gamma_3$ is the Lagrange multiplier associated with the Lagrangian \eqref{lagra1}, and we also have $\tfrac{\partial \mathbf{F}}{\partial \gamma_2}=\mathbf{I}_N$. The complementary slackness and constraint qualifier for \eqref{eq30} i.e. $\gamma_3\gamma_2=0$ and $\gamma_2\geq0$ form the rest of the equations comprising the KKT conditions. It is now readily shown that the solution to \eqref{eq30} is given by
\begin{align} \label{eq31}
&\bar{\gamma}_2=\max [0,\gamma_2] \\
&\gamma_2 \mbox{ solves } \gamma_2\left( \kappa^2\mathbf{w}^H\mathbf{G}\mathbf{F}^{-2}\mathbf{G}^H\mathbf{w}-P_o( \mathbf{w}^H\mathbf{G}\mathbf{F}^{-1}\mathbf{G}^H\mathbf{w})^2\right)=0 \nonumber.
\end{align}

\begin{prop} \label{propos3}
The parameter $\bar{\gamma}_2=0$ solves \eqref{eq31}.
\end{prop}

\begin{proof}
The spectral theorem for Hermitian matrices, allows for a decomposition, $\mathbf{F}=\mathbf{E}(\boldsymbol{\Lambda}+\gamma_2\mathbf{I}_N)\mathbf{E}^H$. The matrix $\boldsymbol{\Lambda}$ is a diagonal matrix comprising eigenvalues in descending order, whereas, $\mathbf{E}$ is unitary and whose columns are the corresponding eigenvectors of $\mathbf{F}$. For ease of exposition, denote $\mathbf{z}\in\mathbbm{C}^{N}:=\mathbf{E}^H\mathbf{G}^H\mathbf{w}$, then assume a function $f(\gamma_2):\mathbbm{R}^{+}\rightarrow\mathbbm{R}$, expressed as
\begin{align} \label{lagra2}
f(\gamma_2)&:=\kappa^2\mathbf{w}^H\mathbf{G}\mathbf{F}^{-2}\mathbf{G}^H\mathbf{w}-P_o( \mathbf{w}^H\mathbf{G}\mathbf{F}^{-1}\mathbf{G}^H\mathbf{w})^2 \nonumber \\
&=\sum\limits_{n=1}^N  \kappa^2\frac{|z_n|^2}{(d_n+\gamma_2)^2} -P_o\left( \sum\limits_{n=1}^N \frac{|z_n|^2}{d_n+\gamma_2}\right)^2
\end{align}
where $z_n,d_n$ are the $n$-th elements of $\mathbf{z}$, and the $n$-th eigenvalue in $\boldsymbol{\Lambda}$. We analyze $f(\gamma_2)$ and $\gamma_2f(\gamma_2)$ in detail. The following (behavior at $0$ and $\infty$) are readily observed
\begin{subequations} \label{lagra3}
\begin{align}
\lim_{\gamma_2\rightarrow \infty}f(\gamma_2)&=f(\infty) =0 \label{lagra32}\\
\lim_{\gamma_2\rightarrow 0}f(\gamma_2)&=\sum\limits_{n=1}^N \kappa^2\frac{|z_n|^2}{d_n^2} -P_o\left(  \sum\limits_{n=1}^N\frac{|z_n|^2}{d_n}\right)^2=f(0) \label{lagra31}
\end{align}
\end{subequations}
 Furthermore, it is seen that
\begin{equation}  \label{lagra4}
\begin{aligned}
&\lim_{\gamma_2\rightarrow\infty}\gamma_2f(\gamma_2)=\lim_{\gamma_2\rightarrow\infty}\frac{f(\gamma_2)}{1/\gamma_2}=\lim_{\gamma_2\rightarrow\infty}\frac{\frac{\mathrm{d}f(\gamma_2)}{\mathrm{d}\gamma_2}}{(-1/\gamma_2^2)} =0
\end{aligned}
\end{equation}
Moreover, consider $f(\gamma_2)=h_1(\gamma_2)-h_2(\gamma_2)=0$, where $h_1(\gamma_2)=\kappa^2\sum\limits_{n=1}^N\frac{f_n^2(\gamma_2)}{|z_n|^2},h_2(\gamma_2)=P_o (\sum\limits_{n=1}^N f_n(\gamma_2))^2$, where $f_n(\gamma_2)=\frac{|z_n|^2}{d_n+\gamma_2}$. Note that $fn(\gamma_2)\downarrow, n=1,2,\ldots,N$ and that $h_i(\gamma_2)\downarrow,i=1,2$, i.e. decreasing functions w.r.t. $\gamma_2 \in[0,\infty)$. Then equation$f(\gamma_2)=0$ implies that 
\begin{equation}\label{lagra41}
\begin{aligned}
\sum\limits_{n=1}^N  \kappa^2\frac{|z_n|^2}{(d_n+\gamma_2)^2} -P_o\left( \sum\limits_{n=1}^N \frac{|z_n|^2}{d_n+\gamma_2}\right)^2=0 \\
\mbox{or } \sum\limits_{n=1}^N (\frac{\kappa^2}{|z_n|^2}-P_o )f_n^2(\gamma_2)=2\sum\limits_{n_1}\sum\limits_{\substack{n_2\\ n_2\neq n_1}} f_{n_1}(\gamma_2)f_{n_2}(\gamma_2)
\end{aligned}
\end{equation}
where $(n_1,n_2)\in (1,2,\ldots,N)$. Recall that $d_n\neq 0 \forall n$, $d_n\geq d_{n+1}, n=1,2,\ldots, N$, and $ |z_n|\neq 0 \forall n$.  A solution to \eqref{lagra41} for $\gamma_2 \in[0,\infty)$ is readily derived in the trivial  case, for example when $f_{n_1}(\gamma_2)=f_{n_2}(\gamma_2)$, $P_o\neq \kappa^2$, and for $|z_n|$ to be some arbitrary constant for all $n$. For $P_o\geq \kappa$ it may now be shown numerically that a solution to \eqref{lagra41} for $\gamma_2 \in [0,\infty)$ does not exist.
\end{proof}
In fact, our extensive numerical simulations reveal that in general and assuming $P_o\geq \kappa$ and for $\gamma_{21}\leq \gamma_{22}$
\begin{align} \label{lagra51}
\begin{cases}
f(\gamma_{21})\geq f(\gamma_{22}) \mbox{ if } f(0)>0 \\
f(\gamma_{21})\leq f(\gamma_{22}) \mbox{ if } f(0)<0 
\end{cases} \gamma_{21} \mbox{ and } \gamma_{22} \in[0,\infty).
\end{align}
That is, $f(\gamma_2)$ is monotonic. From the above arguments, therefore, $\gamma_2f(\gamma_2)=0$ implies that  $\gamma_2=0$. Alternatively nevertheless, a solution to \eqref{eq31} may be found numerically and is computationally cheap. 

{\bf Note:} {\it (Inactive power constraint)} It is noted that trivially $\bar{\gamma}_2=0$ {\it may always} be chosen as a solution with suitable choices of the free parameter $P_o$. This implies that the power constraint is always satisfied and hence is an inactive constraint in the corresponding Lagrangian.

A graphical behavior of $h_i(\gamma_2),i=1,2$ and thus the behavior of $f(\gamma_2)$ is seen from Fig.~\ref{figlagrange}. Using Prop.~\ref{propos3}, the waveform design solution is  unique, a function of $\mathbf{w}$ and expressed as,
\begin{align} \label{eq32}
\mathbf{s}_o(\mathbf{w})=\frac{\kappa \Bigl( \sum\limits_{q=1}^Q \mathbf{Z_q}( \mathbf{w})\Bigr)^{-1}\mathbf{G}^H\mathbf{w} }{\mathbf{w}^H\mathbf{G} \Bigl( \sum\limits_{q=1}^Q \mathbf{Z_q}( \mathbf{w})\Bigr)^{-1}\mathbf{G}^H\mathbf{w}}.
\end{align}

{\bf Note:} ({\it Strong convexity}) To compute the constrained alternating minimization solutions, the respective matrices in \eqref{eq32}, \eqref{weightcomp} must be invertible, implying strong convexity individually  w.r.t. $\mathbf{w}$, $\mathbf{s}$, respectively. This directly necessitates, $\lambda_{\min}\Bigl( \sum\limits_{q=1}^Q \mathbf{Z_q}( \mathbf{w})\Bigr)\neq 0$ and $\lambda_{\min} (\mathbf{R_u}(\mathbf{s}) )\neq 0$, and hence also, positive definiteness of these matrices.

The alternating minimization algorithm is now succinctly stated in Table \ref{table1}.

\begin{figure} [tbp!]
\centering
\includegraphics[scale=0.5]{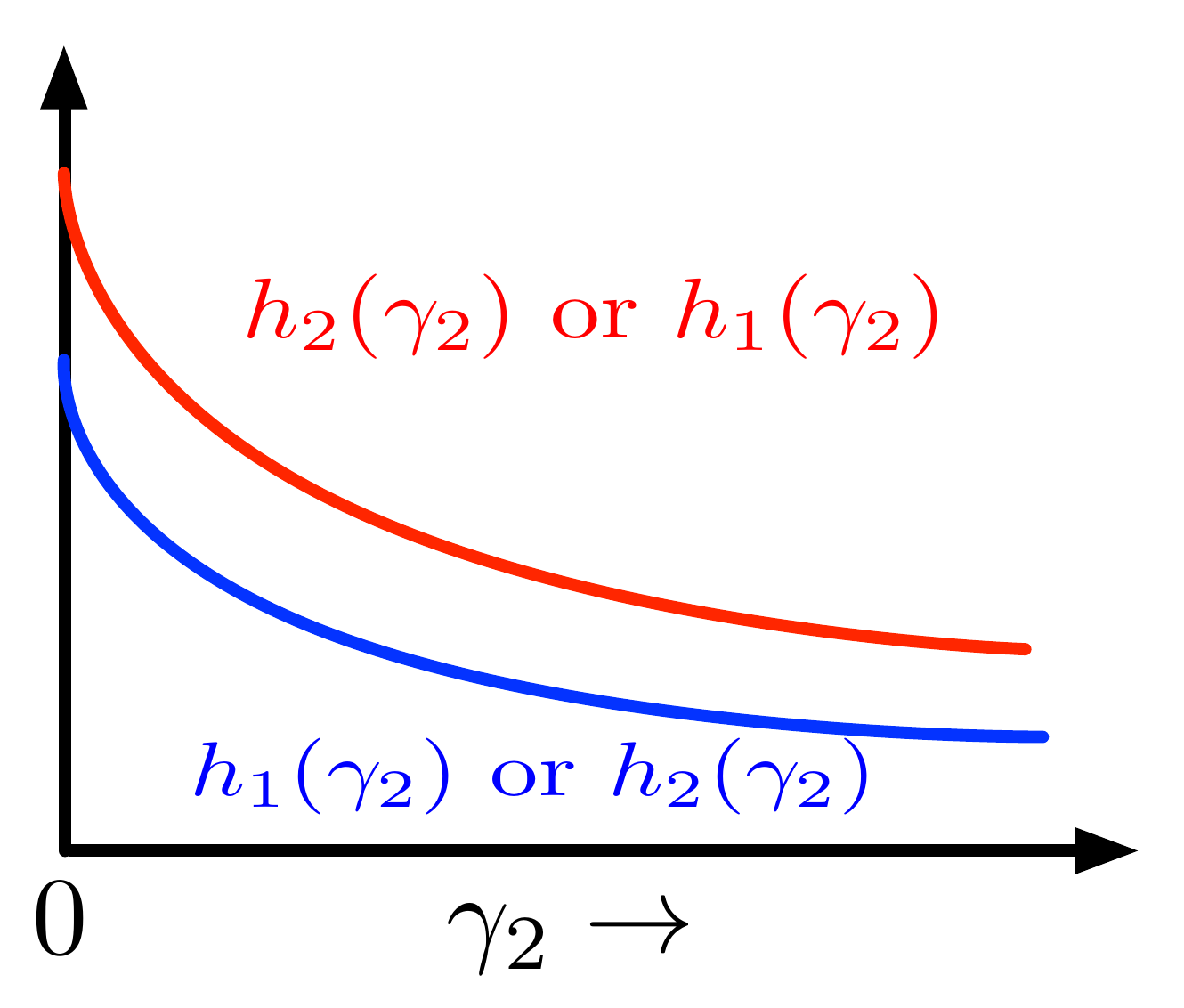}
\caption{Two cases are presented assuming $P_o\geq \kappa$. (a) Blue: $h_1(\gamma_2)$, Red: $h_2(\gamma_2)$ and therefore $f(\gamma_2)$ is decreasing, (b)  Blue: $h_2(\gamma_2)$, Red: $h_1(\gamma_2)$ and therefore $f(\gamma_2)$ is increasing. The blue and red curves intersect at $\infty$.}
\label{figlagrange}
\end{figure}

\begin{remk}\label{dualprop}

({\it Strong duality}) The optimal value of the lagrange dual problem is given by
\begin{equation*} 
\mathbf{w}^H( \mathbf{R_i+R_n})\mathbf{w} +\frac{\kappa^2}{\mathbf{w}^H\mathbf{G} \Bigl( \sum\limits_{q=1}^Q \mathbf{Z_q}( \mathbf{w})\Bigr)^{-1}\mathbf{G}^H\mathbf{w}}.
\end{equation*}
It is therefore trivial to show that the duality gap between \eqref{eq22} and \eqref{eq29} is zero.
 In other words, strong duality holds between the primal in \eqref{eq22} and the dual in \eqref{eq29}. From Slaters condition \cite{Boyd2004} the sufficient condition to ensure strong duality is the existence of \eqref{eq32}, i.e. the inverse of $\sum\limits_{q=1}^Q \mathbf{Z_q}( \mathbf{w})$ exists (see note below), and that the solution in \eqref{eq32} satisfies the power constraint.
\end{remk}

{\bf Note:} ({\it Lower bound on $Q$}) Since $\mathrm{rank}( \sum\limits_{q=1}^Q \mathbf{Z_q}( \mathbf{w}))\leq \sum\limits_{q=1}^Q\mathrm{rank}(  \mathbf{Z_q}( \mathbf{w}))$, {\it assume the worst case} $P=1$, then we have that $\mathrm{rank}(\mathbf{Z_q})=1$. Therefore for $Q$ distinct  (different spatial signature and Doppler) clutter patches, $Q\geq N$ ensures invertibility of $\sum\limits_q \mathbf{Z_q}$.

\begin{table}[htbp!] 
\centering
\caption{Constrained alternating minimization for waveform adaptive radar STAP}
\begin{tabular}{|p{3.3in}|} 
\hline
\begin{enumerate} \label{table1}
\item {\it Initialize}: Start with  an initial waveform design, defined as $\mathbf{s}_o^{(0)}$, set counter $k=1$
\item { \it Filter design}: Design the optimal filter weight vector, 
$\mathbf{w}_{o}^{(k)}=\mathbf{w}_o(\mathbf{s}_o^{(k-1)})$, where \eqref{weightcomp} is used to compute $\mathbf{w}_o( \cdot)$.
\item{ \it Waveform design}: Design the updated waveform $\mathbf{s}_o^{(k)}=\mathbf{s}_o(\mathbf{w}_o^{(k)})$, where \eqref{eq32} is used to compute $\mathbf{s}_o(\cdot)$.
\item {\it Check:} If convergence is achieved, exit, else $k=k+1$, go back to step-2.
\end{enumerate} \\
\hline
\end{tabular}
\end{table}

\subsubsection{Convergence, performance guarantees, and other properties} Denote $(\mathbf{w}_k,\mathbf{s}_k)$ as the sequence of iterates of the algorithm in Table \ref{table1} and define $g(\mathbf{w}_k,\mathbf{s}_k):=\mathbf{w}_k\mathbf{R_u}(\mathbf{s}_k)\mathbf{w}_k^H$, then for $k=1,2,\ldots$
\begin{align} \label{eq33}
\cdots g(\mathbf{w}_k,\mathbf{s}_{k-1} )\geq g(\mathbf{w}_k,\mathbf{s}_k ) \geq g( \mathbf{w}_{k+1},\mathbf{s}_k)\cdots.
\end{align}
Moreover, since at least $\mathbf{R_u}(\mathbf{s})\succeq \mathbf{0}$, i.e. PSD $\forall\mathbf{s}$, we have that $g(\mathbf{w},\mathbf{s})\geq0,\forall\mathbf{w}$. Therefore each of the individual terms in \eqref{eq33} are {\it lower bounded} by zero, in other words $g(\mathbf{w}_{k_1},\mathbf{s}_{k_2})\geq0,k_1=k,\mbox{ or }k+1$ and $k_2=k,\mbox{ or }k+1$, for $k=1,2,\ldots$ .
\begin{prop} \label{propos4}
Iff the iterates $(\mathbf{w}_k,\mathbf{s}_k)$ of the constrained alternating minimization exist, then $\lim\limits_{k\rightarrow\infty}\mathbf{g ( \mathbf{w}_k,\mathbf{s}_k)}$ is finite.
\end{prop}
\begin{proof}
The non-increasing property in \eqref{eq33}, and since each term in \eqref{eq33} is lower bounded, straightforward application of the monotone convergence  theorem to the sequence, $\{\mathbf{g ( \mathbf{w}_k,\mathbf{s}_k)}\}$,completes the proof.
\end{proof}
We note that convergence to a finite limit as evidenced from Prop.~\ref{propos4} is indeed dependent on the constraints via the existence of the iterates $(\mathbf{w}_k,\mathbf{s}_k)$. This however does not imply convergence of the sequence $\{(\mathbf{w}_k,\mathbf{s}_k) \}$, for which, consider the following.
\begin{remk}\label{remark1}
The alternating minimization is a special case of the block Gauss-Siedel and block co-ordinate descent (BCD) algorithm with block size equal to two \cite{Grippo2000,Luo1992}.  
\end{remk}
\begin{defn} \label{deflimit} ({\it Convergence in $\mathbbm{R}^{N}$}) A sequence $\{ \mathbf{x}_k\} \in\mathbbm{R}^{N} ,k=1,2,\ldots$ is said to converge to $\tilde{\mathbf{x}}$, a limit point, if, $\forall \epsilon>0,\;\; \exists K\in{\mathbbm{N}}:\;||\mathbf{x}_k-\tilde{\mathbf{x}}|| \leq\epsilon,\; k>K$.
\end{defn}
\begin{lem} \label{altminlemma}
(Constrained alternating minimization lemma) Assume that a function $g(\mathbf{z}):\mathbbm{R}^{2N}\rightarrow\mathbbm{R},\mathbf{z}=[\mathbf{x}^T,\mathbf{y}^T]^T$ is continuously differentiable over a closed nonempty convex set, $\mathcal{A}=\mathcal{A}_1\times\mathcal{A}_2$. Also, suppose the solution to the constrained optimization problems, $\min \limits_{\mathbf{x}\in\mathcal{A}_1} g(\mathbf{x},\mathbf{y} )$ and $\min \limits_{\mathbf{y}\in\mathcal{A}_2} g(\mathbf{x},\mathbf{y} )$ are uniquely attained. Let $\{ \mathbf{z}_k \}$ be the sequence generated by this algorithm, then every limit point of this sequence is also a stationary point.
\end{lem}
\begin{proof}
The proof in \cite[Prop.~2.7.1]{Bertsekas1999} follows immediately to the alternating minimization assuming two blocks. Also see \cite{Grippo2000}, where the convergence of the two block BCD was analyzed. 
\end{proof}
The above Lem.~\ref{altminlemma} discusses convergence of the constrained alternating minimization.This lemma can be applied by decomposing our problem into its real equivalent along-with real and imaginary decomposition of $\mathbf{w},\mathbf{s}$, and assuming the our constraint set $\mathcal{A}=\mathcal{A}_1\times\mathcal{A}_2$ is closed convex and the minimizers are unique. 
The necessary condition of a unique minimizer \cite{Zangwill1967} at each step is not obvious, but \cite{Powell1973} showed that in the absence of this assumption the algorithm cycles endlessly around a particular objective value \cite{Bertsekas1999}. Further the algorithm provides limit points which are not stationary points \cite{Grippo2000}. 
To discuss the characteristics of the limits points at convergence, consider the remark, presented next.

\begin{remk} \label{limitpointremk} ({\it Characterizing the solutions at convergence})
If $(\mathbf{w}_\star,\mathbf{s}_\star)$ are the limit points of the sequence $\{(\mathbf{w}_k,\mathbf{s}_k)\}$. Then, $(\mathbf{w}_\star,\mathbf{s}_\star)$ is a local minima, i.e. by definition $g(\mathbf{w}_\star,\mathbf{s}_\star)\leq g(\mathbf{w},\mathbf{s})$,$\exists \epsilon>0$ with $(\mathbf{w},\mathbf{s}):\,\sqrt{|| \mathbf{w}-\mathbf{w}_\star||^2+|| \mathbf{s}-\mathbf{s}_\star||^2}\leq \epsilon$. Further, $(\mathbf{w}_\star,\mathbf{s}_\star):g(\mathbf{w}_\star,\mathbf{s}_\star)\leq g(\mathbf{w}_\star,\mathbf{s}),\,\forall \mathbf{s}\in\mathcal{A}_2\mbox{ and } g(\mathbf{w}_\star,\mathbf{s}_\star)\leq g(\mathbf{w},\mathbf{s}_\star),\,\forall \mathbf{w}\in\mathcal{A}_{1}$.
\end{remk}

The first statement in Rem.~\ref{limitpointremk} directly results from  from the stationarity condition as  given  in Lem.~\ref{altminlemma} and also since the objective is non-convex. The second statement in Rem.~\ref{limitpointremk} arises from the individual convexity in $\mathbf{w}$ and $\mathbf{s}$ as shown in Prop.~\ref{propos1}, Prop.~\ref{propos2}. We note readily from Rem.~\ref{remarkstapobj}, that unfortunately there is nothing special or strong about $(\mathbf{w}_\star,\mathbf{s}_\star)$ except the fact that they are local minima. It is well known that global extrema (minima or maxima) are attained only when the objective is either convex or concave. For a problem similar to ours and where the alternating minimization was applied, see \cite[pg.3537]{vaidyanathan2009} the authors state that their algorithm produces limit points which are stronger than local maxima, in our opinion this conclusion is suspect. They further claim that their algorithm produces global extrema in their filter design and waveform dimensions individually, which leads us to believe that their objective is concave, although this was never proved in \cite{vaidyanathan2009}. In our opinion, Rem.~\ref{remarkstapobj} is also relevant to their objective by replacing minima by maxima, and hence  we do not believe that the limit points produced by their algorithm are stronger than local extrema. 

To derive the upper and lower bounds on $g(\mathbf{w}_k,\mathbf{s}_k)-g(\mathbf{w}_{k+1},\mathbf{s}_k)$, the following well known lemmas are useful.
\begin{lem}\label{lemma2}
For any Hermitian matrix, $\mathbf{A}\in\mathbbm{C}^{N\times N}$ and any arbitrary vector $\mathbf{x}\in\mathbbm{C}^{N\times N}$ , we always have $\lambda_{\min}(\mathbf{A})||\mathbf{x}||^2\leq \mathbf{x}^H\mathbf{Ax}\leq\lambda_{\max}(\mathbf{A})||\mathbf{x}||^2$, where $\lambda_{\min}(\mathbf{A})$ and $\lambda_{\max}(\mathbf{A})$ are the min. and max. eigenvalues of matrix $\mathbf{A}$, respectively.
\end{lem}
\begin{proof}
The proof can be seen in \cite{horn1994}, and is in fact fundamental to the Rayleigh-Ritz theorem.
\end{proof}
\begin{lem}\label{lemma3}
For any two Hermitian matrices, $\mathbf{A,B}$, both in $\mathbbm{C}^{N\times N}$,
\begin{align*}
\sum\limits_{i=1}^N \lambda_i(\mathbf{A}) \lambda_{N-i+1}(\mathbf{B})\leq \mathrm{Tr}\{ \mathbf{AB}\}\leq\sum\limits_{i=1}^N \lambda_i(\mathbf{A})\lambda_i(\mathbf{B})
\end{align*}
where $\lambda_i(\cdot)\geq \lambda_{i+1}(\cdot)$, $i=1,2,\ldots,N$.
\end{lem}
\begin{proof}
See \cite[Lemma. II. I]{Lasserre1995} for a proof.
\end{proof}

Consider $g(\mathbf{w}_k,\mathbf{s}_k)$, we have
\begin{align}\label{eq34}
g(\mathbf{w}_k,\mathbf{s}_k)&=\mathbf{w}_k^H\mathbf{R_u}(\mathbf{s}_k)\mathbf{w}_k \nonumber \\
&=(\mathbf{w}_k-\mathbf{w}_{k+1} +\mathbf{w}_{k+1}  )^H \mathbf{R_u}(\mathbf{s}_k)(\mathbf{w}_k-\mathbf{w}_{k+1}+\mathbf{w}_{k+1} ) \nonumber\\
&=(\mathbf{w}_k-\mathbf{w}_{k+1} )^H \mathbf{R_u}(\mathbf{s}_k)(\mathbf{w}_k-\mathbf{w}_{k+1}) \\
&+ \mathbf{w}_{k+1} ^H \mathbf{R_u}(\mathbf{s}_k)\mathbf{w}_{k+1}+\textrm{Re}\{  (\mathbf{w}_k-\mathbf{w}_{k+1} )^H \mathbf{R_u}( \mathbf{s}_k) \mathbf{w}_{k+1} \} \nonumber
\end{align}
Moreover since the square root decomposition exists i.e., $\mathbf{R_u}( \cdot)=\mathbf{R}_{\bf u}^{1/2}( \cdot) \mathbf{R}_{\bf u}^{1/2}(\cdot)$, then application of the Cauchy-Schwartz inequality produces,
\begin{align} \label{eq35}
&\textrm{Re}\{  (\mathbf{w}_k-\mathbf{w}_{k+1} )^H \mathbf{R_u}( \mathbf{s}_k) \mathbf{w}_{k+1} \} \leq \\
&\sqrt{(\mathbf{w}_k-\mathbf{w}_{k+1} )^H \mathbf{R_u}(\mathbf{s}_k)(\mathbf{w}_k-\mathbf{w}_{k+1}) } \sqrt{\mathbf{w}_{k+1} ^H \mathbf{R_u}(\mathbf{s}_k)\mathbf{w}_{k+1}} \nonumber
\end{align}
Using \eqref{eq35} in \eqref{eq34} and since $\mathbf{R_u}(\cdot)$ is PSD, we can show that 
$
g(\mathbf{w}_k,\mathbf{s}_k)-g(\mathbf{w}_{k+1},\mathbf{s}_k)\leq (\mathbf{w}_k-\mathbf{w}_{k+1} )^H \mathbf{R_u}(\mathbf{s}_k)(\mathbf{w}_k-\mathbf{w}_{k+1})
$. Further using \eqref{eq33}, we have the following upper and lower bounds
\begin{align} \label{eq36}
0&\leq g(\mathbf{w}_k,\mathbf{s}_k)-g(\mathbf{w}_{k+1},\mathbf{s}_k) \nonumber \\
&\leq(\mathbf{w}_k-\mathbf{w}_{k+1} )^H \mathbf{R_u}(\mathbf{s}_k)(\mathbf{w}_k-\mathbf{w}_{k+1})
\end{align}
We notice immediately, that at convergence $(\mathbf{w}_k-\mathbf{w}_{k+1} )^H \mathbf{R_u}(\mathbf{s}_k)(\mathbf{w}_k-\mathbf{w}_{k+1})\rightarrow 0$ since $\mathbf{w}_k\rightarrow\mathbf{w}_{k+1}$. Other bounds as in \eqref{eq36} can be readily derived. From Lem.~\ref{lemma2}, we can show that
\begin{align}\label{eq37}
&\leq \lambda_{\min}(\mathbf{R_u}(\mathbf{s}_k) )|| \mathbf{w}_k||^2-\lambda_{\max}(\mathbf{R_u}(\mathbf{s}_k) )|| \mathbf{w}_{k+1}||^2\nonumber \\
&g(\mathbf{w}_k,\mathbf{s}_k)-g(\mathbf{w}_{k+1},\mathbf{s}_k) \\
&\leq \lambda_{\max}(\mathbf{R_u}(\mathbf{s}_k) )|| \mathbf{w}_k||^2-\lambda_{\min}(\mathbf{R_u}(\mathbf{s}_k) )|| \mathbf{w}_{k+1}||^2 \nonumber.
\end{align}

Consider the following.
\begin{lem} \label{lemma4}
If $\mathbf{x}$, $\mathbf{y}$ are arbitrary but distinct complex vectors of size $N$ and let $\mathbf{A}:=\mathbf{xx}^H-\mathbf{yy}^H$, then,
(a)  matrix $\mathbf{A}$ has exactly two real non-zero eigenvalues, the rest $N-2$ eigenvalues are all zeros,
(b) of the two real and non-zero eigenvalues one is always positive and the other is always negative, and 
(c) if the $\mathbf{x}$, $\mathbf{y}$ are not distinct, i.e. $\mathbf{y}=\beta \mathbf{x}$, $\beta\in\mathbbm{C}$, then there exists only one non-zero eigenvalue, $(|1-|\beta|^2|) ||\mathbf{x}||^2$and  the rest $N-1$ eigenvalues are purely zeroes.
\end{lem}
\begin{proof}
First of all we notice $\mathbf{A}$ is Hermitian and hence its eigenvalues are real. The proof for (a) is obvious given the fact that $\mathbf{A}$ is a sum of two distinct outer products. In other words, $\mathrm{rank}(\mathbf{A})=2$, for all $\mathbf{y}\neq\beta\mathbf{x}$ .

Now we know that
\begin{align*}
\mathrm{Tr}\{ \mathbf{A}\}&=\lambda_1+\lambda_2=\mathbf{x}^H\mathbf{x}-\mathbf{y}^H\mathbf{y} \\
\mathrm{Tr}\{\mathbf{AA}^H\}&=\lambda_1^2+\lambda_2^2=||\mathbf{x}||^4+||\mathbf{y}||^4-2|\mathbf{x}^H\mathbf{y}|^2
\end{align*}
where $\lambda_i,i=1,2$ are the two non zero eigenvalues of $\mathbf{A}$. The above  set of equations can be reduced to  a quadratic in any one eigenvalue. It can be shown that the only two possible solutions are then 
\begin{equation} \label{eq38}
 \begin{aligned} \lambda_1&=\frac{||\mathbf{x}||^2-||\mathbf{y}||^2}{2}
\left(  1+ \sqrt{ 1-4\frac{|\mathbf{x}^H\mathbf{y}|^2 -||\mathbf{x}||^2||\mathbf{y}||^2}{(||\mathbf{x}||^2-||\mathbf{y}||^2)^2} } \right) \\
\lambda_2&=\frac{||\mathbf{x}||^2-||\mathbf{y}||^2}{2}
\left(  1- \sqrt{ 1-4\frac{|\mathbf{x}^H\mathbf{y}|^2 -||\mathbf{x}||^2||\mathbf{y}||^2}{(||\mathbf{x}||^2-||\mathbf{y}||^2)^2} } \right) \end{aligned}
\end{equation}
Since $\lambda_i,i=1,2$ are purely real we have, $1-4\frac{|\mathbf{x}^H\mathbf{y}|^2 -||\mathbf{x}||^2||\mathbf{y}||^2}{(||\mathbf{x}||^2-||\mathbf{y}||^2)^2} \geq0$ and from Cauchy Schwarz inequality, we also have that $|\mathbf{x}^H\mathbf{y}|^2 -||\mathbf{x}||^2||\mathbf{y}||^2\leq0 $. Using these two facts, consider two specific cases, both of which are shown easily from elementary algebra, 
\begin{equation} \label{eq39}
\begin{cases}
\lambda_1>0,\lambda_2 <0, & \mbox{ if } ||\mathbf{x}||^2-||\mathbf{y}||^2 \geq0 \\
\lambda_1<0,\lambda_2>0, & \mbox{ if } ||\mathbf{x}||^2-||\mathbf{y}||^2 <0
\end{cases}.
\end{equation}
When $||\mathbf{x}||^2-||\mathbf{y}||^2=0$, it is easily seen that $\lambda_1=\sqrt{||\mathbf{x}||^2||\mathbf{y}||^2-|\mathbf{x}^H\mathbf{y}|^2 } >0$, $\lambda_2=-\lambda_1<0$. We also note immediately from \eqref{eq38} that when, $\mathbf{y}=\beta \mathbf{x}$, $\lambda_1 =(1-|\beta|^2)||\mathbf{x}||^2$, $\lambda_2=0$. This completes the proof.
\end{proof}

It is readily shown that $g(\mathbf{w}_k,\mathbf{s}_k)-g(\mathbf{w}_{k+1},\mathbf{s}_k)=\mathrm{Tr}\{\mathbf{R_u}(\mathbf{s}_k)  ( \mathbf{w}_k\mathbf{w}_k^H-\mathbf{w}_{k+1} \mathbf{w}_{k+1}^H) \}$. Therefore, from Lem.~\ref{lemma3}, and Lem.~\ref{lemma4}, we have,
\begin{align} \label{eq40}
&\leq \lambda_{\max}\big( \mathbf{R_u}( \mathbf{s}_k) \big)\lambda_{-}(\mathbf{w}_k\mathbf{w}_k^H-\mathbf{w}_{k+1} \mathbf{w}_{k+1}^H) \nonumber\\
&+\lambda_{\min}\big( \mathbf{R_u}( \mathbf{s}_k) \big)\lambda_{+}(\mathbf{w}_k\mathbf{w}_k^H-\mathbf{w}_{k+1} \mathbf{w}_{k+1}^H) \nonumber \nonumber \\
&g(\mathbf{w}_k,\mathbf{s}_k)-g(\mathbf{w}_{k+1},\mathbf{s}_k) \\
 &\leq \lambda_{\max}\big( \mathbf{R_u}( \mathbf{s}_k) \big)\lambda_{+}(\mathbf{w}_k\mathbf{w}_k^H-\mathbf{w}_{k+1} \mathbf{w}_{k+1}^H) \nonumber\\
&+\lambda_{\min}\big( \mathbf{R_u}( \mathbf{s}_k) \big)\lambda_{-}(\mathbf{w}_k\mathbf{w}_k^H-\mathbf{w}_{k+1} \mathbf{w}_{k+1}^H) \nonumber
\end{align}
It is not  immediately evident from the analysis which set of bounds in \eqref{eq36}, \eqref{eq37}, \eqref{eq40} are tight, hence combining them we have
\begin{align*}
&\max\left\{ \begin{aligned} g_{lb}^1( \mathbf{R_u}( \mathbf{s}_k), \mathbf{w}_k, \mathbf{w}_{k+1}), \;\; &g_{lb}^2( \mathbf{R_u}( \mathbf{s}_k), \mathbf{w}_k, \mathbf{w}_{k+1}), \\
&g_{lb}^3( \mathbf{R_u}( \mathbf{s}_k), \mathbf{w}_k, \mathbf{w}_{k+1}) \end{aligned}\right\} \\
&\leq g(\mathbf{w}_k,\mathbf{s}_k)-g(\mathbf{w}_{k+1},\mathbf{s}_k)  \leq \\
&\min\left\{ \begin{aligned} g_{ub}^1( \mathbf{R_u}( \mathbf{s}_k), \mathbf{w}_k, \mathbf{w}_{k+1}), \;\; &g_{ub}^2( \mathbf{R_u}( \mathbf{s}_k), \mathbf{w}_k, \mathbf{w}_{k+1}), \\
&g_{ub}^3( \mathbf{R_u}( \mathbf{s}_k), \mathbf{w}_k, \mathbf{w}_{k+1}) \end{aligned}\right\} \\
\end{align*}
 where $g_{lb}^i (\mathbf{R_u}( \mathbf{s}_k), \mathbf{w}_k, \mathbf{w}_{k+1})$, $g_{ub}^i (\mathbf{R_u}( \mathbf{s}_k), \mathbf{w}_k, \mathbf{w}_{k+1})$, $i=1,2,3$ are the lower and upper bounds as given in \eqref{eq36}-\eqref{eq37}, \eqref{eq40}, for $i=1,2,3$, respectively.
 
 Similar upper and lower bounds can be readily derived for the other corresponding terms,  $g(\mathbf{w}_{k+1},\mathbf{s}_k)-g(\mathbf{w}_{k+1},\mathbf{s}_{k+1})$ using analysis presented thus far, and is not the focus now. Let us however denote these corresponding lower and upper bounds to be $h_{lb}^i (\mathbf{R_u}( \mathbf{s}_k), \mathbf{w}_k, \mathbf{w}_{k+1})$, $ h_{ub}^i (\mathbf{R_u}( \mathbf{s}_k), \mathbf{w}_k, \mathbf{w}_{k+1})$, $i=1,2,3$.

\subsection{Constrained proximal alternating minimization}
The proximal version of the constrained alternating minimization is iterative, and for the filter design step, optimizes at the $k$-th iteration,
\begin{align} \label{eq41}
\min_{\mathbf{w} }\;\;\;\;\; &\mathbf{w}^H\mathbf{R_u}( \mathbf{s}_{k-1})\mathbf{w}+\frac{\alpha_{k-1}}{2} || \mathbf{w}-\mathbf{w}_{k-1} ||^2\\
\mbox{ s. t } \;\;\;\;\;&\mathbf{w}^H(\mathbf{v}(f_d)\otimes\mathbf{s}_{k-1}\otimes\mathbf{a}(\theta_t,\phi_t))=\kappa \nonumber
\end{align}
where $\alpha_{k-1} \in \mathbbm{R}^{+}$ can be seen as a weight attached to the regularizer / penalizer  $|| \mathbf{w}-\mathbf{w}_{k-1} ||^2$. This parameter can be interpreted as follows, if it is small,  it encourages the optimizer to look for viable solutions in the vicinity of $\mathbf{w}_{k-1}$. However, if large, it penalizes the optimizer heavily for focusing even slightly in the immediate vicinity of  $\mathbf{w}_{k-1}$.

In a similar spirit, the proximal version of the constrained alternating minimization for the waveform design step at the  $k$-th iteration optimizes,
\begin{align} \label{eq42}
\min\limits_{\mathbf{s}} \;\;\;\;\; &\mathbf{w}^H_{k}\mathbf{R_u}(\mathbf{s})\mathbf{w}_{k} +\frac{\beta_{k-1}}{2} ||\mathbf{s}-\mathbf{s}_{k-1} ||^2\nonumber \\
\mbox{s. t. }\;\;\;\;\; & \mathbf{w}^H_{k}(\mathbf{v}(f_d)\otimes\mathbf{s}\otimes\mathbf{a}(\theta_t,\phi_t))=\kappa   \\ 
\;\;\;\;\;\; & \mathbf{s}^H \mathbf{s}\leq P_o \nonumber \nonumber
\end{align}
where $\beta_{k-1}\in\mathbbm{R}^{+}$ is the weight attached to the regularizer $ ||\mathbf{s}-\mathbf{s}_{k-1} ||^2$. Bounds on $\alpha_{k-1},\beta_{k-1}$ relating it to the Lipschitz constants are deferred to forthcoming analysis. A graphical example comparing the constrained alternating minimization and the proximal constrained alternating minimization is shown in Fig.~\ref{amcamfig}.

\begin{figure*}[tbp!]
\centering
\includegraphics[scale=0.5]{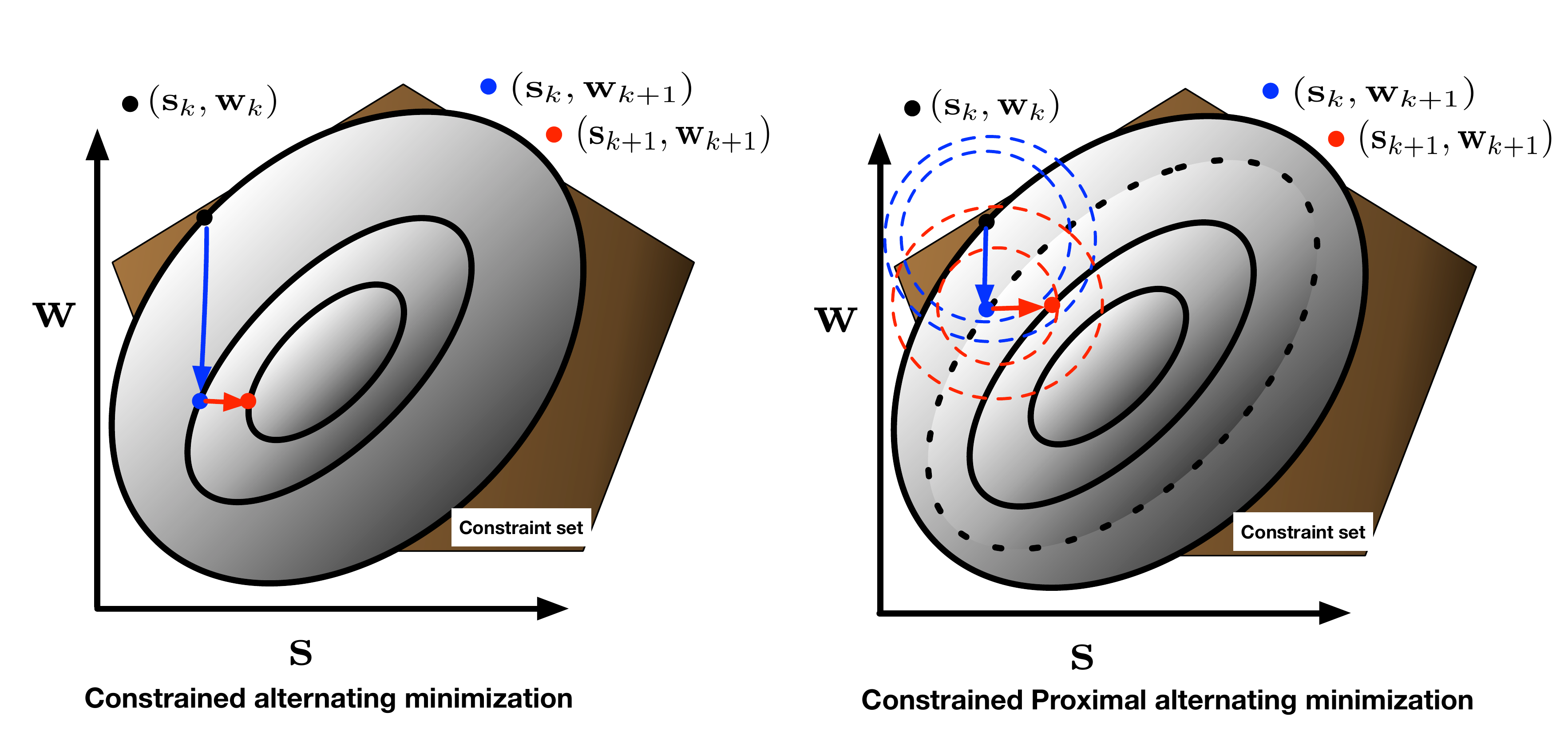}
\caption{Constrained alternating minimization (left) and proximal constrained alternating minimization (right). Iso level contours (each point on a curve has identical function values) and constraint set in background are shown. Outer iso-curves assume higher function values than the inner iso-curves. On right, and for particular $\alpha_{k}$, $\beta_{k}$, spheres (dashed, blue, dashed red) are the (two of the several) spheres of influence of the regularizer. Outer spheres penalize more than the inner.}
\label{amcamfig}
\end{figure*}

\begin{remk}\label{remark2}
The objective functions in \eqref{eq41}, \eqref{eq42} are still individually convex in $\mathbf{w}$, $\mathbf{s}$, respectively.
The regularizer terms $|| \mathbf{w}-\mathbf{w}_{k-1} ||^2$ and $ ||\mathbf{s}-\mathbf{s}_{k-1} ||^2$   are strongly convex, and  $\nabla^2_{\mathbf{w}} ( || \mathbf{w}-\mathbf{w}_{k-1} ||^2)=\mathbf{I}\succ \mathbf{0}$, $\nabla^2_{\mathbf{s}} ( || \mathbf{s}-\mathbf{s}_{k-1} ||^2)=\mathbf{I} \succ \mathbf{0}$, and therefore do not alter the individual convexity of $\mathbf{w}^H\mathbf{R_u}( \mathbf{s}_{k-1})\mathbf{w}$ and $\mathbf{w}^H_{k}\mathbf{R_u}(\mathbf{s})\mathbf{w}_{k}$, w.r.t. $\mathbf{w}$, $\mathbf{s}$, respectively.
\end{remk}
The solutions to \eqref{eq41}, \eqref{eq42} can be cast in terms of the proximal operator as
 \begin{align} 
 \mathbf{w}_{k}=&\mathrm{prox}_{(\alpha_{k-1},\mathbf{w} )} \big( g(\mathbf{w},\mathbf{s}_{k-1}) ;\mathbf{w}_{k-1}  \big) \label{eq43} \\
 &\mbox{ s. t } \; \mathbf{w}^H(\mathbf{v}(f_d)\otimes\mathbf{s}_{k-1}\otimes\mathbf{a}(\theta_t,\phi_t))=\kappa \nonumber \\ 
 \nonumber  \\
 \mathbf{s}_{k}=&\mathrm{prox}_{(\beta_{k-1},\mathbf{s} )} \big( g(\mathbf{w}_{k},\mathbf{s}) ;\mathbf{s}_{k-1}  \big) \label{eq44} \\
 &\mbox{s. t. } \; \mathbf{w}^H_{k}(\mathbf{v}(f_d)\otimes\mathbf{s}\otimes\mathbf{a}(\theta_t,\phi_t))=\kappa    \nonumber \\ 
& \;\;\;\;\;\;\;\;\mathbf{s}^H \mathbf{s}\leq P_o  \nonumber
 \end{align}
where, for a general $f(\mathbf{x}):\mathbbm{C}^N\rightarrow \mathbbm{R}$, the proximal operator is defined as
\begin{align} \label{eq45}
\mathrm{prox}_{(\alpha,\mathbf{x})} \big(f(\mathbf{x});\mathbf{y} \big):= \argmin \limits_{\mathbf{x} } \; \mathbf{f} (\mathbf{x}) +\frac{\alpha}{2} ||\mathbf{x}-\mathbf{y} ||^2.
\end{align}
The proximal operator has  a rich history in the literature, and well documented properties, see for example \cite{Parikh2013,Rockafeller1973,Rockafeller1976,Bertsekas1994}. A useful and interesting fact of this operator is that iff $\mathbf{x}_{o}$ minimizes $f(\mathbf{x}$) then $\mathbf{x}_{o}=\mathrm{prox}_{(\alpha,\mathbf{x})} (f(\mathbf{x});\mathbf{x}_{o})$, a proof is seen in \cite{Parikh2013}. 

{\bf Trust region interpretation}. The objective now is to relate the unconstrained proximal minimization as in \eqref{eq45} to  a well known technique in numerical optimization. A generalized trust region subproblem can be formulated for $\mathbf{f}(\mathbf{x}):\mathbbm{C}^N\rightarrow \mathbbm{R}$ \cite{More93}
\begin{align}\label{trustreg}
\min \limits_{\mathbf{x}} \;\;\; &f( \mathbf{x}) \nonumber \\
\mbox{s. t. } \;\;\; & || \mathbf{U}\mathbf{x}-\mathbf{v}||^2 \leq \delta
\end{align}
where $\mathbf{U}$, $\mathbf{v}$ are a general nonsingular matrix, and a vector, both characterizing the trust region. The positive scalar $\delta$ may be interpreted as  a parameter which specifies the extent of the trust region. For $\mathbf{U}=\mathbf{I}$ and $\mathbf{v}=\mathbf{y}$,  the proximal minimization as in  \eqref{eq45} and the trust region problem in \eqref{trustreg} are equivalent for specific values of $\alpha$ and $\delta$. In particular every solution of \eqref{eq45} is a solution to \eqref{trustreg} for  a particular $\delta$. In the same spirit, every solution to \eqref{trustreg} is an unconstrained minimizer to $f(\cdot)$ or a  solution to \eqref{eq45} for a particular $\alpha$, see also \cite{Rockafeller1976,Parikh2013}. 

The proximal optimizations problems, \eqref{eq41}, \eqref{eq42} can be cast as equivalent constrained trust region subproblems, where for the $k$-th iteration, the trust region is characterized by the previous iteration, $\mathbf{w}_{k-1}$, $\mathbf{s}_{k-1}$, respectively.

{\bf Closed form:} A closed form solution to \eqref{eq41} is readily derived, expressed as in \eqref{eq46} 
\begin{equation} \label{eq46}
\begin{aligned} \mathbf{w}_{k}&=\big( \mathbf{R}_{\bf u}(\mathbf{s}_{k-1})+\frac{\alpha_{k-1}}{2}\mathbf{I} \big)^{-1} \big( \frac{\alpha_{k-1}}{2}\mathbf{w}_{k-1}-\frac{\gamma_4^{\ast}}{2} \big(\mathbf{v}(f_d)\otimes\mathbf{s}_{k-1}\otimes\mathbf{a}(\theta_t,\phi_t) \big) \big) \\
\gamma_4&=\frac{\alpha_{k-1} \mathbf{w}_{k-1}^H \big( \mathbf{R}_{\bf u}(\mathbf{s}_{k-1})+\dfrac{\alpha_{k-1}}{2}\mathbf{I} \big)^{-1} \big(\mathbf{v}(f_d)\otimes\mathbf{s}_{k-1}\otimes\mathbf{a}(\theta_t,\phi_t) \big)-2\kappa }{\big(\mathbf{v}(f_d)\otimes\mathbf{s}_{k-1}\otimes\mathbf{a}(\theta_t,\phi_t) \big)^H
\big( \mathbf{R}_{\bf u}(\mathbf{s}_{k-1})+\dfrac{\alpha_{k-1}}{2}\mathbf{I} \big)^{-1}
 \big(\mathbf{v}(f_d)\otimes\mathbf{s}_{k-1}\otimes\mathbf{a}(\theta_t,\phi_t) \big)} \end{aligned}
\end{equation} 

where $\gamma_4$ is the Lagrange parameter associated with \eqref{eq41}. The solution to \eqref{eq42} is also in closed form and the procedure to obtain it is similar to that used in deriving \eqref{eq32}. Assuming that the Lagrange parameters for \eqref{eq42} are $\gamma_5=\gamma_{5r}+j\gamma_{5i},\,\gamma_6\in\mathbbm{R}^{+}$, the solution is expressed in \eqref{eq47},
\begin{equation} \label{eq47}
\mathbf{s}_{k}=\big( \sum\limits_{q=1}^Q\mathbf{Z_q}(\mathbf{w}_{k})+\frac{\beta_{k-1}}{2}\mathbf{I}+\gamma_6\mathbf{I}\big)^{-1} \big(\frac{\beta_{k-1}}{2}\mathbf{s}_{k-1} -\frac{\gamma_5}{2}\mathbf{G}^H\mathbf{w}_{k}\big)
\end{equation}
where,
\begin{equation*}
\begin{aligned}
\gamma_{5r}&=2\frac{\frac{\beta_{k-1}} {2}\mathrm{Re} \left\{\mathbf{w}_{k}^H\mathbf{G} \big( \sum\limits_{q=1}^Q\mathbf{Z_q}(\mathbf{w}_{k})+\frac{\beta_{k-1}}{2}\mathbf{I}+\gamma_6\mathbf{I}\big)^{-1} \mathbf{s}_{k-1}\right\}-\kappa} {\mathbf{w}_{k}^H\mathbf{G} \big( \sum\limits_{q=1}^Q\mathbf{Z_q}(\mathbf{w}_{k})+\frac{\beta_{k-1}}{2}\mathbf{I}+\gamma_6\mathbf{I}\big)^{-1} \mathbf{G}^H\mathbf{w}_{k} } \\
\gamma_{5i}&=\frac{\beta_{k-1} \mathrm{Im} \left\{\mathbf{w}_{k}^H\mathbf{G} \big( \sum\limits_{q=1}^Q\mathbf{Z_q}(\mathbf{w}_{k})+\frac{\beta_{k-1}}{2}\mathbf{I}+\gamma_6\mathbf{I}\big)^{-1} \mathbf{s}_{k-1}\right\} }{\mathbf{w}_{k}^H\mathbf{G} \big( \sum\limits_{q=1}^Q\mathbf{Z_q}(\mathbf{w}_{k})+\frac{\beta_{k-1}}{2}\mathbf{I}+\gamma_6\mathbf{I}\big)^{-1} \mathbf{G}^H\mathbf{w}_{k} }.
\end{aligned}
\end{equation*}
The Lagrange parameter $\gamma_6$ is obtained by solving, the following
\begin{align}\label{eq48}
\gamma_6 r(\gamma_6)=0,\;\gamma_6\geq0
\end{align}
obtained from the complementary slackness constraint on the Lagrange dual and where,
\begin{equation*}
\begin{aligned} r(\gamma_6)&=(P_o-\frac{\beta_{k-1}^2}{4} a_k) \bigg( \mathbf{w}_{k}^H\mathbf{G} \big( \sum\limits_{q=1}^Q\mathbf{Z_q}(\mathbf{w}_{k})+\frac{\beta_{k-1}}{2}\mathbf{I}+\gamma_6\mathbf{I}\big)^{-1} \mathbf{G}^H\mathbf{w}_{k} \bigg)^2 \\
-&2\big( b_i\frac{db_i}{d\gamma_6} +(b_r-\kappa) \frac{db_r}{d\gamma_6}\big) \mathbf{w}_{k}^H\mathbf{G} \big( \sum\limits_{q=1}^Q\mathbf{Z_q}(\mathbf{w}_{k})+\frac{\beta_{k-1}}{2}\mathbf{I}+\gamma_6\mathbf{I}\big)^{-1} \mathbf{G}^H\mathbf{w}_{k}  \\
-&(b_{i}^2+(b_r-\kappa)^2) \mathbf{w}_{k}^H\mathbf{G} \big( \sum\limits_{q=1}^Q\mathbf{Z_q}(\mathbf{w}_{k})+\frac{\beta_{k-1}}{2}\mathbf{I}+\gamma_6\mathbf{I}\big)^{-2} \mathbf{G}^H\mathbf{w}_{k}.
\end{aligned} 
\end{equation*}
 Where we also define
 \begin{align*}
 a_{k}=&\mathbf{s}_{k-1}^H\big( \sum\limits_{q=1}^Q\mathbf{Z_q}(\mathbf{w}_{k})+\frac{\beta_{k-1}}{2}\mathbf{I}+\gamma_6\mathbf{I}\big)^{-1} \mathbf{s}_{k-1} \\
 b_{r}&=\frac{\beta_{k-1}}{2}\mathrm{Re}\left\{ \mathbf{w}_{k}^H\mathbf{G} \big( \sum\limits_{q=1}^Q\mathbf{Z_q}(\mathbf{w}_{k})+\frac{\beta_{k-1}}{2}\mathbf{I}+\gamma_6\mathbf{I}\big)^{-1} \mathbf{s}_{k-1}\right\} \\
 b_{i}&=\frac{\beta_{k-1}}{2}\mathrm{Im}\left\{ \mathbf{w}_{k}^H\mathbf{G} \big( \sum\limits_{q=1}^Q\mathbf{Z_q}(\mathbf{w}_{k})+\frac{\beta_{k-1}}{2}\mathbf{I}+\gamma_6\mathbf{I}\big)^{-1} \mathbf{s}_{k-1}\right\}
 \end{align*}  
 
 Further, since the derivative, $\mathrm{Re}\{ \cdot\},\mathrm{Im}\{ \cdot\}$ are all linear we also have 
 \begin{align*}
 \dfrac{db_r}{d\gamma_6}&=-\frac{\beta_{k-1}}{2} \mathrm{Re}\left\{ \mathbf{w}_{k}^H\mathbf{G} \big( \sum\limits_{q=1}^Q\mathbf{Z_q}(\mathbf{w}_{k})+\frac{\beta_{k-1}}{2}\mathbf{I}+\gamma_6\mathbf{I}\big)^{-2} \mathbf{s}_{k-1}\right\} \\
 \dfrac{db_i}{d\gamma_6}&=-\frac{\beta_{k-1}}{2} \mathrm{Im}\left\{ \mathbf{w}_{k}^H\mathbf{G} \big( \sum\limits_{q=1}^Q\mathbf{Z_q}(\mathbf{w}_{k})+\frac{\beta_{k-1}}{2}\mathbf{I}+\gamma_6\mathbf{I}\big)^{-2} \mathbf{s}_{k-1}\right\}.
 \end{align*}

\begin{remk} \label{propos5}
In general $r(\gamma_6)$ is not monotone and there exist one or more zero crossings excluding $\gamma_6=\infty$. However in our extensive numerical simulations, and assuming $P_o>>\kappa^2, \gamma_6=0$ solves \eqref{eq48}.
\end{remk}

It is readily seen that $\lim \limits_{\gamma_6\rightarrow 0} r(\gamma_6)=r(0)\neq 0,\lim \limits_{\gamma_6\rightarrow \infty} r(\gamma_6)=0,\lim \limits_{\gamma_6\rightarrow \infty} \gamma_6 r(\gamma_6)=0$. Nevertheless unlike Prop.~\ref{propos3},  Rem.~\ref{propos5} is not straightforward  to demonstrate analytically, however can be shown numerically. See Section IV for some demonstrative examples not specific to the radar problem.  

The value of $\gamma_6=0$ is substituted in \eqref{eq47} to obtain the final waveform solution $\mathbf{s}_{k}(\cdot)$. 

\begin{remk} \label{remstrongdual}
({\it Strong duality}) The primal problem, \eqref{eq42} and its associated dual have zero duality gap. This is straightforward but tedious to show. However we provide the optimal values attained by the primal as well as the dual, given below,
\begin{align} \label{dualopt2}
\mathbf{w}_k^H(\mathbf{R_i+R_n })\mathbf{w}_k&+\mathbf{s}_k^{\ast H} \big( \sum\limits_{q=1}^Q\mathbf{Z_q}(\mathbf{w}_{k})+\frac{\beta_{k-1}}{2}\mathbf{I}\big)^{-1} \mathbf{s}_k^{\ast} \nonumber \\
&+\frac{\beta_{k-1} || \mathbf{s}_{k-1}||^2}{2}  
-\beta_{k-1}\mbox{Re}\{ \mathbf{s}_{k}^{\ast H} \mathbf{s}_{k-1}\}
\end{align}
where using \eqref{eq47}, Prop.~\ref{propos5}, 
\begin{equation*}
\begin{aligned}
\mathbf{s}_{k}^{\ast}&=(\sum\limits_{q=1}^Q\mathbf{Z_q}(\mathbf{w}_{k})+\frac{\beta_{k-1}}{2}\mathbf{I})^{-1} ( \frac{\beta_{k-1}}{2}\mathbf{s}_{k-1} -\frac{\gamma_5}{2}\mathbf{G}^H\mathbf{w}_{k}) \\
 \gamma_5&=\frac{\beta_{k-1}\mathbf{w}_{k}^H\mathbf{G} ( \sum\limits_{q=1}^Q\mathbf{Z_q}(\mathbf{w}_{k})+\frac{\beta_{k-1}}{2}\mathbf{I})^{-1} \mathbf{s}_{k-1}-2\kappa} {\mathbf{w}_{k}^H\mathbf{G} ( \sum\limits_{q=1}^Q\mathbf{Z_q}(\mathbf{w}_{k})+\frac{\beta_{k-1}}{2}\mathbf{I})^{-1} \mathbf{G}^H\mathbf{w}_{k} }.
 \end{aligned}
 \end{equation*}
\end{remk}
This is not surprising since it is similar to Rem.~\ref{dualprop}. However, in this case the condition on the existence of the matrix is irrelevant, since the inverse in \eqref{eq47} always exists. Hence Slater's condition now is a simple constraint qualifier (the power constraint) which must be satisfied as in Rem.~\ref{dualprop}.

{\bf Interpretation with specific ranges of $\alpha_{k-1},\,\beta_{k-1}$ and related to the Lipschitz constants}. Some definitions and lemmas are useful for future discussions and are expressed below
 
 \begin{defn}\label{mydef2}
({\it Lipschitz continuous gradient}) A function $f(\mathbf{\bar{x}}) :\mathbbm{R}^N\rightarrow\mathbbm{R}$ has a Lipschitz constant (and trivially real positive), $\mathtt{L}$, when $||\nabla_{\mathbf{\bar{x}}}f(\mathbf{\bar{x}})-\nabla_{\mathbf{\bar{y}}}f(\mathbf{\bar{y}})|| \leq\mathtt{L} || \mathbf{\bar{x}-\bar{y}}||$, and $\forall \mathbf{\bar{x}}$, $\mathbf{\bar{y}}\in \mathbbm{R}^N$.
\end{defn}
{\bf Note:} ({\it upper bound on Hessian} ) If $f(\mathbf{\bar{x}})$ has a Lipschitz continuous gradient, with constant $\mathtt{L}$, then  using Taylor's theorem, it can be proved that $\nabla_{\mathbf{\bar{x}}}^2 f(\mathbf{\bar{x}})\preceq \mathtt{L}\mathbf{I}$.
\begin{remk} \label{remark3}
The Lipschitz constant for $f(\mathbf{\bar{x}})=\mathbf{\bar{x}}^T\mathbf{\bar{B}}\mathbf{\bar{x}}$ is the maximum eigenvalue of $\mathbf{\bar{B}}$, i.e. $\lambda_{\max}(\mathbf{\bar{B}})$, where $\mathbf{\bar{B}}\in\mathbbm{R}^{N\times N}, \; \mathbf{\bar{x}}\in\mathbbm{R}^N$.
\end{remk}
This is readily seen since $\nabla_{\mathbf{\bar{x}}}\mathbf{\bar{x}}^T\mathbf{\bar{B}\bar{x}}=\mathbf{\bar{B}\bar{x}}$. Further since the induced (by  an arbitrary $\mathbf{\bar{z}} \in \mathbbm{R}^N$) spectral norm (notation: $||| \cdot |||$) is defined  as
\begin{align*} 
|||\mathbf{\bar{B}}|||:=\sup\limits_{\mathbf{\bar{z}}} \{ \frac{||\bf \bar{B}\bar{z}||}{||\bf \bar{z}||} : \mathbf{\bar{z}}\in\mathbbm{R}^N,\mathbf{\bar{z}}\neq\mathbf{0} \},||{\bf \bar{B}\bar{z}}||=\sqrt{ \mathbf{\bar{z}}^T\mathbf{\bar{B}}^T\mathbf{\bar{B}\bar{z}}}
\end{align*}
but we know from Lem.~\ref{lemma2} that $\mathbf{\bar{z}}^T\mathbf{\bar{B}\bar{z}}\leq \lambda_{\max}( \mathbf{\bar{B}})|| \bf \bar{z}||^2$ and that eigenvalues of $\mathbf{\bar{B}}$ and $\mathbf{\bar{B}}^T$ are identical. This further implies that $\mathbf{\bar{z}}^T\mathbf{\bar{B}}^T\mathbf{\bar{B}\bar{z}}\leq \lambda_{\max}^2( \mathbf{\bar{B}})|| \bf \bar{z}||^2$. Therefore from Definition \ref{mydef2}, it is readily seen that the Lipschitz constant is the maximum eigenvalue of $\mathbf{\bar{B}}$. 
\begin{lem}\label{lemma5}
({\it Descent lemma}) If $f(\mathbf{\bar{x}}) :\mathbbm{R}^N\rightarrow\mathbbm{R}$ is continuously differentiable and has a Lipschitz continuous gradient described by constant $\mathtt{L}$, then $f(\mathbf{\bar{x}})\leq f(\mathbf{\bar{y}})+\nabla_{\mathbf{\bar{y}}}f(\mathbf{\bar{y}})^T(\mathbf{\bar{x}}-\mathbf{\bar{y}}) +\frac{\mathtt{L}}{2}||\mathbf{\bar{x}}-\mathbf{\bar{y}}||^2$.
\end{lem}
\begin{proof}
See \cite[Prop.~A.24]{Bertsekas1999} and also \cite[Lem2.2]{Beck2013} relevant in general for the BCD.
\end{proof}

Consider an arbitrary $g(\mathbf{x}):= \mathbf{x}^H\mathbf{Bx}$, and $\mathbf{B}=\mathbf{B}^H$, $\mathbf{x}\in\mathbbm{C}^N$. Since $g(\mathbf{x}):\mathbbm{C}^N\rightarrow\mathbbm{R}$, a real equivalent of $g(\mathbf{x})$ could be defined as $\bar{g}(\bar{\mathbf{x}}):=\bar{\mathbf{x}}^T\bar{\mathbf{B} }\bar{\mathbf{x}}$ where 
\begin{equation*}
\bar{\mathbf{B}}:= \left[ \begin{matrix} \mbox{Re}\{ \mathbf{B} \} & -\mbox{Im} \{\mathbf{B} \}  \\
\mbox{Im} \{\mathbf{B} \} & \mbox{Re}\{ \mathbf{B} \} \end{matrix} \right] \in\mathbbm{R}^{2N\times 2N}, \;\; \bar{\mathbf{x}}=[\mbox{Re} \{ \mathbf{x} \}^T \mbox{Im} \{ \mathbf{x} \}^T ]^T \in \mathbbm{R}^{2N}.
\end{equation*}

\begin{lem}\label{lemma6}
The matrix $\bar{\mathbf{B}} :=\left[ \begin{smallmatrix} \mathrm{Re}\{ \mathbf{B} \} & -\mathrm{Im} \{\mathbf{B} \}  \\
\mathrm{Im} \{\mathbf{B} \} & \mathrm{Re}\{ \mathbf{B} \} \end{smallmatrix} \right]\in\mathbbm{R}^{2N\times 2N}$ and $ \left[ \begin{smallmatrix} \mathbf{B} & \mathbf{0} \\ \mathbf{0} & \mathbf{B}^{\ast}  \end{smallmatrix} \right] \in\mathbbm{C}^{2N\times 2N} $ have identical eigenvalues, $\tilde{\lambda}_i,i=1,2,\ldots,2N$.  Moreover, if $\mathbf{B}$ is Hermitian, then $\tilde{\lambda}_i \in \mathbbm{R}^{+},i=1,2,\ldots,2N$ are equal to twice the multiplicity of the eigenvalues of  $\mathbf{B}\in\mathbbm{C}^{N\times N}$.
\end{lem}
\begin{proof} 
Owing to the complex to real-real isomorphism, it can be shown after algebraic manipulations that
\begin{align} \label{eq49}
\left[ \begin{matrix} \mathbf{B} & \mathbf{0} \\ \mathbf{0} & \mathbf{B}^{\ast}  \end{matrix} \right] =\mathbf{P}^{H}\bar{\mathbf{B}}\mathbf{P}, \;\; \mathbf{P}=\frac{1}{\sqrt{2}}\left[ \begin{matrix} j\mathbf{I} & \mathbf{I} \\ \mathbf{I} & j\mathbf{I} \end{matrix} \right], \;\; \mathbf{P}^H=\mathbf{P}^{-1}.
\end{align}
That is \eqref{eq49} indicates that $\bar{\mathbf{B}}$ and $\left[ \begin{smallmatrix} \mathbf{B} & \mathbf{0} \\ \mathbf{0} & \mathbf{B}^{\ast}  \end{smallmatrix} \right]$ are unitary equivalent. Therefore they share the same eigenvalues. Furthermore if $\mathbf{B}$ is Hermitian its eigenvalues are purely real, and  hence trivially,  the eigenvalues of $\mathbf{B}$, $ \mathbf{B}^{\ast}$ are identical, and their eigenvectors are complex conjugates of one another. Hence the block diagonal matrix has identical eigenvalues as $\mathbf{B}$ but with multiplicity two.
\end{proof}

Consider the objective in \eqref{eq41}, \eqref{eq42}. Define $\bar{g}( \mathbf{\bar{w}},\mathbf{\bar{s}}_{k-1}),\bar{g}(\mathbf{\bar{w}}_{k-1},\mathbf{\bar{s}}_{k-1})$ as the real equivalents of $g( \mathbf{w},\mathbf{s}_{k-1}),g (\mathbf{w}_{k-1},\mathbf{s}_{k-1})$, respectively for the filter design objective as in \eqref{eq41}. In addition, denote  $\mathtt{L}_{1k-1}$ as the Lipschitz constant associated with $\bar{g}(\bar{\mathbf{w}}_{k-1},\mathbf{\bar{s}}_{k-1})$. Similarly using the same notation and for the objective in the waveform design objective  as in \eqref{eq42} consider the real equivalents, $\bar{g}( \mathbf{\bar{s}},\mathbf{\bar{w}}_{k}),\bar{g}(\mathbf{\bar{s}}_{k-1},\mathbf{\bar{w}}_{k})$ and the Lipschitz constant denoted as $\mathtt{L}_{2k-1}$. Then the following inequalities can now be shown.
\begin{equation} \label{eq50}
\begin{aligned} &\bar{g}( \mathbf{\bar{w}})+\frac{\mathtt{L}_{1k-1}}{2}||\mathbf{\bar{w}}_{k-1}-\mathbf{\bar{w}} ||^2
  \geq \bar{g}(\mathbf{\bar{w}}_{k-1}) \\
  +&\nabla \bar{g}(\mathbf{\bar{w}}_{k-1})^T(\mathbf{\bar{w}}-\mathbf{\bar{w}}_{k-1}) 
  +\frac{\mathtt{L}_{1k-1}}{2}||\mathbf{\bar{w}}_{k-1}-\mathbf{\bar{w}}||^2 \geq \bar{g}(\mathbf{\bar{w}}) \end{aligned}
\end{equation}
\begin{equation} \label{eq51}
\begin{aligned} &\bar{g}( \mathbf{\bar{s}})+\frac{\mathtt{L}_{2k-1}}{2}||\mathbf{\bar{s}}_{k-1}-\mathbf{\bar{s}} ||^2
  \geq \bar{g}(\mathbf{\bar{s}}_{k-1}) \\
  +&\nabla \bar{g}(\mathbf{\bar{s}}_{k-1})^T(\mathbf{\bar{s}}-\mathbf{\bar{s}}_{k-1}) 
  +\frac{\mathtt{L}_{2k-1}}{2}||\mathbf{\bar{s}}_{k-1}-\mathbf{\bar{s}}||^2 \geq \bar{g}(\mathbf{\bar{s}}) \end{aligned}
\end{equation}
where in \eqref{eq50}, the known's $\mathbf{\bar{s}}_{k-1}$ and  in \eqref{eq51}, the known's $\mathbf{\bar{w}}_{k}$ are respectively treated as constants, therefore suppressed in notation  for brevity. We further note that \eqref{eq50}, \eqref{eq51} are tight, i.e.  for $\mathbf{\bar{w}}_{k}=\mathbf{\bar{w}}_{k-1}$,  $\mathbf{\bar{s}}_{k}=\mathbf{\bar{s}}_{k-1}$ the inequalities are strict equality's. The Lipschitz constants, $\mathtt{L}_{1k-1}$, $\mathtt{L}_{2k-1}$ are readily derived using Lem.~\ref{lemma6}.  
\begin{remk} \label{remark4}
It is readily seen that if $\alpha_{k-1}\geq \mathtt{L}_{1k-1}$ and $\beta_{2k-1}\geq \mathtt{L}_{2k-1}$ the inequalities in \eqref{eq50}, \eqref{eq51} are valid by replacing $\mathtt{L}_{1k-1},\mathtt{L}_{2k-1}$ with $\alpha_{k-1},\beta_{k-1}$, respectively.
\end{remk}

The term in the first inequalities of \eqref{eq50}, \eqref{eq51} are the proximal minimization objectives with $\alpha_{k-1}=\mathtt{L}_{1k-1},\beta_{k-1}=\mathtt{L}_{2k-1}$. The  inequalities of \eqref{eq50}, \eqref{eq51}  are obtained from first applying the convexity Def.~\ref{mydef1}(b) (first order definition) and then subsequently adding the respective terms $\tfrac{\mathtt{L}_{1k-1}}{2}||\mathbf{\bar{w}}_{k-1}-\mathbf{\bar{w}} ||^2$, $ \tfrac{\mathtt{L}_{2k-1}}{2}||\mathbf{\bar{s}}_{k-1}-\mathbf{\bar{s}}||^2$ and then using  Lem.~\ref{lemma5}, the descent lemma. 

Additionally, it is recalled that the functions associated with the second  inequalities of \eqref{eq50}, \eqref{eq51} are the  (unconstrained) objectives which are minimized by the gradient descent with step size $\mathtt{L}_{1k-1}$, $\mathtt{L}_{2k-1}$, respectively. That is, the new iterations are then $\mathbf{\bar{w}}_{k}=\mathbf{\bar{w}}_{k-1}-\frac{1}{\mathtt{L}_{1k-1}} \nabla_{\mathbf{\bar{w}}}\bar{g}( \mathbf{\bar{w}})$, and  $\mathbf{\bar{s}}_{k}=\mathbf{\bar{s}}_{k-1}-\frac{1}{\mathtt{L}_{2k-1}} \nabla_{\mathbf{\bar{s}}}\bar{g}( \mathbf{\bar{s}})$. Therefore from \eqref{eq50}, \eqref{eq51} and Rem.~\ref{remark4} {\it we note that the proximal objective, the gradient descent objective are all surrogate albeit tight upper bounds on the true objective $\forall \alpha_{k-1}\geq\mathtt{L}_{1k-1}$ and $\forall \beta_{k-1}\geq\mathtt{L}_{2k-1}$}. This interpretation is graphically depicted in Fig.~\ref{fig3} for the filter design objective as in \eqref{eq41} but for $\alpha_{k-1}=\mathtt{L}_{1k-1}$. A similar graphic interpretation is obvious for the waveform design stage and is therefore not shown.

{\bf Tikhonov interpretation} This interpretation is immediate from \eqref{eq46}, \eqref{eq47}. In fact from \eqref{eq41}, \eqref{eq42},  the quadratic  regularizers $||\mathbf{w}-\mathbf{w}_{k-1} ||^2, || \mathbf{s}-\mathbf{s}_{k-1} ||^2$ are essentially Tikhonov regularization terms. Geometrically they are spheres centered at $\mathbf{w}_{k-1}$, $\mathbf{s}_{k-1}$ and encourage the current iterates to be in the vicinity of the previous iterates. Furthermore, since in the limit, the regularizer terms only decrease, this may be also seen as a vanishing Tikhonov regularization problem \cite{Parikh2013} for each iteration in both the waveform and the filter vectors.

{\bf Proximal minimization: A training data starved STAP solution} The regularization in \eqref{eq41}, \eqref{eq42} leads to {\it diagonally loaded} solutions \eqref{eq46}, \eqref{eq47} when compared to the constrained alternating minimization solutions as in \eqref{weightcomp} and \eqref{eq32}. In particular, the diagonal loading serves two important purposes, {\it firstly it offers a numerically stable solution by conditioning . Secondly and more importantly, it permits a weight vector solution when $\mathrm{rank}( \mathbf{R_u}(\mathbf{s}))\leq NML$}.   

Practical STAP contends with rank deficient correlation matrices due to lack of sufficient training data from neighboring range cells due to outlier contamination or heterogeneity in the data. The solution in \eqref{eq41} ameliorates over the training data starved STAP scenarios.

So far, we have considered the algorithms for waveform design without enforcing constraints such as const. modulus or sidelobe constraints. The minimum eigenvector solution belongs to this class of unconstrained waveform design. We will revisit this design by considering \eqref{eq18} and Lem.~\ref{lemma2}.
\begin{remk} \label{remark5}
The min. eigenvector solution in \eqref{eq21} is still optimal in the presence of clutter, provided $\mathbf{R_i}+\mathbf{R_n}$ and $\mathbf{R_c}(\mathbf{s})$ share the same eigenvector corresponding to their min.eigenvalues, but with $\lambda_{\min}(\mathbf{R_c}(\mathbf{s}))=0$, always.
\end{remk}

This is readily seen since the optimization in \eqref{eq18}, ignoring the constraint for now could be recast as $\max \limits_\mathbf{s} (\mathbf{v}(f_d)\otimes\mathbf{s}\otimes\mathbf{a}(\theta_t,\phi_t))^H\mathbf{R}_{\bf u}^{-1}(\mathbf{s})(\mathbf{v}(f_d)\otimes\mathbf{s}\otimes\mathbf{a}(\theta_t,\phi_t))$.  Now using Woodbury's identity \cite{Kayest1998}, we have
\begin{equation} \label{eqeigmin}
\begin{aligned}
&(\mathbf{R_i}+\mathbf{R_n}+\mathbf{R_u}(\mathbf{s}))^{-1}  =(\mathbf{R_i}+\mathbf{R_n})^{-1} \\
-&(\mathbf{R_i}+\mathbf{R_n})^{-1}\mathbf{R_c}(\mathbf{s} ) \bigl( \mathbf{I}+ (\mathbf{R_i}+\mathbf{R_n})^{-1} \mathbf{R_c}(\mathbf{s})\bigr) ^{-1}(\mathbf{R_i}+\mathbf{R_n})^{-1}.
\end{aligned}
\end{equation}
Further using the eigenvector relations, $(\mathbf{R_i}+\mathbf{R_n})(\mathbf{v}(f_d)\otimes\mathbf{s}\otimes\mathbf{a}(\theta_t,\phi_t))=\lambda_{\min}(\mathbf{R_i}+\mathbf{R_n})(\mathbf{v}(f_d)\otimes\mathbf{s}\otimes\mathbf{a}(\theta_t,\phi_t))$ and $\mathbf{R_c} (\mathbf{s}) (\mathbf{v}(f_d)\otimes\mathbf{s}\otimes\mathbf{a}(\theta_t,\phi_t))=\lambda_{\min}(\mathbf{R_c}( \mathbf{s}))(\mathbf{v}(f_d)\otimes\mathbf{s}\otimes\mathbf{a}(\theta_t,\phi_t))=\mathbf{0}$ in \eqref{eqeigmin}, it is readily seen that $(\mathbf{v}(f_d)\otimes\mathbf{s}\otimes\mathbf{a}(\theta_t,\phi_t))^H(\mathbf{R_i}+\mathbf{R_n}+\mathbf{R_u}(\mathbf{s}))^{-1} )(\mathbf{v}(f_d)\otimes\mathbf{s}\otimes\mathbf{a}(\theta_t,\phi_t))=\lambda_{\min}^{-1}( \mathbf{R_i}+\mathbf{R_n})$.

The simplest example where Rem.~\ref{remark5} is satisfied is when the noise correlation matrix is scaled identity (may not be practical for narrowband radar), clutter correlation matrix is low rank. In STAP and for ideal scenarios, insights to the clutter rank are obtained by the Brennan's rule \cite{guerci2003,klemm2002,ward1994}. A high clutter rank prevails due to the practical effects such as, the intrinsic clutter motion,velocity misalignment and crabbing, mutual coupling and antennae element mismatches as well as clutter ambiguities in Doppler resulting in aliasing \cite{ward1994}. 

\begin{figure} [tbp!]
\centering
\includegraphics[scale=0.5]{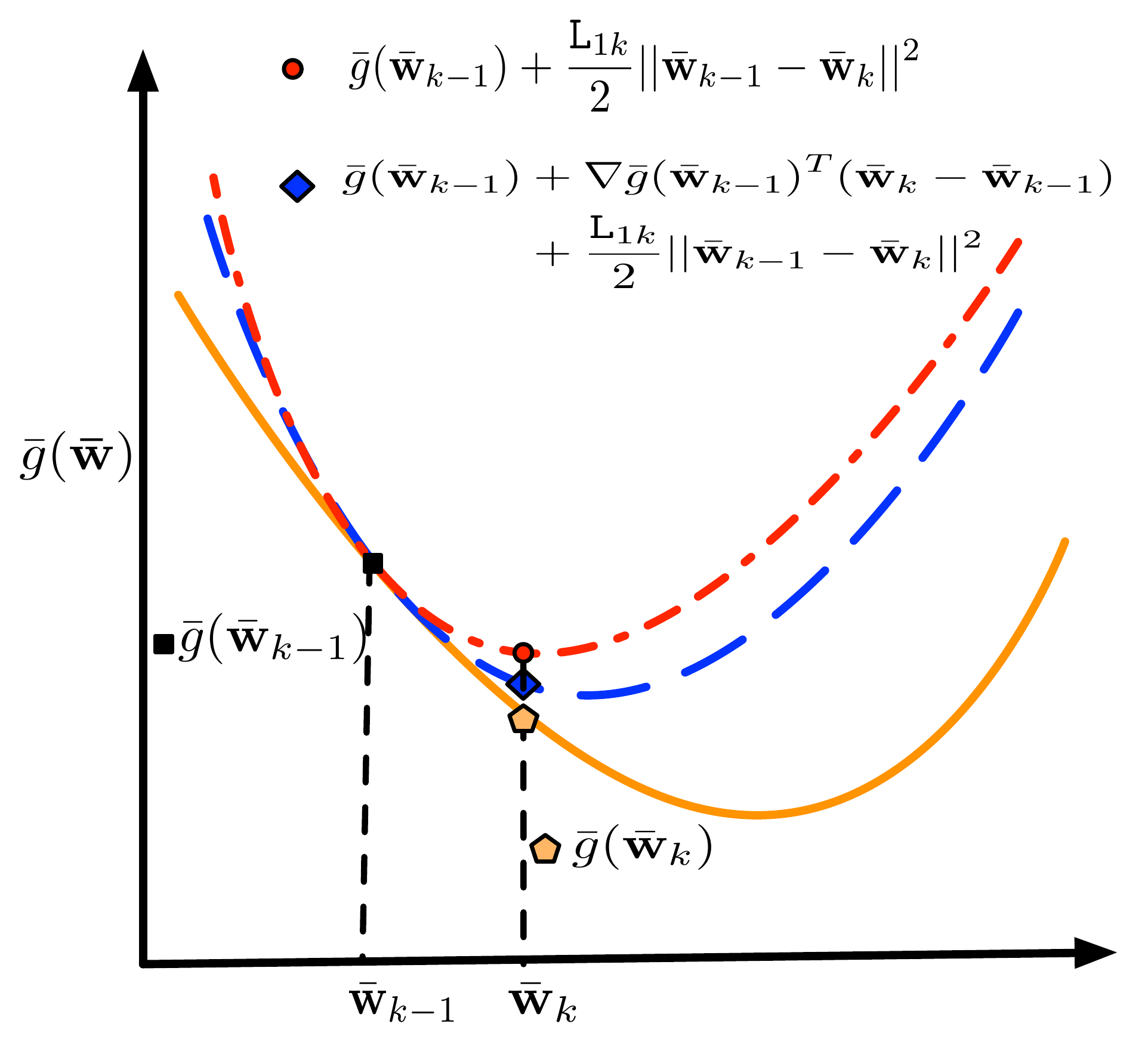}
\caption{Upper bounds on the objective for the proximal algorithm w.r.t. the filter design. A similar graphical interpretation for the waveform design but with $\mathtt{L}_{2k-1}$ is also easy depicted but not shown here.}
\label{fig3}
\end{figure}
\subsection{Constant modulus alternating minimization}
So far, the optimization problems had no specific constraints (except the power/energy constraint) on the waveform, constant modulus is a desirable property to have in a waveform \cite{Setlurradar2014}. The optimum weight vector is unchanged by introducing the const. modulus constraint, and is identical to \eqref{weightcomp} for the constrained alternating minimization \footnote{The analysis of the proximal constrained alternating minimization with the const. mod. constraint is omitted, but can be readily derived from the analysis of its non-proximal counterpart, presented here.}.

Since the optimization w.r.t. weight vector is unchanged, we only treat the optimization for $\mathbf{s}$ but with the const. mod. constraint for a fixed but arbitrary $\mathbf{w}$, formulated below
\begin{align}
\min\limits_{\mathbf{s}} \;\;\;\;\; &\mathbf{w}^H\mathbf{R_u}(\mathbf{s})\mathbf{w} \nonumber \\
\mbox{s. t. }\;\;\;\;\; & \mathbf{w}^H(\mathbf{v}(f_d)\otimes\mathbf{s}\otimes\mathbf{a}(\theta_t,\phi_t))=\kappa  \label{eq52} \\ 
\;\;\;\;\;\; & |s_i|=\rho,i=1,2,\ldots,N. \nonumber \nonumber
\end{align}
 where $s_i$ is the $i$-th component in $\mathbf{s}$. Unlike say \eqref{eq17}, notice that in \eqref{eq52}, constraining  the power of the waveform is unnecessary since $\rho$ is fixed but could be chosen arbitrarily to scale up / down the waveforms energy to satisfy hardware limitations. Therefore, the last $N$ constraints in \eqref{eq52} implicitly impose the power requirements, but more importantly also impose the constant modulus constraint.

The Lagrangian of \eqref{eq52} is expressed as
\begin{align} \label{lagraconmod}
\mathcal{L}(\mathbf{s}, \gamma_7,\boldsymbol{\gamma}_5)&=\mathbf{w}^H\mathbf{R_u}(\mathbf{s})\mathbf{w} +\mbox{Re}\{ \gamma_7^{\ast} (\mathbf{w}^H\mathbf{Qs}-\kappa)\} \nonumber  \\
&+\mathbf{s}^H\mathbf{D}_{\gamma}\mathbf{s} -\rho\mathbf{1}^T\boldsymbol{\gamma}_8
\end{align}
where the Lagrange parameter, $\gamma_7\in\mathbbm{C}$, and the Lagrange parameter vector $\boldsymbol{\gamma}_8=[\gamma_{8_1},\gamma_{8_2},\ldots,\gamma_{8_N}]^T\in\mathbbm{R}^{N}$ are for the Capon constraint and the $N$ const. mod. constraints, respectively. Furthermore in  \eqref{lagraconmod}, define $\mathbf{D}_{\gamma}= \left[ \begin{smallmatrix} \gamma_{8_1}& & \\
& \ddots & \\ & &\gamma_{8_N} \end{smallmatrix} \right]$, i.e. a diagonal matrix. The KKT conditions are expressed as
\begin{subequations}\label{kkt1}
\begin{align} 
\mathbf{s}_o(\mathbf{w})&=\frac{\kappa \bigl( \sum\limits_{q=1}^Q \mathbf{Z_q}( \mathbf{w})+\mathbf{D}_{\gamma}\bigr)^{-1}\mathbf{G}^H\mathbf{w} }{\mathbf{w}^H\mathbf{G} \bigl( \sum\limits_{q=1}^Q \mathbf{Z_q}( \mathbf{w}) +\mathbf{D}_{\gamma}\bigr)^{-1}\mathbf{G}^H\mathbf{w}} \\
|s_{oi}(\mathbf{w})|&=\rho,i=1,2,\ldots,N.
\end{align}
\end{subequations}
The waveform which simultaneously satisfies \eqref{kkt1}(a)(b) is the solution. Moreover, note that \eqref{kkt1}(a)(b) are $2N$ non-linear equations with $2N$ unknowns. The first $N$ unknowns are $s_{oi}(\mathbf{w}),i=1,2\ldots,N$ and the next $N$ unknowns are the Lagrange parameters $\gamma_{8_i}$. Unfortunately, \eqref{kkt1} is not in closed form but can be solved numerically for the $N$ parameters, $\gamma_{8_i},i=1,2,\ldots,N$ via a numerical root finder. Nonetheless we note that $\gamma_{8_i}\in (-\infty,\infty)$ and a reasonable initialization point is not forthcoming for the numerical root finding. 

{\bf Eliminating the constant modulus constraints} Instead of solving the $2N$ non-linear equations as in \eqref{kkt1}(a)(b), we take an alternative approach. One may reformulate the optimization \eqref{eq52} by eliminating the last $N$ constraints, by imposing  a structure on $\mathbf{s}$, namely, $s_i=\rho\exp(j\alpha_i)$. Other structures exists but from our experience, complex exponentials are the easiest to manipulate. The new optimization problem is now w.r.t. $\boldsymbol{\alpha}=[\alpha_1,\alpha_2,\ldots,\alpha_{N}]^T\in\mathbbm{R}^N$, expressed as
\begin{align}
\min\limits_{\boldsymbol{\alpha}} \;\;\;\;\; &\mathbf{w}^H\mathbf{R_u}(\mathbf{s})\mathbf{w} \nonumber \\
\mbox{s. t. }\;\;\;\;\; & \mathbf{w}^H(\mathbf{v}(f_d)\otimes\mathbf{s}\otimes\mathbf{a}(\theta_t,\phi_t))=\kappa  \label{eq53}
\end{align}
where in,
$\mathbf{s}=\rho[\exp(j\alpha_1),\exp(j\alpha_2),\ldots,\exp(j\alpha_{N})]^T$ and $\alpha_i \in [0,2\pi),i=1,2,\ldots, N$.
The Lagrangian corresponding to \eqref{eq53} is
\begin{align} \label{eq54}
\mathcal{L}(\boldsymbol{\alpha},\gamma_9)=\mathbf{w}^H\mathbf{R_u}(\mathbf{s})\mathbf{w} +\mbox{Re}\{ \gamma_9^{\ast}(\mathbf{w}^H\mathbf{Gs}-\kappa)\}.
\end{align}
The KKT's are expressed as, $\tfrac{\partial\mathcal{L}(\boldsymbol{\alpha},\gamma_9)}{\partial \boldsymbol{\alpha}}=\mathbf{0}$ and $\mathbf{w}^H\mathbf{Gs}=\kappa$. Noting that $\boldsymbol{\alpha}$ is purely real, we have
\begin{equation}\label{eq56}
\begin{aligned}
\frac{\partial\mathcal{L}(\boldsymbol{\alpha},\gamma_9)}{\partial \boldsymbol{\alpha}}=-&j\sum\limits_{q=1}^Q\mathbf{Z_q}\mathbf{s}\odot \mathbf{s}^{\ast}+j\sum\limits_{q=1}^Q\mathbf{Z}_{\bf q}^{\ast}\mathbf{s}^{\ast}\odot \mathbf{s} \\
+&\mbox{Im}\{ \gamma_9^{\ast}(\mathbf{w}^H\mathbf{G} )^T\odot \mathbf{s}\} =\mathbf{0}.
\end{aligned}
\end{equation}
The above equation can be simplified as, $\mbox{Im}\{ \sum\limits_{q=1}^Q\mathbf{Z}_{\bf q}^{\ast}\mathbf{s}^{\ast}\odot \mathbf{s} -\frac{\gamma_9^{\ast}}{2}(\mathbf{w}^H\mathbf{G})^T\odot \mathbf{s}\}$. Using this in \eqref{eq56}, and taking the complex conjugate, while absorbing the negative sign into the constant $\gamma_9$\footnote{new $\gamma_9$=old $-\gamma_9$.}, we have the KKTs in final form expressed as
\begin{subequations}\label{eq57}
\begin{align} 
\mbox{Im}\left\{ \left(\sum\limits_{q=1}^Q\mathbf{Z_q}(\mathbf{w}) \mathbf{s}_o+\frac{\gamma_9}{2} \mathbf{G}^H \mathbf{w} \right)\odot \mathbf{s}_o^* \right\}&=\mathbf{0} \\
\mathbf{w}^H\mathbf{G}\mathbf{s}_o&=\kappa
\end{align}
\end{subequations}
where $\mathbf{0}$ is a column vector of all zeros and of dimension $N$. The optimal solution, $\mathbf{s}_o$, is a  function of the optimal $\boldsymbol{\alpha}_o$. This relationship although evident from \eqref{eq53} is not explicitly stressed in \eqref{eq57} for notational succinctness. 
Define $\mathbf{Z}_{\bf Q}:=\sum\limits_{q=1}^Q\mathbf{Z_q}(\mathbf{w})$ and let $z_{ij},i=1,2,\ldots,N,j=1,2\ldots,N$ be the $ij$-th element of $\mathbf{Z}_{\bf Q}$. Noting that $\mathbf{Z}_{\bf Q}$ is Hermitian, we also have $\mbox{Im}\{z_{ii}\}=0,\,\forall i$, $z_{ji}=z_{ij}^{\ast}$.

\begin{prop}\label{propos6}
The Lagrange parameter $\gamma_9=0$ solves \eqref{eq57}.
\end{prop}
\begin{proof}
For any $z\in\mathbbm{C}$, and any $\theta\in [0,2\pi]$, we have $\mbox{Im}\{ z\exp(j\theta)\}=\mbox{Re}\{z\}\sin(\theta)+\mbox{Im}\{z\}\cos(\theta)$. Using this and the fact that $\mathbf{Z}_{\bf Q}=\mathbf{Z}_{\bf Q}^H$, the $i$-th equation in \eqref{eq57}(a) can be simplified as
\begin{equation}\label{eq58}
\begin{aligned}
&2\rho
\big( \sum\limits_{j=1, j\neq i}^{N} \mbox{Re}\{ z_{ij}\}\sin(\alpha_{j}^o-\alpha_{i}^o )+\mbox{Im}\{z_{ij} \} \cos( \alpha_{j}^o-\alpha_{i}^o )\big)\\
&=\mbox{Im}\{ \gamma_9u_i\exp(-j\alpha_{i}^o)\}, i=1,2,\ldots,N
\end{aligned}
\end{equation}
where $u_i$ is the $i$-th element of $\mathbf{u}=\mathbf{G}^H\mathbf{w}$  Adding the $N$ equations in \eqref{eq58}, it easily seen that $\sum \limits_{i=1}^N\mbox{Im}\{ \gamma_9u_i\exp(-j\alpha_i^o\}=0$ but we know from \eqref{eq57}(b) that $\rho\sum \limits_{i=1}^N u_i\exp(-j\alpha_i^o)=\kappa$, where $\kappa\in\mathbbm{R}$. Therefore this implies that $\mbox{Im}\{\gamma_9\}=0$ or in other words, $\gamma_9 $ is purely real. Substituting this back into \eqref{eq57}(a) and following the same arguments as before, this is possible if trivially $\rho=0$ or $\gamma_9=0$, the former is false since $\rho=0$ does not solve \eqref{eq57}(b), therefore the latter must be true.
\end{proof}

{\bf  Interpretation of $\gamma_9=0$}. With $\gamma_9=0$, from \eqref{eq57}(a) we have that
\begin{equation} \label{finalkkts_cm}
\begin{aligned}
\mbox{Im}\left\{ \sum\limits_{q=1}^Q\mathbf{Z_q}(\mathbf{w}) \mathbf{s}_o \right\}&=\mathbf{0} \\
\mathbf{w}^H\mathbf{G}\mathbf{s}_o&=\kappa.
\end{aligned}
\end{equation}
The first equation in \eqref{finalkkts_cm} does not depend on $\rho$, but the second does. Therefore $\gamma_9=0$ does not imply that the constraint in \eqref{eq53} is inactive. Rather, this implies that the KKTs enforce the Capon constraint in \eqref{eq53} for the constant modulus waveform by varying the {\it unspecified} modulus parameter $\rho$. 

The result in Prop.~\ref{propos6} has some very interesting consequences. Using $\gamma_9=0$, the $N$ equations in \eqref{eq58} and therefore \eqref{eq57}(a), can be rewritten as a some linear matrix equation $\bar{\mathbf{Z}}_{\bf Q}\mathbf{p}_{\boldsymbol{\alpha}_{\bf o}}=\mathbf{0}$, where $\bar{\mathbf{Z}}_{\bf Q}\in\mathbbm{R}^{N\times\binom{N}{2}}$ and the vector $\mathbf{p}_{\boldsymbol{\alpha}_{\bf o}}=[\sin(\alpha_{2}^o-\alpha_{1}^o) ,\sin(\alpha_{3}^o-\alpha_{1}^o) \ldots, \sin(\alpha_{N}^o-\alpha_{N-1}^o), \cos(\alpha_{2}^o-\alpha_{1}^o),\ldots,\cos(\alpha_{N}^o-\alpha_{N-1}^o) ]^T$ i.e. has $\binom{N}{2}$ components consisting of sines and cosines of all possible differences of $\alpha_i^o-\alpha_j^o,\forall i, \forall j\neq i$. In other words, $\mathbf{p}_{\boldsymbol{\alpha}^{\bf o}}\in \mbox{null}\big( \bar{\mathbf{Z}}_{\bf Q}\big)$. The rank of $\bar{\mathbf{Z}}_{\bf Q}$ is not easy to calculate here but its maximum value is $N$. Therefore from the rank-nullity theorem, $\dim(\mbox{null}(\bar{\mathbf{Z}}_{\bf Q}))\geq N(N-2)$. Clearly there could exist multiple vectors which are in this null space but we are not certain if this translates to multiple solutions of $\boldsymbol{\alpha}_o$ from this linear equation alone. Nonetheless, if multiple solutions exist to  this linear equation, they must also satisfy \eqref{eq57}(b) to be considered as  a solution to \eqref{eq53}. In any case the optimal solution(s) are in, $\mathcal{C}_{\boldsymbol{\alpha}^{\bf o} }\subset\mathbbm{R}^{N}$, with
\begin{align}\label{eq59}
\mathcal{C}_{\boldsymbol{\alpha}^{\bf o} }=\{\boldsymbol{\alpha}^o:\mathbf{p}_{\boldsymbol{\alpha}_{\bf o}} \in \mbox{null}( \bar{\mathbf{Z}}_{\bf Q}),\sum\limits_{i=1}^N u_i^{\ast}\exp(j \alpha_i^o)=\frac{\kappa}{\rho} \}. 
\end{align}
It remains to be seen if $\mathcal{C}_{\boldsymbol{\alpha}^{\bf o} }$ is  singleton, or comprises many elements, but we are optimistic that it would not turn out to be empty.
\subsection{Practical Considerations: Classical STAP v.s Waveform adaptive STAP}
Here we addresses practical considerations on the fast time-slow time model in STAP which aids in the waveform design and compare this with the classical model in STAP (slow time).

{\bf Hardware} The fast-time slow-time model in STAP does not necessitate newer hardware nor does it require any modifications to the existing hardware. It does however assume that the current state-of-art permits arbitrary waveform generation and adaptive transmitting capabilities \cite{cochran2009waveform}.

{\bf Computational complexity} The inclusion of the waveform causes the correlation matrices to have larger dimension. Inverting large matrices are computationally prohibitive. Classical STAP requires inverting a complex $ML\times ML$ matrix which has a complexity of $O((ML)^{2.373})$-$O((ML)^{3})$ \cite{VirginiaWilliams2012}. Waveform adaptive STAP requires inverting complex $NML\times NML $ complex matrices which has a computational complexity of $O((NML)^{2.373})$-$O((NML)^{3})$ \cite{VirginiaWilliams2012}.

{\bf Training data} Due to the larger dimensions of the correlation matrices by inclusion of the waveform, it suddenly appears, albeit deceivingly, that more training data (from more neighboring range cells) are needed to estimate the correlation matrices. This is not true since inclusion of waveform simply includes the fast time samples. Hence the fast-time slow-time model uses the raw data prior to pulse compression or matched filtering, hence the training data requirements is identical to that required in the classical STAP case. Note that we are not interested in resolving targets within the pulse duration but rather outside it.


\begin{figure*}[htbp!]
\centering
  \subfloat[ ] { \includegraphics[scale=0.5]{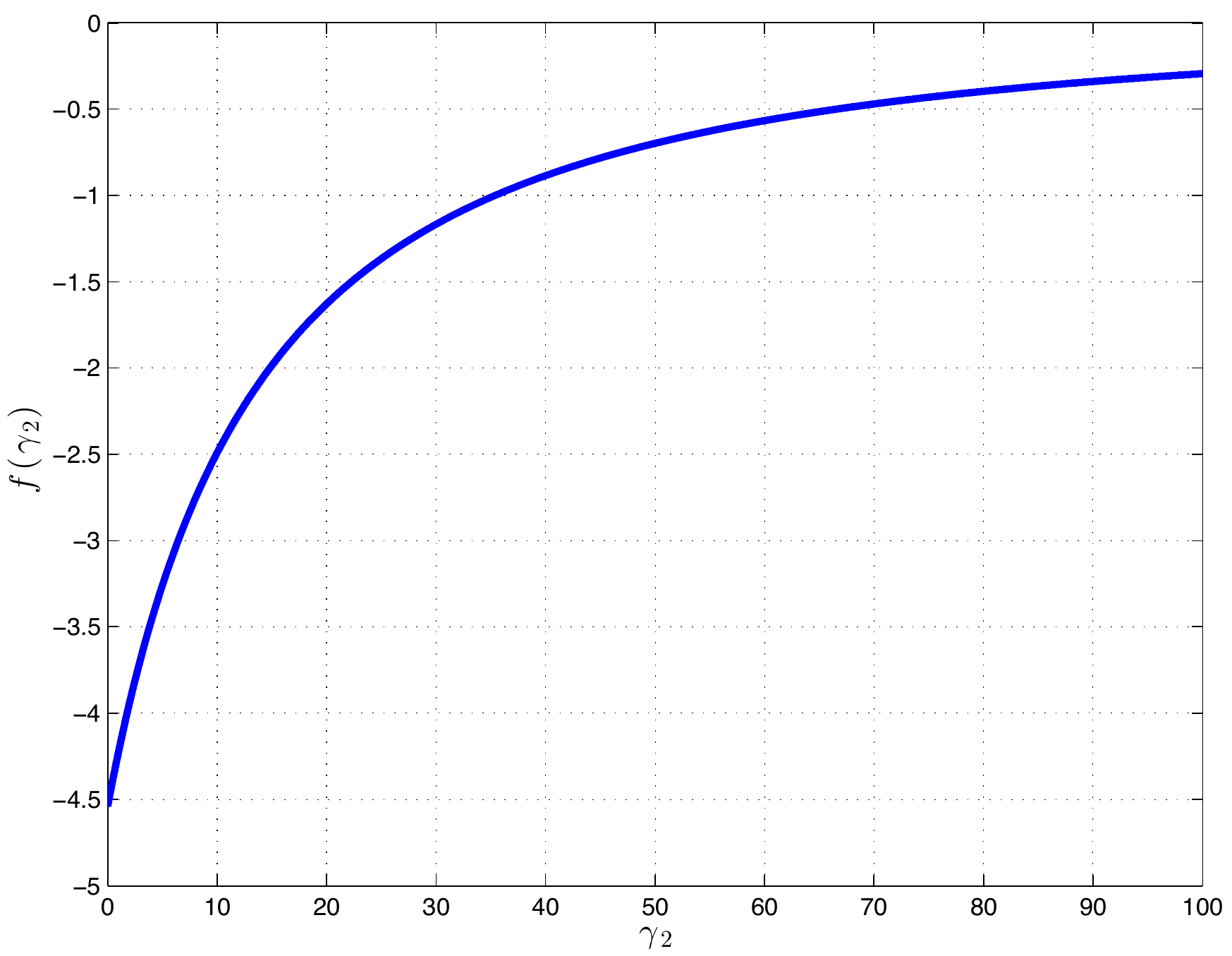} } 
  \subfloat [ ]{ \includegraphics[scale=0.5]{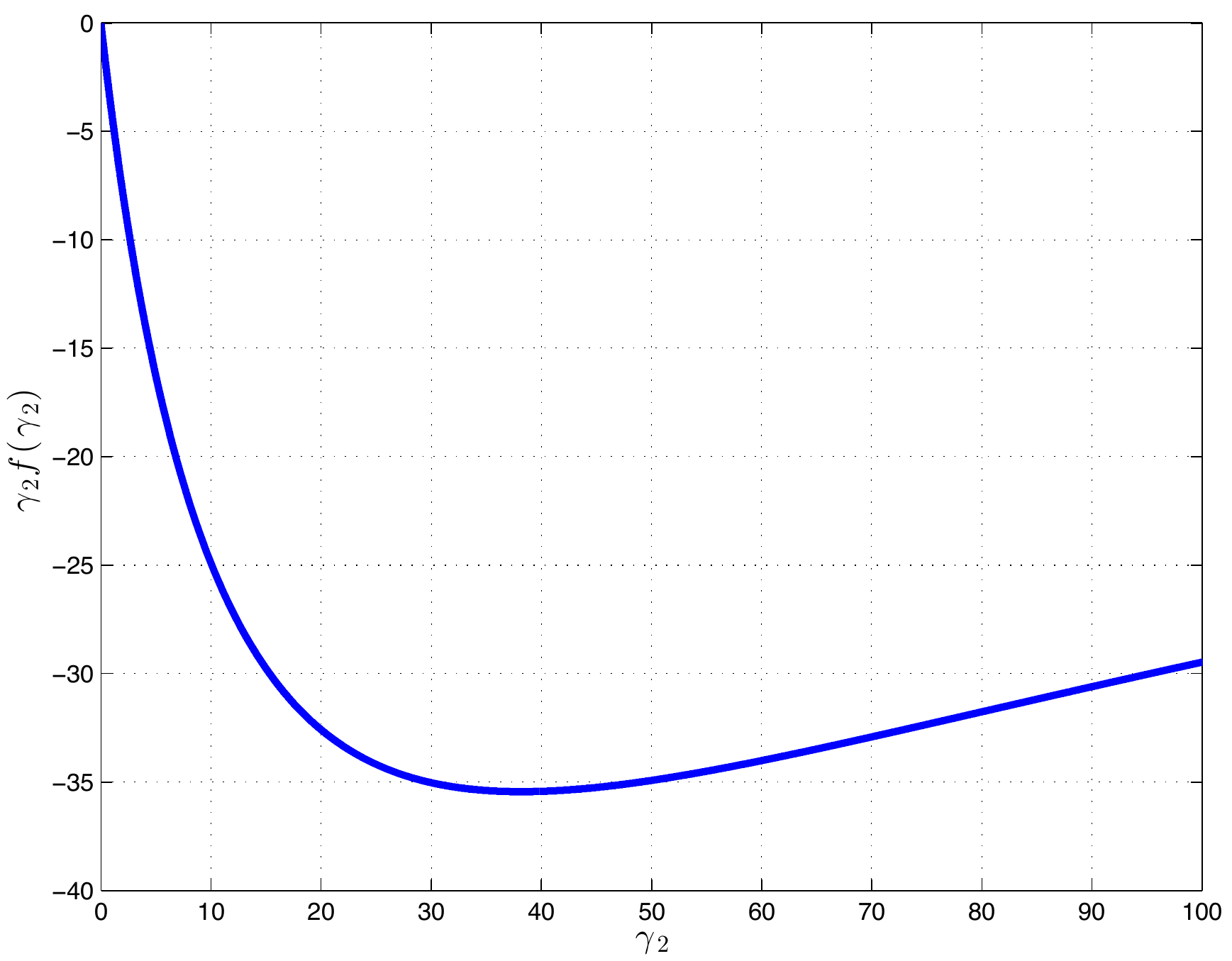} } \\
  \subfloat[ ] { \includegraphics[scale=0.5]{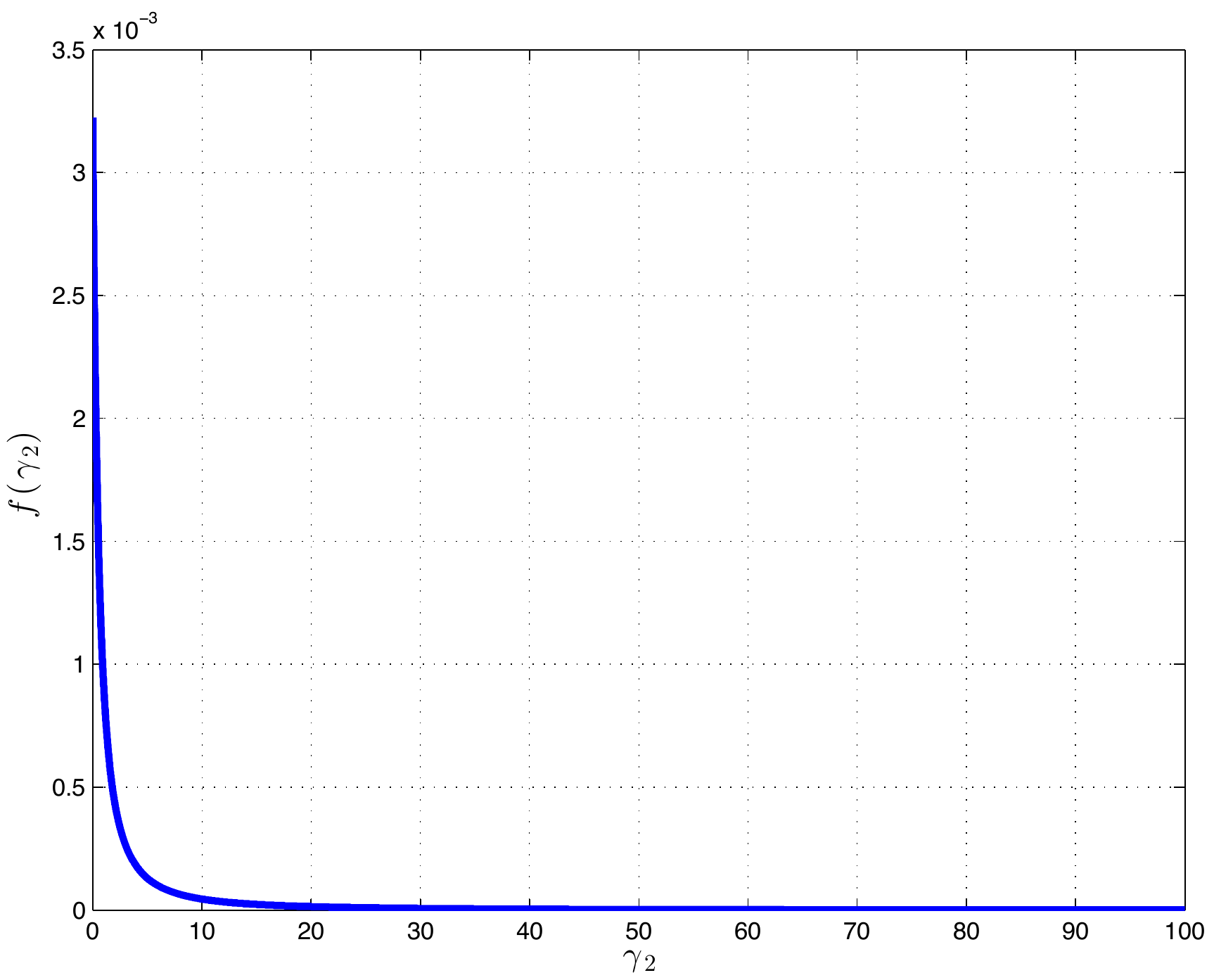} } 
  \subfloat [ ]{ \includegraphics[scale=0.5]{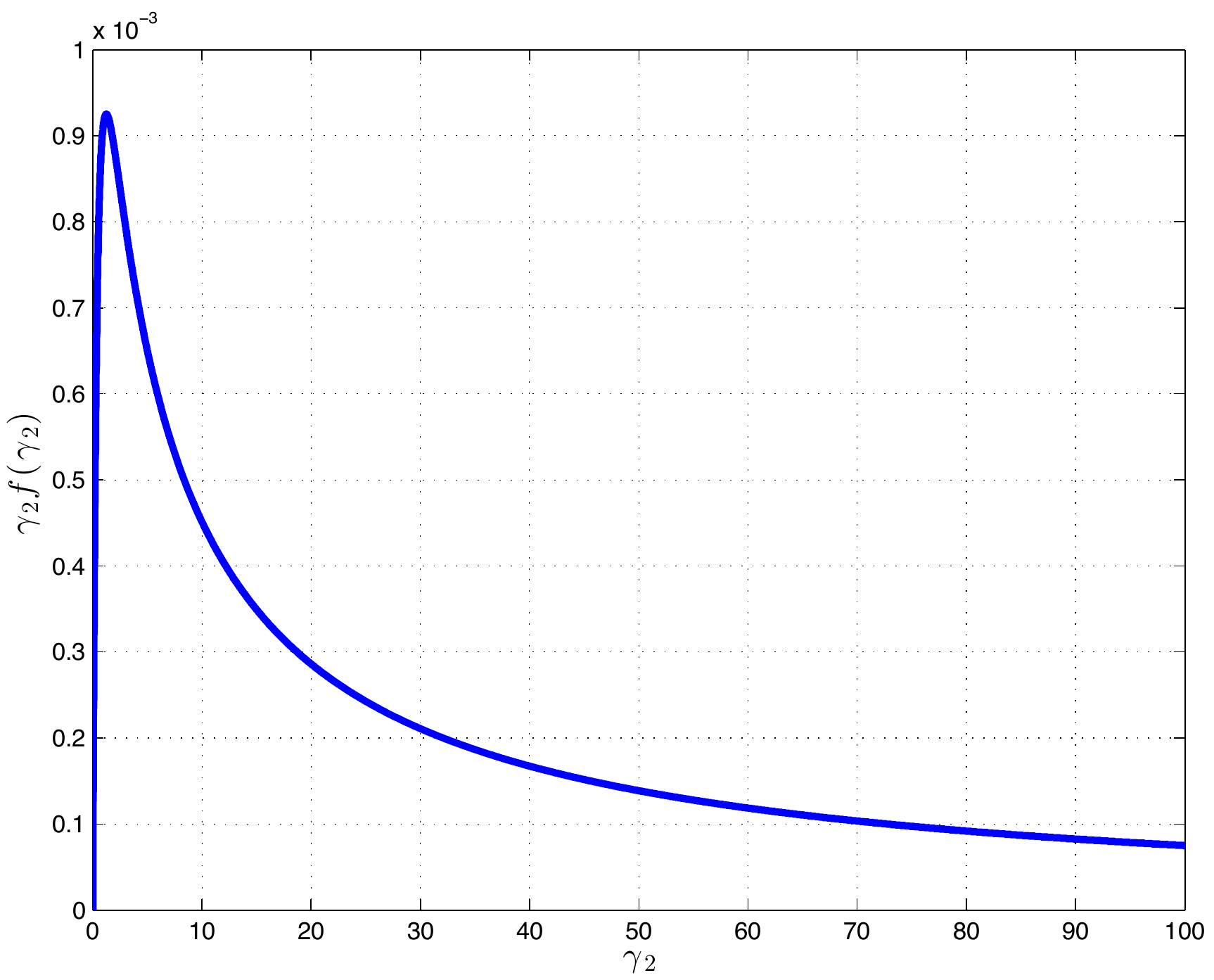} } \\
\caption{ Simulations supporting Prop.~\ref{propos3}, x-axis $\gamma_2$. Monotone increasing (a) $f(\gamma_2)$ and corresponding (b) $\gamma_2f(\gamma_2)$. Monotone decreasing (c) $f(\gamma_2)$ and corresponding (d) $\gamma_2f(\gamma_2)$.}
\label{lagrasupport1}
\end{figure*}

\begin{figure*}[htbp!]
\centering
  \subfloat[ ] { \includegraphics[scale=0.5]{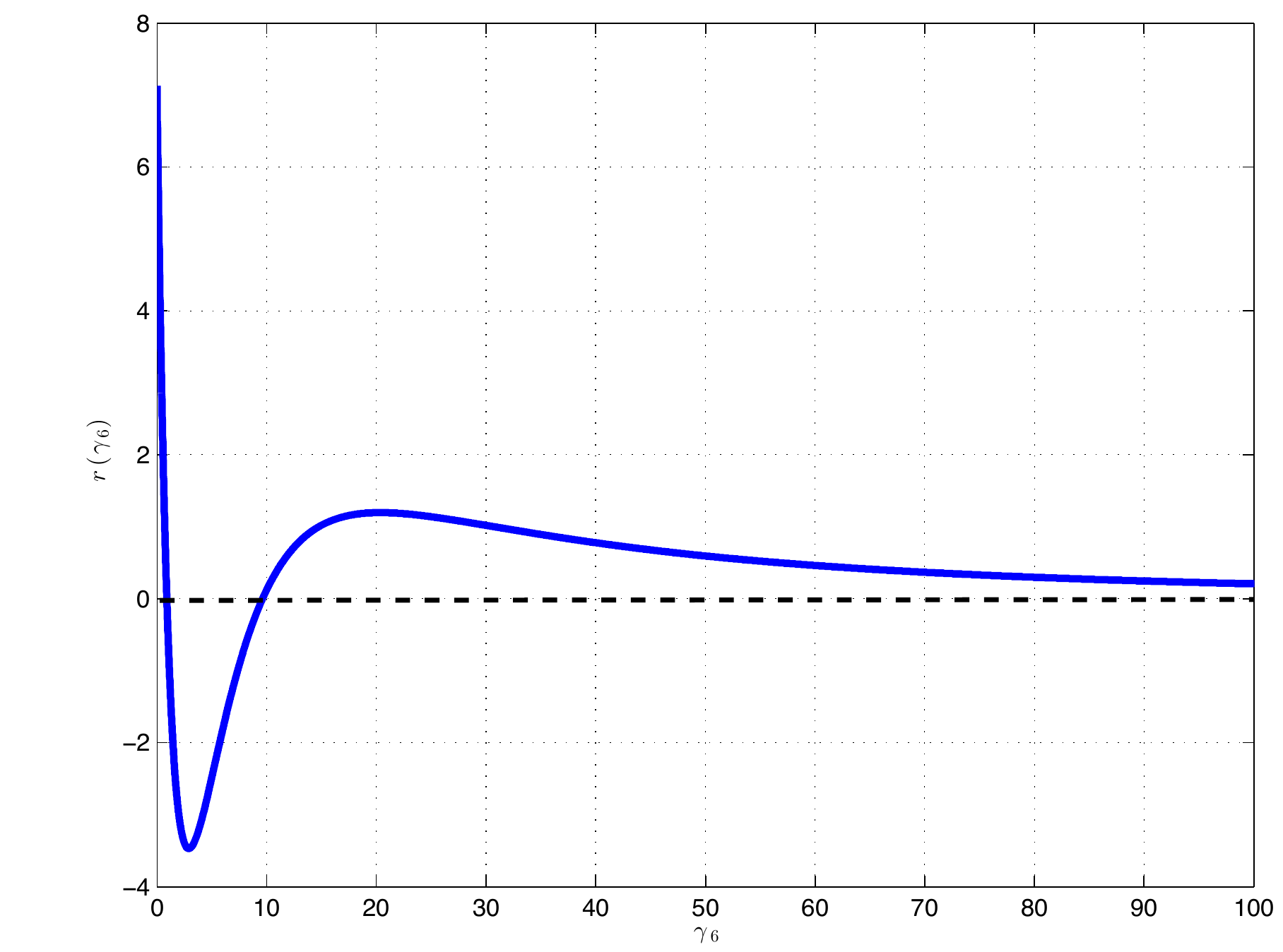} } 
  \subfloat [ ]{ \includegraphics[scale=0.5]{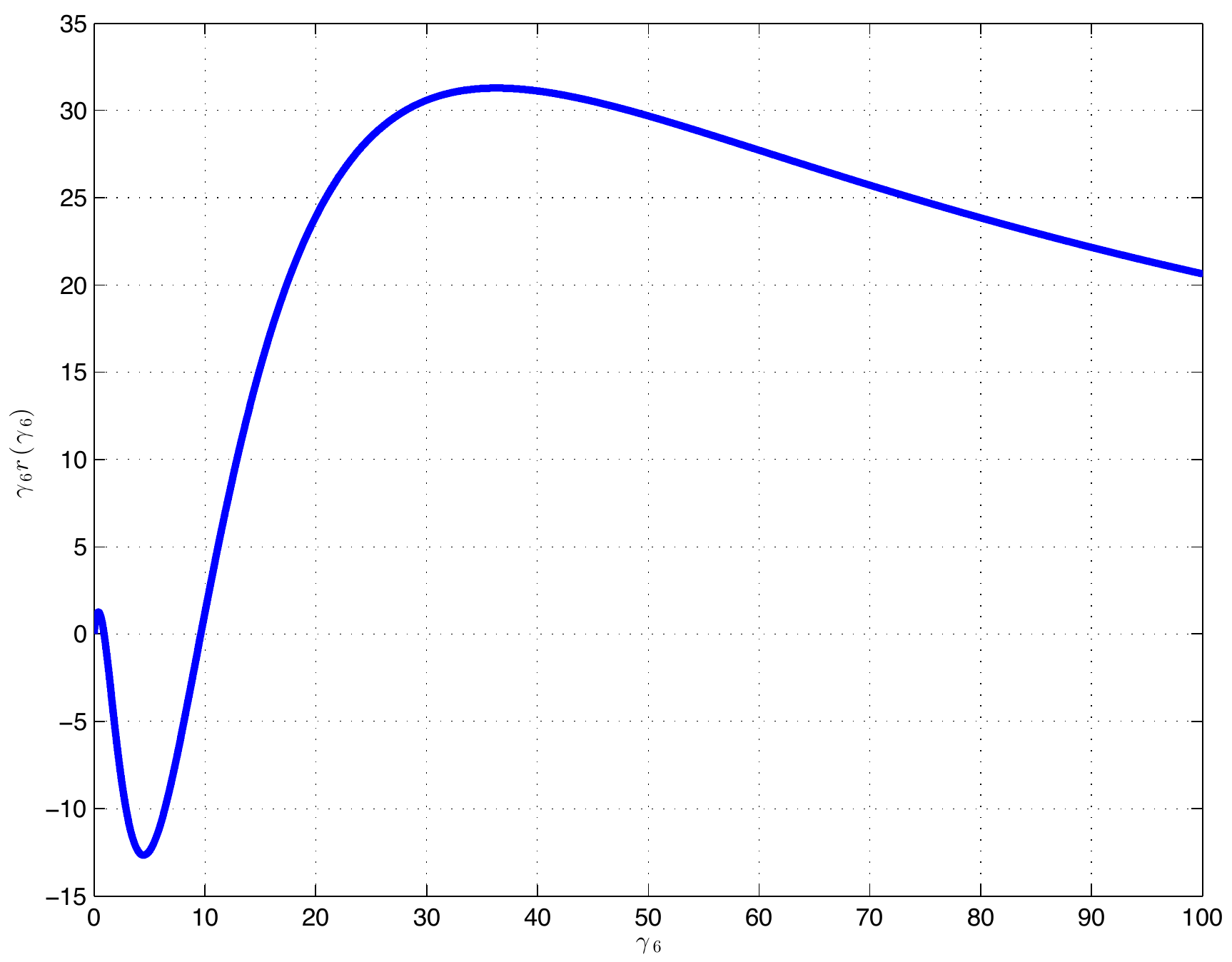} } \\
  \subfloat[ ] { \includegraphics[scale=0.5]{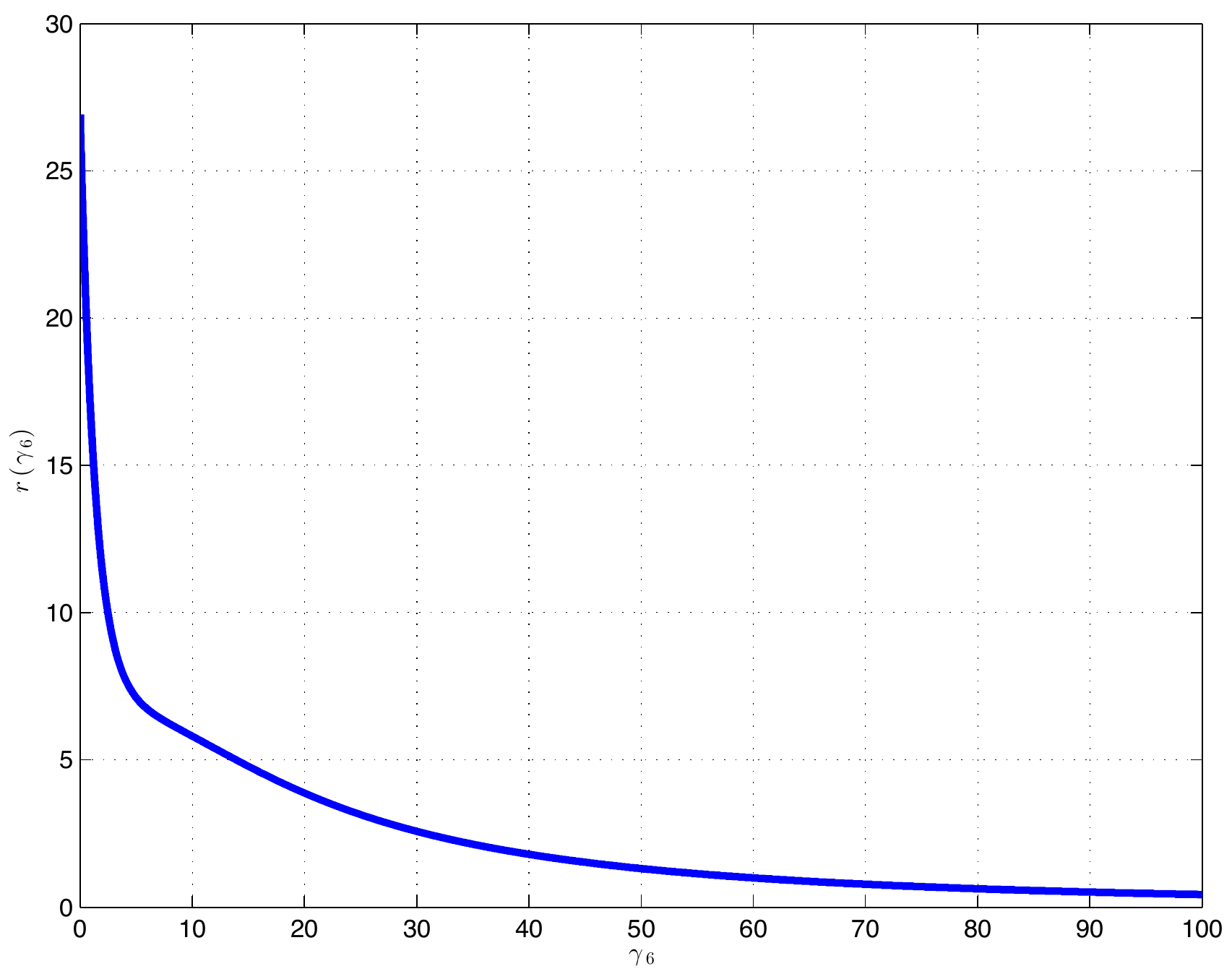} } 
\caption{ Simulations supporting Rem.~\ref{propos5}, x-axis $\gamma_6$. Example showing one zero crossing of (a) $r(\gamma_6)$ and corresponding (b) $\gamma_2r(\gamma_6)$. Monotone decreasing example for $P_o>>\kappa^2$ in (c) $r(\gamma_6)$.}
\label{lagrasupport2}
\end{figure*}

\begin{figure}[htbp!]
\centering
\includegraphics[scale=0.5]{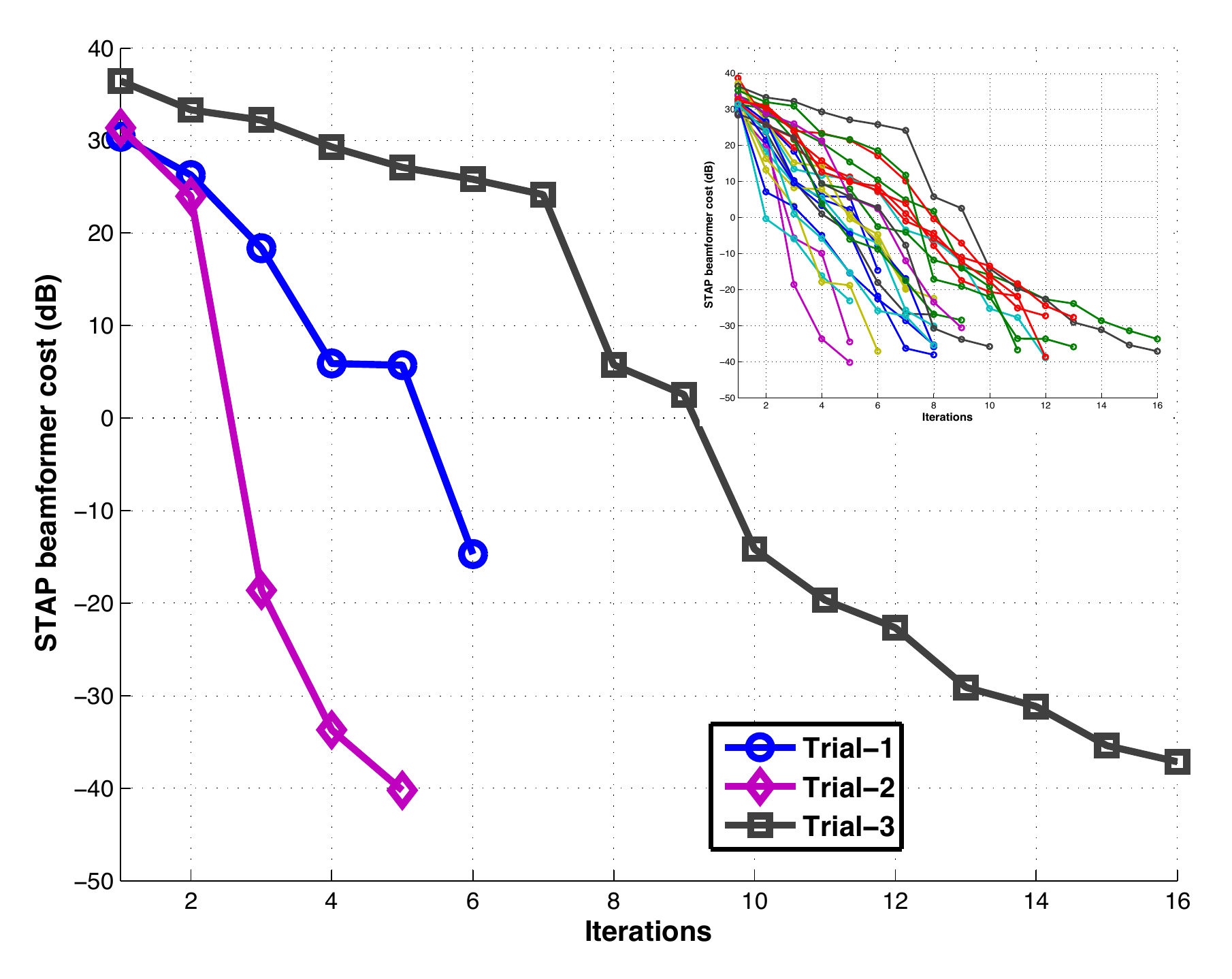}
\caption{Constrained alternating minimization: objective costs vs. iterations for 3 random, independent waveform initializations (inset: for 25 random initializations).}
\label{fig4}
\end{figure}
\begin{figure*}[htbp!]
\centering
  \subfloat[ ] { \includegraphics[scale=0.5]{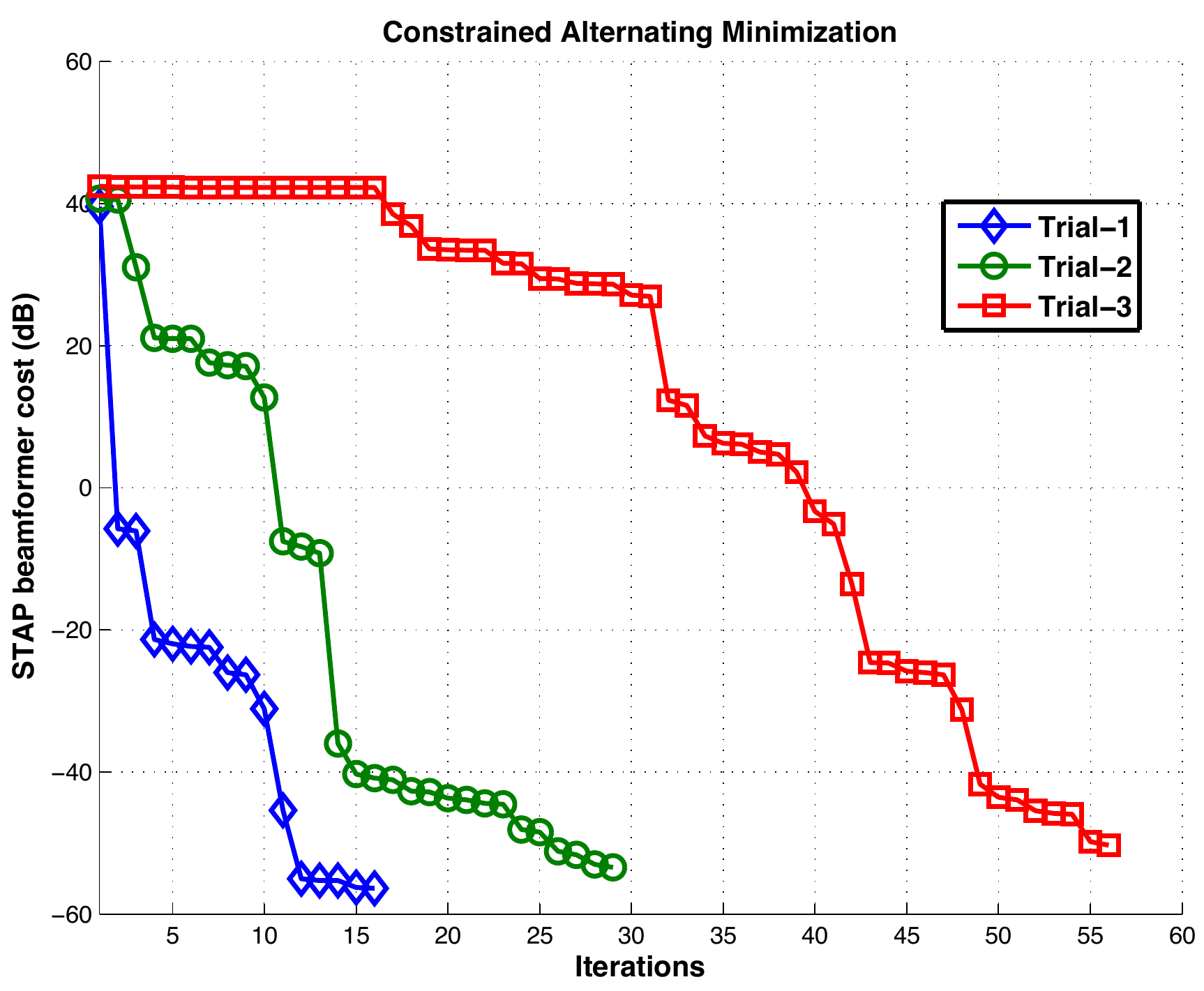} }
  \subfloat [ ]{ \includegraphics[scale=0.5]{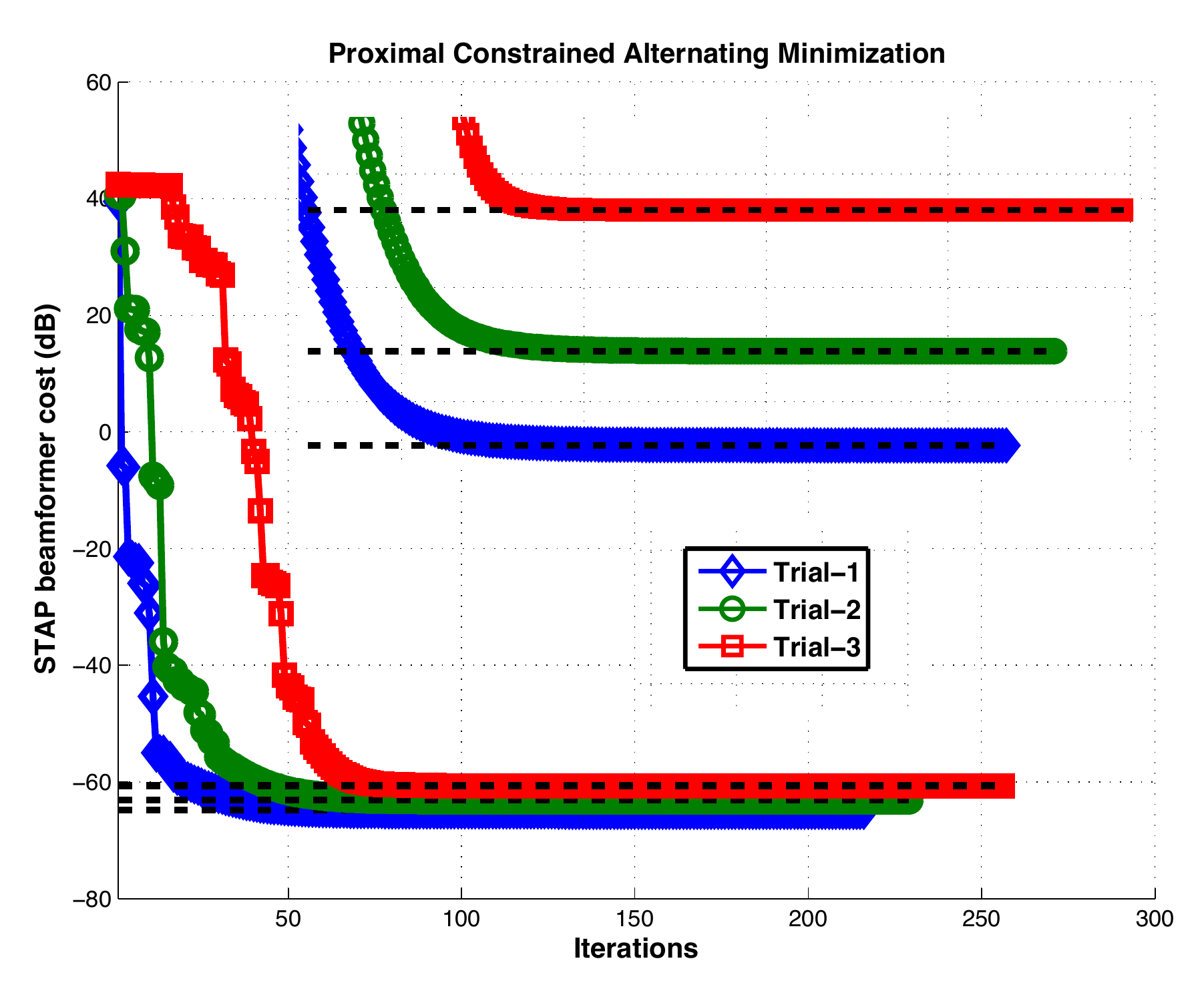} }
\caption{(a) Constrained alternating minimization, (b) Proximal constrained alternating minimization (inset: magnified), minimum eigenvector waveform (dashed black).}
\label{fig5}
\end{figure*}
\begin{figure}[htbp!]
\centering
  \includegraphics[scale=0.5]{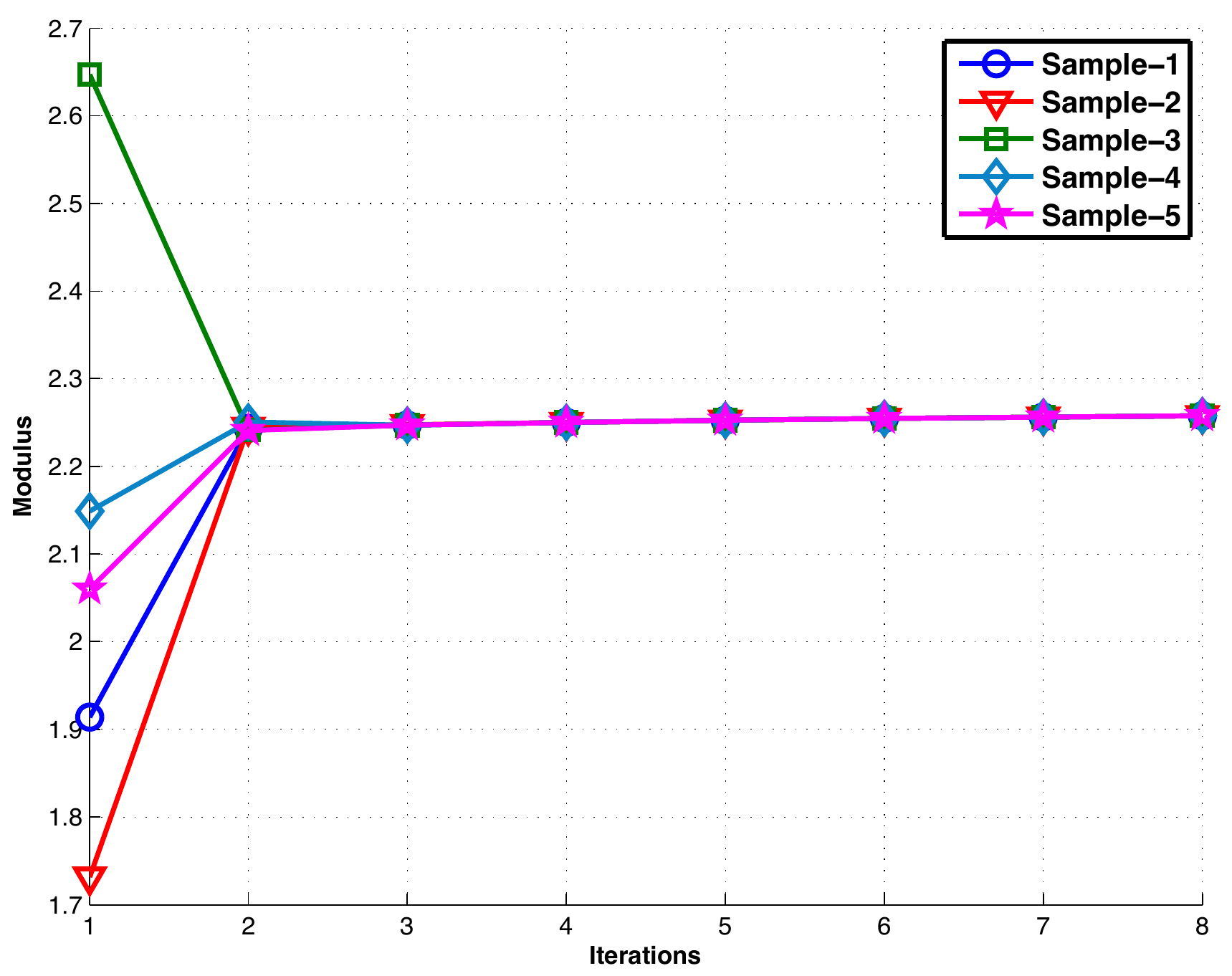} 
\caption{Convergence of non const. mod initial waveform to  a con. mod Con.  waveform:}
\label{fig6}
\end{figure}
\begin{figure*}[htbp!]
\centering
  \subfloat[ ] { \includegraphics[scale=0.5]{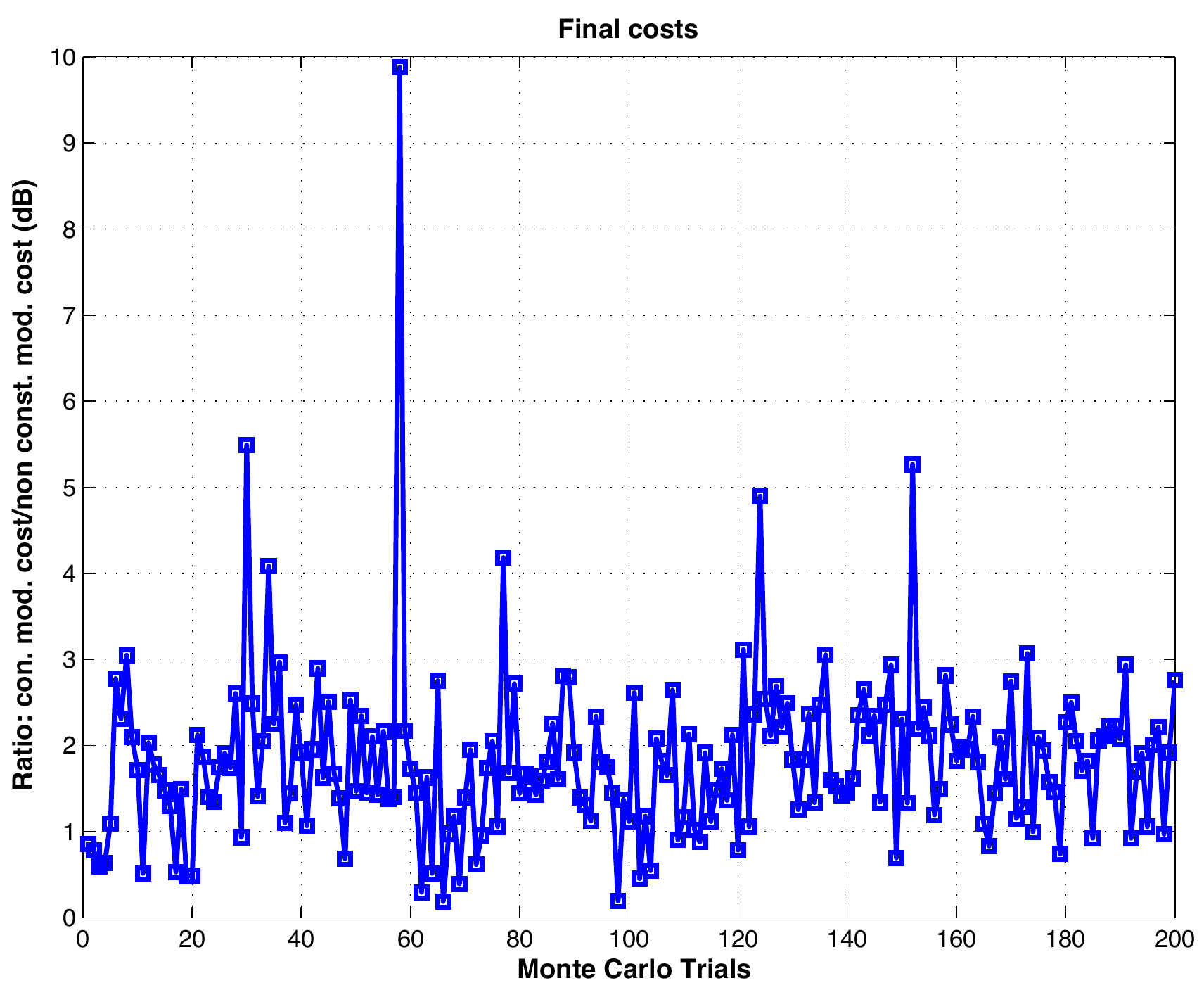}}
  \subfloat [ ]{ \includegraphics[scale=0.5]{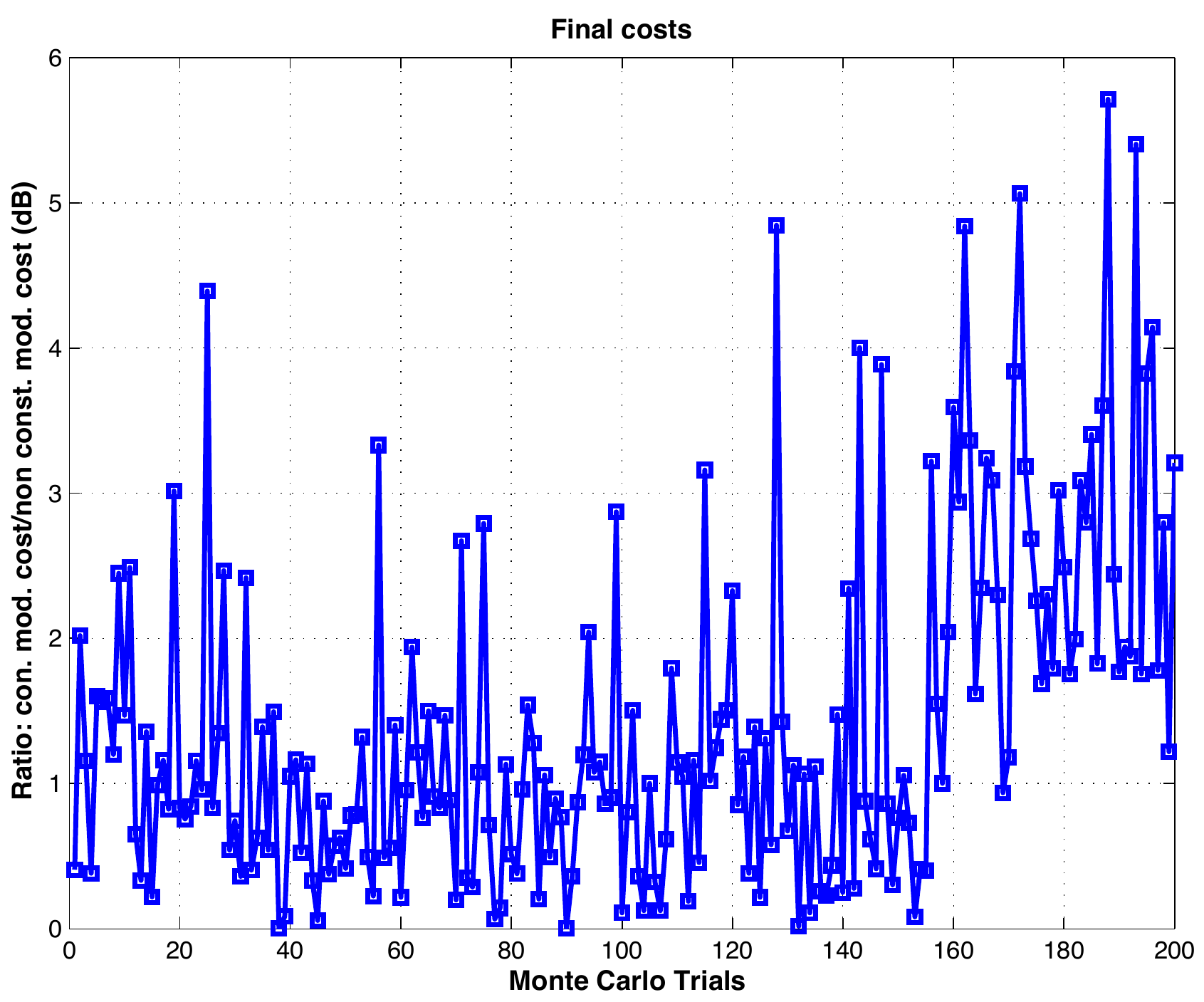} }
\caption{Constant modulus waveform design comparison with non const. mod. design, 200 trials initialized with: (a) random non-const. mod. Gaussian waveforms (b) random unit modulus waveforms, with phase drawn uniformly from $[-\pi,\pi]$.  }
\label{fig7}
\end{figure*}
\begin{figure}[htbp!]
\centering
 \subfloat[ ] { \includegraphics[scale=0.5]{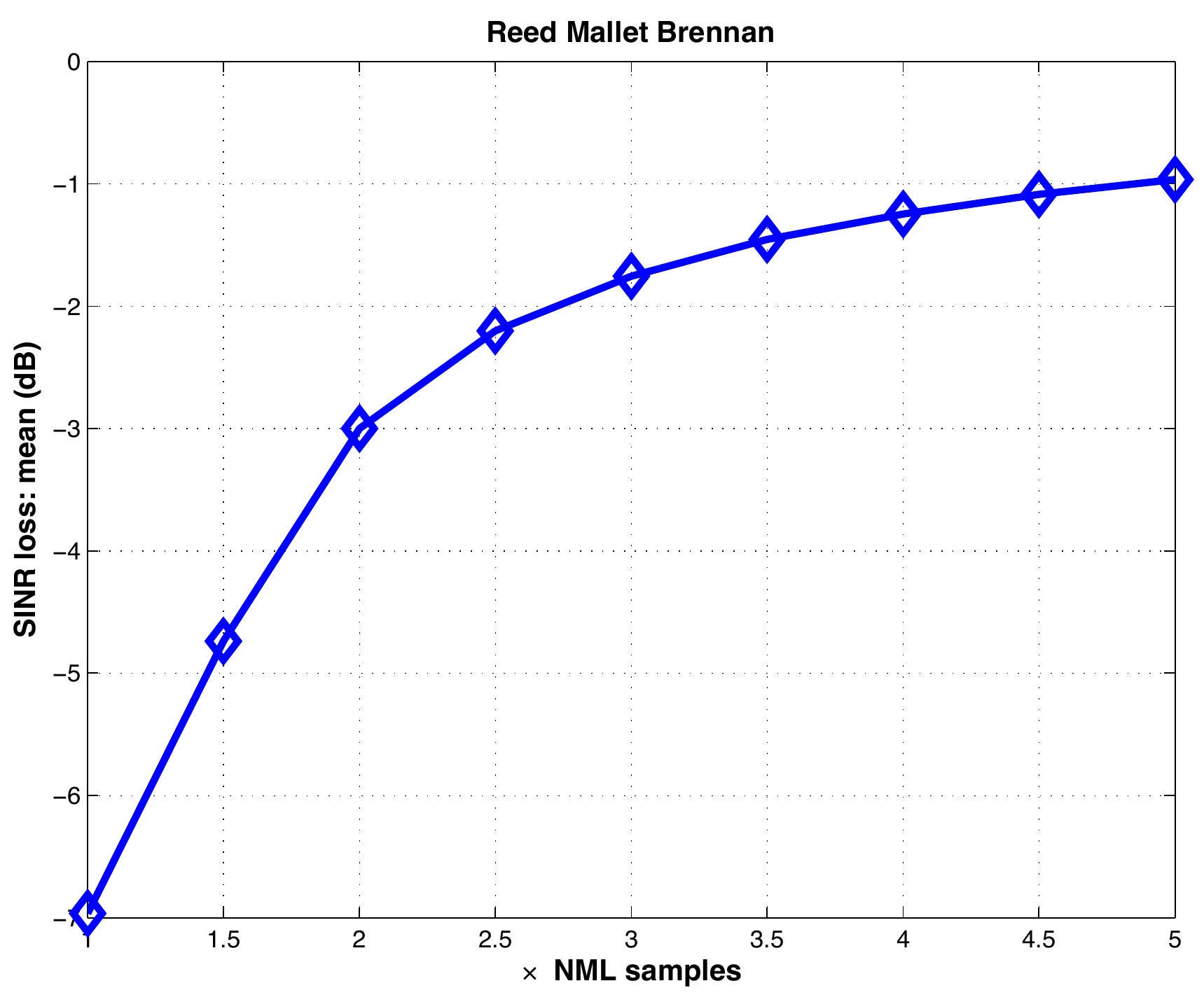} }
  \subfloat [ ]{ \includegraphics[scale=0.5]{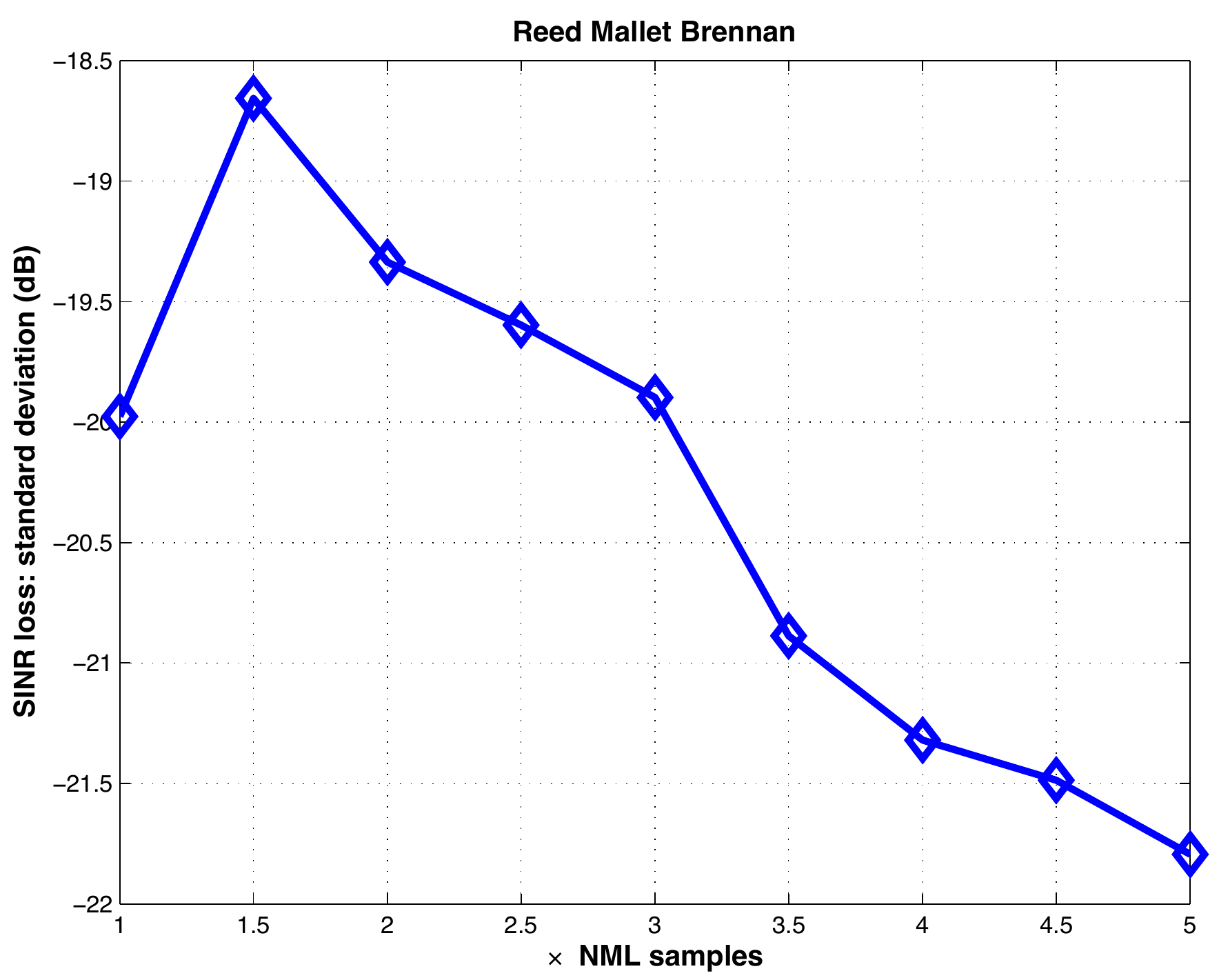} }
\caption{Oracle: Reed Mallet Brennan rule.}
\label{fig8}
\end{figure}
\begin{figure}[htbp!]
\centering
  \subfloat[ ] { \includegraphics[scale=0.5]{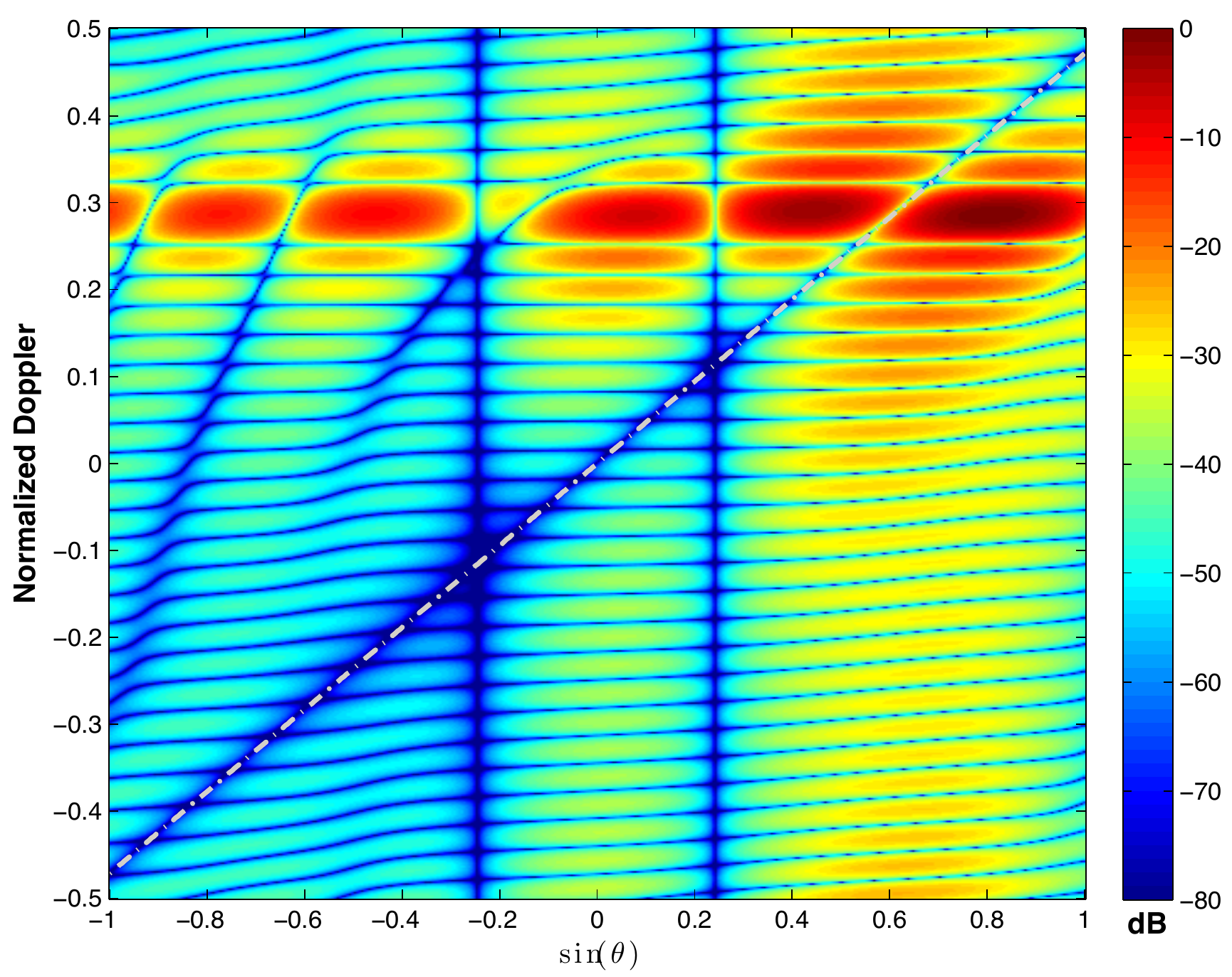} }
  \subfloat [ ]{ \includegraphics[scale=0.5]{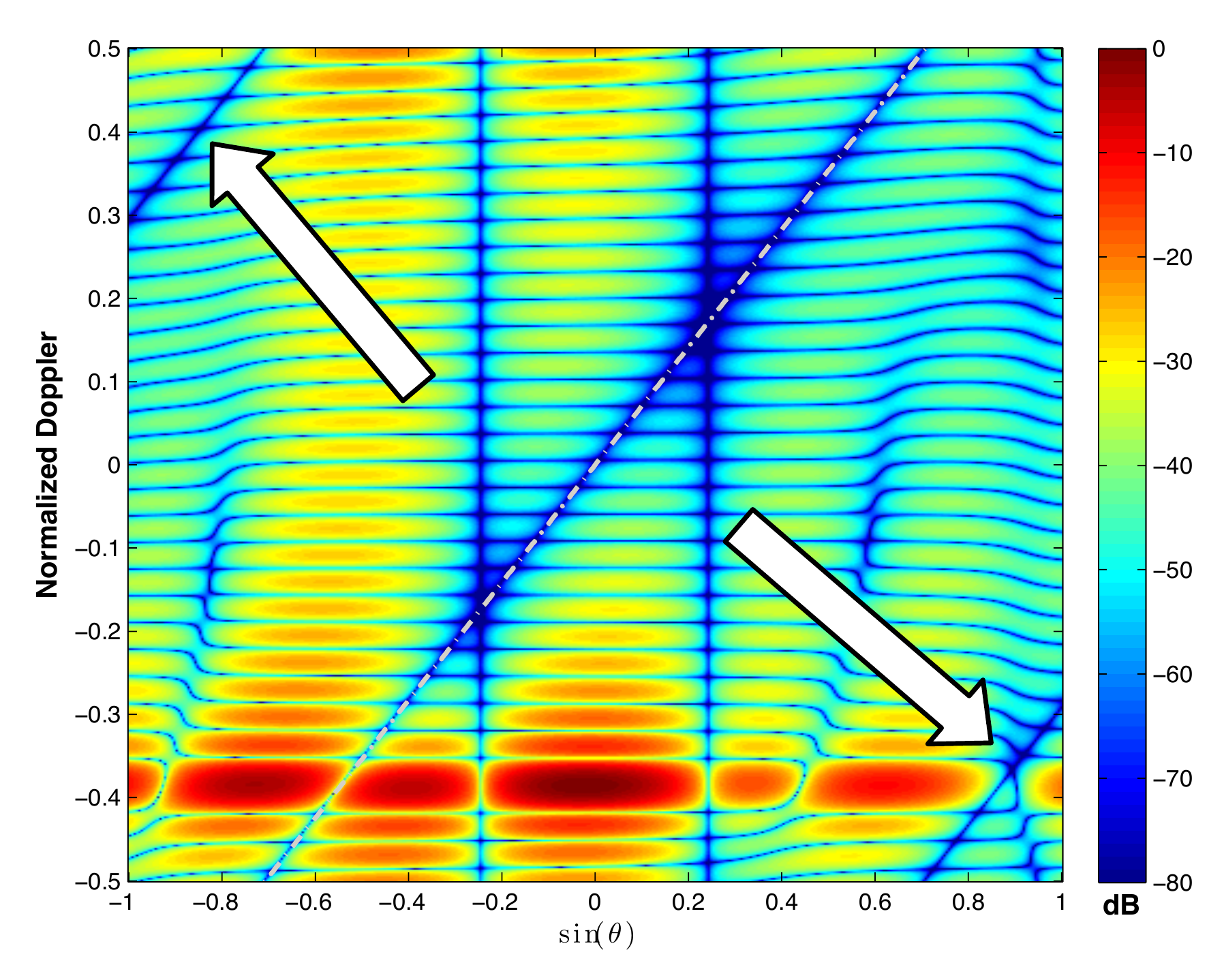} }
\caption{Adapted patterns using designed waveform from alternating minimization, dashed line is the Doppler as a function of angle predicted by theory. In, (a) no clutter Doppler ambiguities, (b) clutter Doppler ambiguities shown with arrows.}
\label{fig9}
\end{figure}
\begin{figure*}[htbp!]
\centering
  \subfloat[ ] { \includegraphics[scale=0.5]{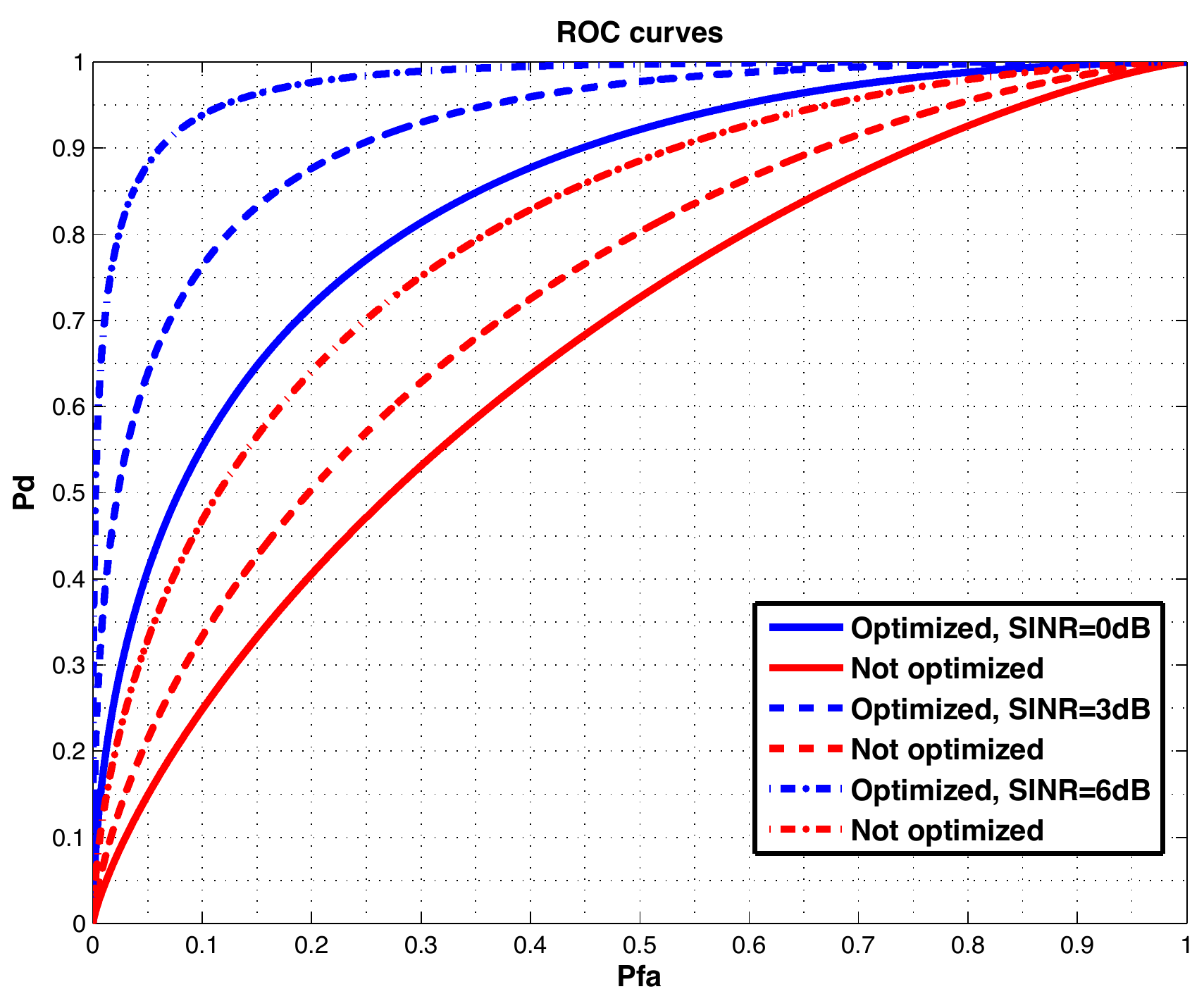} } 
  \subfloat [ ]{ \includegraphics[scale=0.5]{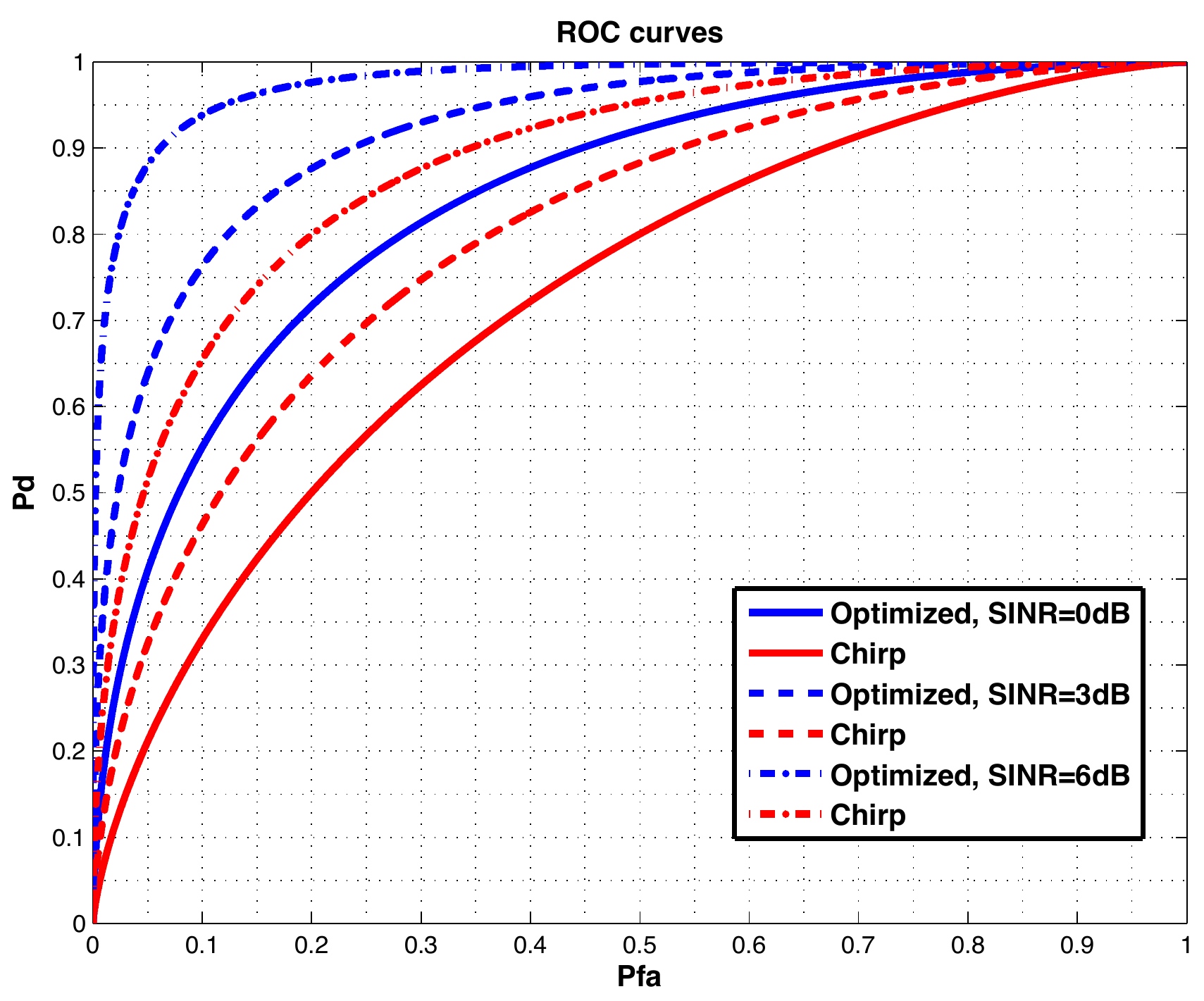} }
\caption{ROC (a) non con. mod design, (b) con. mod. design }
\label{fig10}
\end{figure*}
\begin{figure*}[htbp!]
\centering
  \subfloat[ ] { \includegraphics[scale=0.5]{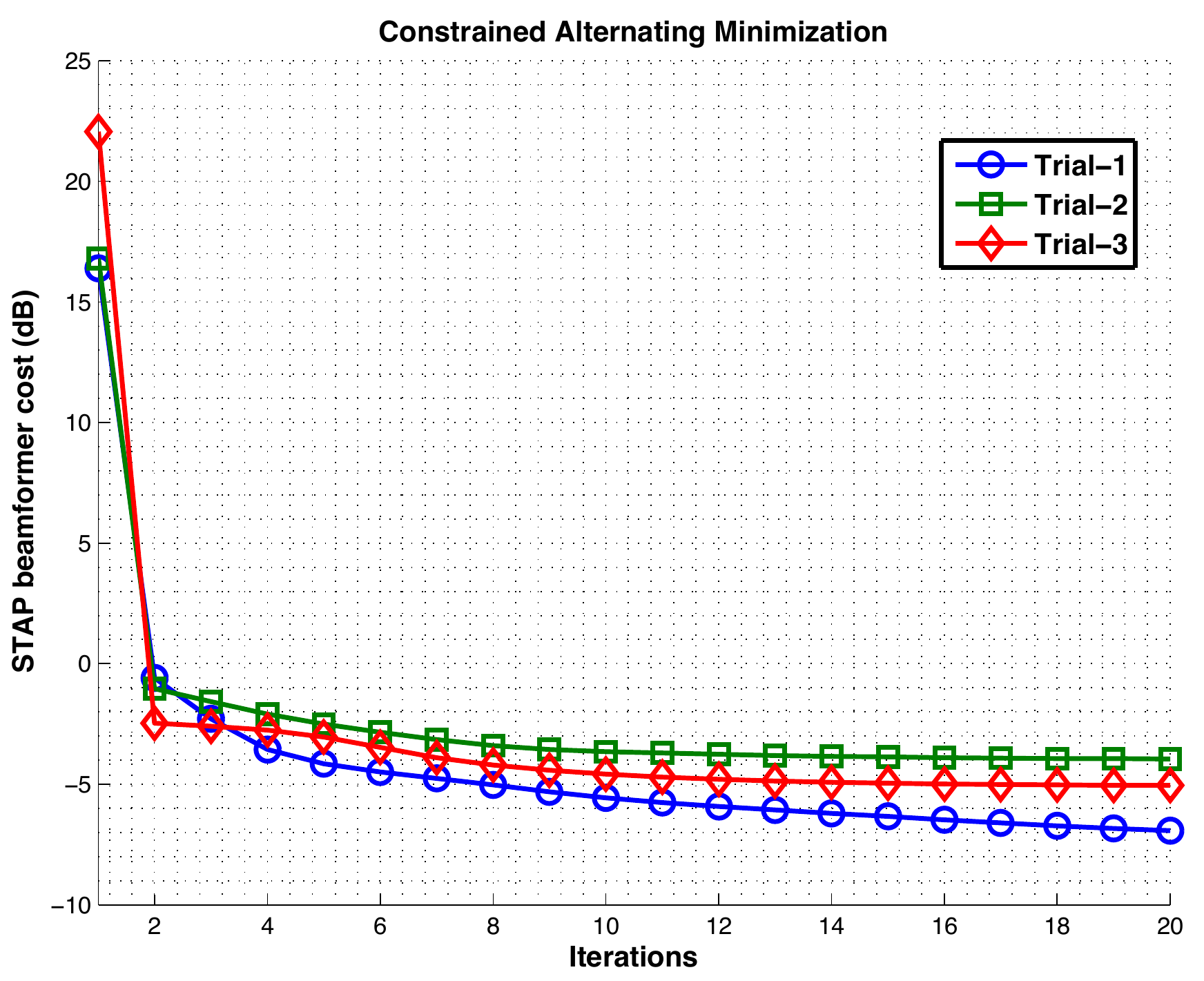} } 
  \subfloat [ ]{ \includegraphics[scale=0.5]{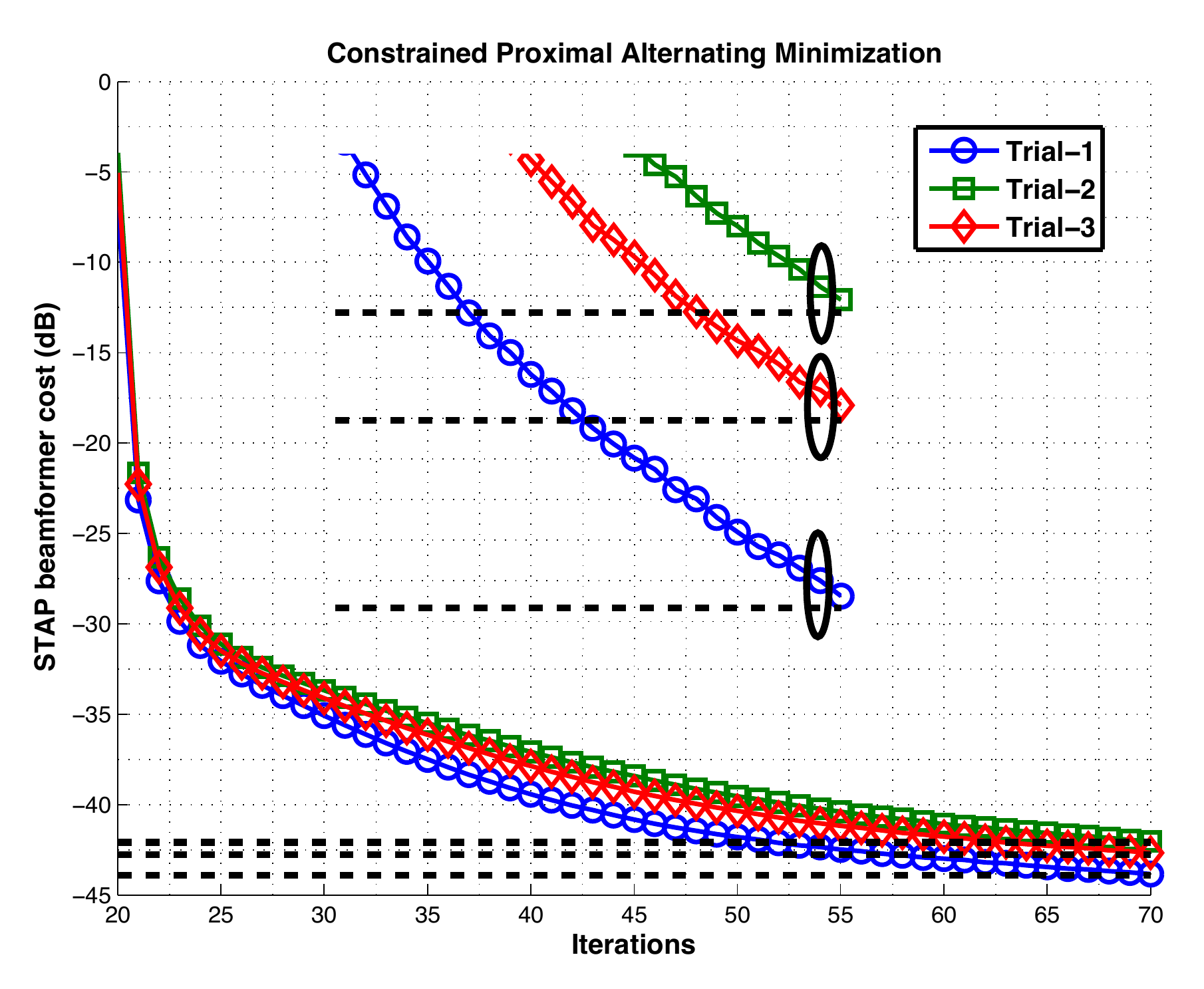} }
\caption{Rank deficient waveform adaptive STAP (a) constrained alternating minimization, (b) constrained proximal alternating minimization.}
\label{fig11}
\end{figure*}
\section{Simulations}\label{sec:simulations}
First we will addresses simulations not specific to radar.

\subsection{Simulations supporting: Prop.~\ref{propos3} and Rem.~\ref{propos5}}

We ran simulations with random $z_n$ and random $d_n$ to analyze $f(\gamma_2)$ and $\gamma_2f(\gamma_2)$ numerically.   In our extensive simulations we chose $z_n$ from complex normal distributions with different means and different variances. Since $d_n>0$ for all $n$, we used uniform distributions with different supports on the positive real axis excluding zero. We show only two representative simulation results for the monotonically increasing and decreasing cases in Fig.~\ref{lagrasupport1}(a)(c), respectively. The corresponding function $\gamma_2f(\gamma_2)$ are also shown in Fig.~\ref{lagrasupport1}(b)(d) for the two cases. 

Simulations for supporting Rem.~\ref{propos5} is presented next. Some parameters specifying the function $r(\gamma_6)$ were simulated randomly with the identical distributions used as in generating Fig.~\ref{lagrasupport2}. The parameter $\kappa=2,P_o=10$ was used in generating Fig.~\ref{lagrasupport2}(a), the function $\gamma_6r(\gamma_6)$ is also shown in Fig.~\ref{lagrasupport2}(b). As such, it is noted that $P_o=10$ is  a a contrived example, typical radar applications will require $P_o$ to be in several hundred KW or several Hundred MW. 

The zero crossing is the intersection of the dashed line (black) with the blue curve in Fig.~\ref{lagrasupport2}(a). Now using $P_o=20$ and keeping the other parameters fixed we obtain Fig.~\ref{lagrasupport2}(c) which shows that $r(\gamma_6)$ is monotonic decreasing whose limit at $\infty$ is 0.

{\bf Radar Specific simulations:} Here onward, some  parameters are common to all the simulation examples and are stated now. The simulation parameters are in SI units  unless mentioned otherwise. To reduce computation complexity while inverting large matrices and computing their eigen-decompositions, we considered the number of, sensors, waveform transmissions, and fast time samples in the waveform as $M=5,L=32, N=5$, respectively. The carrier frequency was chosen to be 1GHz, and the radar bandwidth was 50MHz. The element spacing $d=\lambda_o/2$. 
\subsection{Constrained alternating minimization}
The noise correlation matrix was assumed to have a correlation function given by $\exp(-|0.005n|),\; n=0,1,\ldots,NML$. Two interference sources were considered at $(\theta=0.3941,\phi=0.3)$ and at $(-0.4941,0.3)$. Both these interference sources had identical  discrete correlation functions given by $0.2^{|n|},\; n=\pm 0, \pm 1,\ldots$. To simulate clutter we considered two clutter patches, consisting of five scatters each. The clutter correlation functions corresponding to the two patches were $\exp(-0.2|p|) \mbox{ and } \exp(-0.1|p|), \;p=\pm 0,\pm 1,\ldots,\pm P$. The rest of the parameters are identical to those used in \cite{Setlurradar2013}. 

In Fig.~\ref{fig4}, the STAP beamformer objective vs. iterations are shown for 3 independent, random waveform initializations but the {\it inset} shows 25 independent initializations or trials. The alternating minimization was initialized with waveforms whose fast time samples are chosen independently from a standard complex Gaussian distribution. The algorithm was terminated as soon as the current waveform iterate invalidated the set power constraint. From the figure and its inset it is clear that the STAP beamformer output is non-increasing thereby validating the monotonicity property of this algorithm.  More importantly from Fig.~\ref{fig4}, we see that the final objective value and the iterations to reach it for each trial are different from one another, attributed to the joint non-convexity of the objective w.r.t. $\mathbf{w}$ and $\mathbf{s}$. Sensitivity to the random initialization is therefore duly noted.

\subsection{Constrained proximal alternating minimization}
All the simulations parameters are identical to the previous case. The constrained alternating minimization was initialized with random waveforms as in Fig.~\ref{fig4}, immediately followed by its proximal counterpart. The termination of the former algorithm was identical to the previous case, then, the latter was run for 200 iterations. Three representative trials are shown in Fig.~\ref{fig5}(a)(b), for the constrained alternating minimization and its proximal counterpart. In Fig.~\ref{fig5}(b), the dashed black lines are the final objective values obtained from the min. eigenvector waveform having the same energy as its proximal counterpart. For the three trials and not surprisingly, the proximal objective value, for all practical purposes, is identical to that obtained from the waveform derived from \eqref{lsubopt} as evidenced from the {\it inset}. Therefore validating the implementation of both the constrained as well as its proximal counterpart. From Fig.~\ref{fig5}(b) and unlike Fig.~\ref{fig4}, three accumulation points w.r.t. the objective are clearly visible for the three trials indicating {\it strong convergence}. 
\subsection{Constant modulus}
The constant modulus algorithm was implemented numerically via the KKTs (i.e. \eqref{eq57}) and using the results from Prop.~\ref{propos6}. The simulation parameters are identical to the two previous scenarios. In Fig.~\ref{fig6}, the modulus of the fast time waveform samples vs. iterations are shown for the constant modulus alternating minimization algorithm. As seen from this figure, the algorithm was initialized with a non-constant modulus waveform. For this random initialization, convergence to a constant modulus is achieved in three iterations or less. We have however encountered cases where the algorithm has not converged for several iterations. Nevertheless this problem was not encountered when the algorithm was initialized with  a random constant modulus waveform. Thus in practice, it is advocated that this algorithm be initialized with  an arbitrary constant modulus waveform, viz. a chirp, rectangular pulse, etc..

The ratio of the final objective for the constant modulus algorithm to the objective for the non-constant modulus waveform design using the constrained alternating minimization is seen in Fig.~\ref{fig7}(a)(b) for 200 random waveform initializations. After convergence, not unexpectedly, the constant modulus objective is more than the non-constant modulus objective. This trend is readily observed from Fig.~\ref{fig7}(a)(b) for the 200 trials. This is to be expected since constant modulus waveforms are a subset of their non-constant modulus counterparts. In particular, the amplitude is constrained temporally in the constant modulus design, while the phase is allowed to be optimized. Whereas, the phase and amplitude are both optimized the non-constant modulus design. From these figures we can see that on one end, this ratio is as much as 10dB, and on the other it is almost 0dB. Nonetheless on the average, the non-const. modulus waveforms have lower objective values than objective values derived from the const. modulus waveforms. 
\subsection{Oracle sample support requirements}
The ideal SINR is $\tfrac{\rho_t^2|\mathbf{w}^H_{o}(\mathbf{v}(f_d)\otimes\mathbf{s}_{o}\otimes\mathbf{a}(\theta_t,\phi_t))|^2}{\mathbf{w}^H_{o}\mathbf{R_u}(\mathbf{s}_{o})\mathbf{w}_{o}}$ where $\mathbf{w}_o,\mathbf{s}_o$ are obtained after optimization. Using the estimated covariance matrix, say the sample covariance matrix, the definition of the estimated SINR is $\tfrac{\rho_t^2|\mathbf{w}^H_{est}(\mathbf{v}(f_d)\otimes\mathbf{s}_{est}\otimes\mathbf{a}(\theta_t,\phi_t))|^2}{\mathbf{w}^H_{est}\hat{\mathbf{R}}_{\bf u}(\mathbf{s}_{est})\mathbf{w}_{est}}$, where $\hat{\mathbf{R}}_{\bf u}(\cdot)$ is the estimated sample covariance matrix, and $\mathbf{w}_{est},\mathbf{s}_{est}$ are the optimized weight and waveform vectors by using the estimated covariance in the optimization instead.

A true SINR loss can be computed  by using the estimated i.e. $\hat{\mathbf{R}}_{\boldsymbol{\gamma}}^{pq}$  in \eqref{cluteq}  and running the optimization algorithm for each Monte Carlo trial, resulting in an estimated $\mathbf{s}_{est}$. This is computationally heavy  on our current resources, therefore not reported here.
However, we will assume that an oracle has provided the optimal waveform to be transmitted. Then the oracle loss of SINR due to the estimated covariance  is a  random variable, captured by,
\begin{align*}
SINR_{\mathrm{loss}}=\frac{\mathbf{w}^H_{o}\mathbf{R_u}(\mathbf{s}_{o})\mathbf{w}_{o}}{\mathbf{w}^H_{est}\hat{\mathbf{R}}_{\bf u}(\mathbf{s}_{o})\mathbf{w}_{est}} .
\end{align*}
Random data is now generated from zero mean multivariate complex Gaussian distributions to compute the sample covariance matrices, i.e. $\hat{\mathbf{R}}_{\bf i}, \hat{\mathbf{R}}_{\bf n}$ and $\hat{\mathbf{R}}_{\boldsymbol{\gamma}}^{pq}$. Two hundred Monte Carlo trials were run with differing sample supports. The mean and standard deviation of the oracle  $SINR_{\mathrm{loss}}$ are shown in Fig.~\ref{fig8}(a)(b). Not surprisingly the RMB rule is followed perfectly. For the same sample support, the standard deviation is  a few orders less than the mean. 
\subsection{Adapted patterns}
The adapted pattern for the waveform dependent STAP objective function is expressed as
\begin{align}
\mathcal{P}(f_d,\theta)=\lvert\mathbf{w}_o^H ( \mathbf{v}(f_d)\otimes\mathbf{s}_o \otimes \mathbf{a}(\theta,\phi) ) \rvert^2,\;\mbox{ for a fixed }\phi. 
\label{eq60}
\end{align}
The adapted pattern in \eqref{eq60} is a function of angle, Doppler, the optimal weight and the waveform vectors, $\mathbf{w}_o,\mathbf{s}_o$, respectively. Two examples are shown in Fig.~\ref{fig9}(a)(b). Two interferers at $(\theta=-0.2,\phi=\pi/3)$ and at $(-0.2,\pi/3)$ were chosen. We modeled the clutter discretely from all azimuth angles from $-\pi/2\mbox{ to } \pi/2$ in discrete increments of $-0.005\pi/2$ radians. The clutter patches were fixed at an elevation angle of $\pi/4$ radians. The target was assumed to be at $\theta_t=0.7,\phi_t=\pi/4$ with normalized Doppler equal to 0.31 and $\theta_t=0,\phi_t=\pi/4$ with normalized Doppler equal to -0.4 in Fig\ref{fig9}(a)(b), respectively.  The adapted patterns in Fig.~\ref{fig9} are identical (upto a scaling) to those obtained from the classical STAP adapted pattern. This is not a surprise but is rather reassuring since the waveform in \eqref{eq60} affects all the Doppler frequencies and the azimuths identically. Moreover, we can always consider $\mathbf{s}_o \otimes \mathbf{a}(\theta,\phi)$ as a new /modified spatial steering vector. Hence as expected the inclusion of the optimal waveform will not alter the shape of the classical STAP adapted pattern.
\subsection{Detection}
Here, we investigate the impact of detection using the optimized waveforms and randomly selected waveforms. The detection test for the presence of a target at a particular range cell is cast as a binary hypothesis test,
\begin{equation}\label{eq61}
\begin{aligned}
&\mathcal{H}_0 \,\, : \mathbf{w}^H\bar{\mathbf{y}}=\mathbf{w}^H\mathbf{y_u} \,\,\,
&\mathcal{H}_1\,\, : \mathbf{w}^H\bar{\mathbf{y}}=\mathbf{w}^H\mathbf{y}+\mathbf{w}^H\mathbf{y_u}
\end{aligned}
\end{equation}
where $\mathbf{y},$ $\mathbf{y_u}$ have been been defined in \eqref{eq9}, \eqref{eq10}. Assuming that $\mathbf{y_u}$ is complex normal distributed, the test in \eqref{eq61} is readily evaluated. The weight vector is obtained after the optimization. The ROC curves for SINRs 0dB, 3dB and 6dB are shown in Fig.~\ref{fig10}(a)(b) for the non const. modulus and const. modulus design, respectively. For generating Fig.~\ref{fig10}(a),  a random waveform was used having the same energy as that obtained after the alternating minimization algorithm. The waveform samples were drawn independently from a complex Gaussian distribution. In Fig.~\ref{fig10}(b), a chirp waveform was used having the same bandwidth and energy as its optimized constant modulus counterpart. From these figures and as expected, from a  detection standpoint, an optimized waveform performs much better than transmitting an un-optimized waveform. 
\subsection{Realistic STAP waveform design}
We consider  a scenario frequently encountered in STAP, the sample covariance matrix is rank deficient due to the paucity of training data. The simulation parameters are identical to those used as in Fig.~\ref{fig4}, except that we considered ground clutter from all azimuths in $[-\pi/2,\pi.2]$, similar to those used in generating Fig.~\ref{fig9}. Furthermore, we constrained the rank of the resulting correlation matrices to be 30, equal to the numerical rank of the clutter correlation matrix for generating Fig.~\ref{fig11}. The alternating minimization is first used for 20 iterations assuming an arbitrary diagonal loading factor equal to 100. After termination of this algorithm, the proximal algorithm was employed for 50 iterations. The results are shown in Fig.~\ref{fig11}(a)(b). It is noted that in practice the `true' min. eigenvector cannot be computed due to the rank deficiency. Interestingly nonetheless, the designed waveforms after the proximal optimization result in a STAP objective value which is close to that obtained from the waveform estimated from the 'true' min. eigenvector. However, extensive simulations for the rank deficient STAP are needed to verify if this behavior is seen for other classes of noise plus interference, and clutter correlation matrices. 

\section{Conclusions}
Waveform design in STAP was the focus of this report assuming the dependence of the clutter response on the transmitted waveform. Our preliminary simulations indicate that the objective function was jointly non-convex in the weight and waveform vectors. However, we showed analytically that the objective function is individually convex in the waveform and the weight vector. This motivated a constrained alternating minimization technique which iteratively optimizes one vector while keeping the other fixed. 
A constrained proximal alternating minimization technique was propose to handle rank deficient STAP correlation matrices. To addresses practical design constraints we incorporated constant modulus constraints in our alternating minimization formulation. Simulations were chosen to demonstrate the monotonic decrease  of the MVDR objective function using this alternating minimization algorithm. Preliminary simulations were presented to validate the theory.

\appendices




\section*{Acknowledgment}
This work was sponsored by US AFOSR under project 13RY10COR. All views and opinions expressed here are the authors own and does not constitute endorsement from the Department of Defense or the USAF. 

%
%

\ifCLASSOPTIONcaptionsoff
  \newpage
\fi



%


\bibliographystyle{IEEEtran}
\bibliography{refs}


%




\end{document}